\documentclass{article}
\usepackage[utf8]{inputenc}
\usepackage{amsmath,amssymb,amsfonts,amsthm,amstext}
\usepackage{xcolor}
\usepackage{mathtools}
\usepackage{bbm}

\usepackage{array}
\usepackage[margin=1in]{geometry}
\usepackage{float}
\usepackage{braket}
\usepackage[bottom]{footmisc}
\usepackage{booktabs}
\usepackage{makecell}
\usepackage{multirow}
\usepackage[shortlabels]{enumitem}
\usepackage[normalem]{ulem}
\newcolumntype{P}[1]{>{\centering\arraybackslash}p{#1}}

\usepackage[pagebackref]{hyperref}
\hypersetup{
    pdftitle={}, 
    pdfauthor={}, 
    colorlinks=true, 
    linkcolor=blue, 
    citecolor=blue, 
    urlcolor=blue 
}

\renewcommand{\backref}[1]{}

\renewcommand{\backrefalt}[4]{%
\ifcase #1 %
\or
[p.\ #2]%
\else
[pp.\ #2]%
\fi}

\usepackage{tikz}
\usetikzlibrary{backgrounds,fit,decorations.pathreplacing,calc}
\usetikzlibrary{positioning,shapes.misc}

\usepackage[capitalize,nameinlink]{cleveref}
\crefname{thm}{Theorem}{Theorems}
\crefname{property}{Property}{Properties}
\crefname{lemma}{Lemma}{Lemmas}
\crefname{proposition}{Proposition}{Propositions}
\crefname{defn}{Definition}{Definitions}
\crefname{theorem}{Theorem}{Theorems}
\crefname{conjecture}{Conjecture}{Conjectures}
\crefname{corollary}{Corollary}{Corollaries}
\crefname{claim}{Claim}{Claims}
\crefname{section}{Section}{Sections}
\crefname{appendix}{Appendix}{Appendices}
\crefname{figure}{Fig.}{Figs.}
\crefname{table}{Table}{Tables}

\newlist{propertyenum}{enumerate}{1} 
\setlist[propertyenum]{label=\alph*), ref=(\alph*)}
\crefalias{propertyenumi}{property}

\newtheorem{thm}{Theorem}[section]
\newtheorem{lemma}[thm]{Lemma}
\newtheorem{proposition}[thm]{Proposition}

\newtheorem{claim}[thm]{Claim}
\newtheorem*{conjecture*}{Conjecture}

\newtheorem*{property}{Properties of ENTCFs}

\theoremstyle{definition}
\newtheorem{defn}[thm]{Definition}

\newtheorem*{remark}{Remark}

\newcommand{\amatrix}[2]{\hspace{-3pt} \begin{pmatrix} #1 \\ #2 \end{pmatrix}}

\newcommand{\tr}{\mathrm{Tr}}
\newcommand{\pr}{\mathrm{Pr}}
\newcommand{\C}{\mathbb{C}}

\newcommand{\N}{\mathbb{N}}
\newcommand{\thexp}{\mathrm{th}}
\newcommand*{\medcap}{\mathbin{\scalebox{1.3}{\ensuremath{\cap}}}}
\newcommand*{\medcup}{\mathbin{\scalebox{1.3}{\ensuremath{\cup}}}}

\newcommand{\thetaset}{[2N]\cup\{0,\diamond\}}
\newcommand{\thetasetint}{[2N]\cup\{0\}}

\newcommand{\qtestset}{\{0^{2N},1^{2N},0^N1^N,1^N0^N\}}

\newcommand{\negl}{{\mathrm{negl}(\lambda)}}
\newcommand{\notnegl}{{\mu(\lambda)}}

\newcommand{\hatx}{\hat{x}}
\newcommand{\hatb}{\hat{b}}
\newcommand{\hath}{\hat{h}}
\newcommand{\hatPi}{\hat{\Pi}}

\newcommand{\Gen}{\mathrm{Gen}}
\newcommand{\key}{\mathrm{key}}
\newcommand{\CHK}{\mathrm{CHK}}

\newcommand{\Preimage}{\mathrm{INV}}
\newcommand{\Sigmaset}{\mathrm{\Sigma}}
\newcommand{\poly}{\mathrm{poly}}

\newcommand{\modulo}{\mathrm{mod}}

\newcommand{\sigmax}{\sigma^{X}}
\newcommand{\sigmaz}{\sigma^{Z}}

\newcommand{\tildealpha}{\tilde{\alpha}}
\newcommand{\tildesigma}{\tilde{\sigma}}

\newcommand{\real}{\operatorname{Re}}

\newcommand{\str}{\mathrm{str}}
\usepackage{tikz-cd}

\newcommand{\capprox}{\overset{c}{\approx}}

\newcommand{\csimeq}{\overset{c}{\simeq}}
\newcommand{\bignorm}[1]{\Bigl\lVert #1 \Bigr\rVert}
\newcommand{\normstate}[2]{\|#1\|_{#2}}
\newcommand{\norminfty}[1]{\|#1\|_{\infty}}
\newcommand{\normone}[1]{\|#1\|_{1}}
\newcommand{\ipstate}[3]{\langle #1, #2 \rangle_{#3}}

\newcommand{\anc}{\mathrm{anc}}

\newcommand{\PosH}{\mathrm{Pos}(\mathcal{H})}
\newcommand{\Pos}{\mathrm{Pos}}
\newcommand{\LH}{\mathcal{L}(\mathcal{H})}
\newcommand{\UH}{\mathcal{U}(\mathcal{H})}

\newcommand{\BraKet}[2]{\langle #1 | #2 \rangle}
\newcommand{\ketbra}[2]{|#1\rangle\langle#2|}
\newcommand{\ketbrasame}[1]{|#1\rangle\langle#1|}

\DeclareMathOperator{\Supp}{Supp}
\DeclareMathOperator{\Rank}{Rank}
\DeclareMathOperator{\vc}{vec}

\crefname{pluralequation}{Eqs.}{Eqs.}
\Crefname{pluralequation}{Equations}{Equations}

\newcommand{\1}{\mathbbm{1}}

\newcommand{\calA}{\mathcal{A}}

\newcommand{\calD}{\mathcal{D}}

\newcommand{\calF}{\mathcal{F}}
\newcommand{\calG}{\mathcal{G}}
\newcommand{\calH}{\mathcal{H}}
\newcommand{\calK}{\mathcal{K}}
\newcommand{\calL}{\mathcal{L}}
\newcommand{\calP}{\mathcal{P}}
\newcommand{\calV}{\mathcal{V}}
\newcommand{\calX}{\mathcal{X}}
\newcommand{\calY}{\mathcal{Y}}
\newcommand{\calZ}{\mathcal{Z}}

\newcommand{\calVdag}{\mathcal{V}^{\dagger}}
\newcommand{\tildeD}{\widetilde{D}}
\newcommand{\tildeS}{\widetilde{S}}
\newcommand{\tildeZ}{\widetilde{Z}}
\newcommand{\tildeX}{\widetilde{X}}
\newcommand{\hatX}{\hat{X}}

\newcommand{\tildepsi}{\widetilde{\psi}}

\newcommand{\zetadiamond}{\zeta_{\diamond}}
\newcommand{\chidiamond}{\chi_{\diamond}}

\newcommand{\abar}{\overline{a}}

\newcommand{\gammaP}{\gamma_{P}}
\newcommand{\gammaH}{\gamma_{H}}
\newcommand{\gammaHq}{\gamma_{H,q}}

\newcommand{\fail}{\epsilon}
\newcommand{\failP}{\epsilon_P}
\newcommand{\failH}{\epsilon_H}
\newcommand{\failHq}{\epsilon_{H,q}}

\DeclarePairedDelimiter{\abs}{\lvert}{\rvert}
\DeclarePairedDelimiter{\norm}{\lVert}{\rVert}

\DeclareMathOperator*{\E}{\mathbb{E}}

\numberwithin{equation}{section}

\title{Parallel self-testing of EPR pairs under computational assumptions}

\author{
Honghao Fu\thanks{Massachusetts Institute of Technology. \texttt{honghaof@mit.edu} }
\and 
Daochen Wang\thanks{University of Maryland. \texttt{wdaochen@gmail.com}}  
\and 
Qi Zhao\thanks{University of Maryland and The University of Hong Kong. \texttt{zhaoqi@cs.hku.hk}}
}
\date{}

\begin{document}

\maketitle

\begin{abstract}
    Self-testing is a fundamental feature of quantum mechanics that allows a classical verifier to force untrusted quantum devices to prepare certain states and perform certain measurements on them. The standard approach assumes at least two spatially separated devices. Recently, Metger and Vidick~\cite{mv2021selftest} showed that a single EPR pair of a single quantum device can be self-tested under \emph{computational assumptions}. In this work, we generalize their results to give the first parallel self-test of $N$ EPR pairs and measurements on them in the single-device setting under the same computational assumptions. We show that our protocol can be passed with probability negligibly close to $1$ by an honest quantum device using $\poly(N)$ resources. Moreover, we show that any quantum device that fails our protocol with probability at most $\fail$ must be $\poly(N,\fail)$-close to being honest in the appropriate sense. In particular, our protocol can test any distribution over tensor products of computational or Hadamard basis measurements, making it suitable for applications such as device-independent quantum key distribution \cite{metger2021diqkd} under computational assumptions. Moreover, a simplified version of our protocol is the first that can efficiently certify an arbitrary number of qubits of a single cloud quantum computer using only classical communication.
\end{abstract}

\tableofcontents

\section{Introduction}

Self-testing is a fundamental feature of quantum mechanics that allows a classical verifier to force a quantum device (sometimes called prover) to prepare certain states and measure them in certain bases up to local isometries.
The term ``self-test'' was first coined by Mayers and Yao~\cite{mayersyao2004selftest} in 2004, but its concept can be traced back to much earlier works that study the remarkable features of quantum
correlations~\cite{bell1964epr,summerswerner1987bell,tsirelson1987bell,popescu1992bell,bmr1992bell}.

To perform a self-test, the verifier inputs questions $x$ to the devices and they return answers $a$. The key idea of self-testing is that if $a$ and $x$ obey certain \emph{nonlocal} correlations, then the verifier can deduce the devices'  behavior, assuming they are spatially separated. This assumption, which implies non-communication between the devices, is crucial because otherwise, the devices could reproduce any correlation by using a lookup table. 

The literature on this nonlocal type of self-test is vast~\cite{supic2020selftest}. They address topics such as which correlations can self-test which states, e.g., \cite{coladangelo2017bipartite,goh2018geometry,miller2019ghz}; how efficient and robust a self-test can be, e.g., \cite{mys2012robust,mckague2017chsh,nv2017quantumlinearity,nv2018lowdegree,chao2018largeentanglement,fu2019}; and how to use self-testing to, e.g., certify a quantum computer's components~\cite{sekatski2018certifying}, delegate quantum computations~\cite{ruv2013command,coladangelo2019leash}, and characterize the complexity of quantum correlations~\cite{ji2020mipre}.

A limitation of nonlocal self-testing is the assumption of spatial separation. In practice, it is difficult
to certify this assumption, especially if the device is compact or falls outside our physical control. Therefore, it is interesting to ask whether we can replace this assumption with another one so that we can self-test a \emph{single} quantum device. We illustrate the nonlocal and single-device settings in \cref{fig:nonlocal_single}.

\newcommand{\Cross}{$\mathbin{\tikz [x=2ex,y=2ex,line width=.2ex] \draw (0,0) -- (1,1) (0,1) -- (1,0);}$}
\begin{figure}[H]
\center
\resizebox{340pt}{!}{
        \begin{tikzpicture}[thick]
        \tikzstyle{operator} = [draw,rounded rectangle,fill=white,minimum size=2em] 
        \tikzstyle{device} = [draw,rounded rectangle,fill=white,minimum size=4.5em] 
        \tikzstyle{phase} = [fill,shape=circle,minimum size=5pt,inner sep=0pt]
        \tikzstyle{surround} = [fill=white,thick,draw=black,rounded corners=2mm]
        \tikzset{edge/.style = {->,> = latex'}}
        
        \node[operator] (v1) at 
        (0, -1) {Verifier};
        
        \node[operator] (p1) at
        (3, 0) {Device 1};
        
        \node[operator] (p2) at 
        (3, -2) {Device 2};
        \path [->] 
        ([yshift=2mm]v1.east) 
        edge node [above=0.5mm] {x} 
        ([yshift=2mm]p1.west);
        \path [->] 
        ([yshift=0.1mm]p1.west) 
        edge node [below=0.5mm] {a} 
        ([yshift=0.1mm]v1.east);
        \path [->] 
        ([yshift=-0.1mm]v1.east) 
        edge node [above=0.5mm] {y} 
        ([yshift=-0.1mm]p2.west);
        \path [->] 
        ([yshift=-2mm]p2.west) 
        edge node [below=0.5mm] {b} 
        ([yshift=-2mm]v1.east);
        
        \node at (3,-1)  {\Cross};
        \draw[red,ultra thick] (4.5,0) -- (4.5,-2);
        \path [->] 
        ([xshift=-1mm]p1.south) 
        edge node [] {} 
        ([xshift=-1mm]p2.north);
        \path [->] 
        ([xshift=1mm]p2.north) 
        edge node [] {} 
        ([xshift=1mm]p1.south);
        
        \node[device] (v2) at 
        (6, -1) {Verifier};
        
        \node[device] (p3) at 
        (9.5, -1) {Device};
        
        \path [->] 
        ([yshift=6mm]v2.east) 
        edge node [above=0mm] {$x_1$} 
        ([yshift=6mm]p3.west);
        \path [->] 
        ([yshift=1mm]p3.west) 
        edge node [above=0mm] {$a_1$} 
        ([yshift=1mm]v2.east);
        \path [->] 
        ([yshift=-1mm]v2.east) 
        edge node [below=0.1mm] {$x_2$} 
        ([yshift=-1mm]p3.west);
        \path [->] 
        ([yshift=-6mm]p3.west) 
        edge node [below=0.1mm] {$a_2$} 
        ([yshift=-6mm]v2.east);
        \end{tikzpicture}}
	\caption{Self-testing in the nonlocal setting (left)  involves (at least) two spatially separated devices that cannot communicate. In the single-device setting (right), there is only one device.}
  \label{fig:nonlocal_single}
\end{figure}

\paragraph{Computational self-testing.} 

Recently, beginning with seminal work by Mahadev~\cite{mahadev2018verification} on the classical verification of quantum computations, a series of works, e.g., \cite{gheorghiuvidick2019rsp,
brakerski2020simpler,
chia2020noninteractive,alagic2020noninteractive, vidickzhang2020classicalzero, vidickzhang2021classicalproofs,kahanamokumeyer2021classicallyverifiable,brakerski2021cryptographic,
hirahara2021lowdepth,
liu2021depthefficient,
zhu2021demonstration, mv2021selftest, metger2021diqkd, mizutani2021computational}, have explored how computational assumptions can be leveraged by a classical verifier to control a single quantum device in certain ways. Typically, the assumption used is that the Learning-With-Errors (LWE)~\cite{regev2009lwe} problem is hard to solve efficiently, even for quantum computers, which is a standard assumption. However, except for \cite{gheorghiuvidick2019rsp,vidickzhang2021classicalproofs,mv2021selftest,metger2021diqkd, mizutani2021computational}, the level of control established in these works is much weaker than in nonlocal self-testing. For example, if a device passes Mahadev's verification protocol~\cite{mahadev2018verification}, it only means that, to quote \cite{mv2021selftest}, ``\textit{there exists} a quantum state such that the distribution over the prover's answers \textit{could have been} produced by performing the requested measurements on this state''. We do not know whether the prover \textit{actually} prepared that state and performed the requested measurements on it.

Metger and Vidick~\cite{mv2021selftest} are the first to explicitly propose the self-testing of a single device under computational assumptions. We interchangeably refer to this as computational or single-device self-testing. In a sense, their work is a culmination of many previous results because self-testing offers the strongest control. The main limitation of \cite{mv2021selftest} and follow-up work \cite{mizutani2021computational} is that they only self-test two and three qubits, respectively. In this work, we introduce a protocol that self-tests $N$ EPR pairs and measurements on them in the computational setting. Our protocol can be passed with probability negligibly close to $1$ by an honest quantum device using $\poly(N)$ resources. On the other hand, we show that any quantum device accepted by our protocol with probability $\geq 1-\epsilon$ must be $\poly(N,\epsilon)$-close\footnote{By $\poly(N,\epsilon)$, we mean a real function of $N$ and $\epsilon$ of order $O(N^a\epsilon^b)$ as $N\rightarrow \infty$ and $\epsilon\rightarrow 0$, where $a,b>0$ are constants.} to being honest in the appropriate sense. 

\paragraph{Main results.}

 We give a protocol for testing the following states and measurements.
 \begin{equation}\label{eq:states}
 \begin{aligned}
    \hspace{12pt} \textbf{States.} \;  &\Bigl\{ \ket{\tau^{\theta,v}} \coloneqq \ket{v_1} \otimes \cdots \otimes \ket{v_{\theta-1}} \otimes \ket{(-)^{v_{\theta}}} \otimes \ket{v_{\theta+1}} \otimes \cdots \otimes\ket{v_{2N}} \bigm| \theta\in \{1,\ldots,2N\}, \, v\in \{0,1\}^{2N} \Bigr\} \\
    &\bigcup \, \Bigl \{ \ket{\tau^{0,v}} \coloneqq \ket{v_1} \otimes \cdots \otimes\ket{v_{2N}} \bigm| v\in \{0,1\}^{2N} \Bigr\} 
    \\
    &\bigcup \, \Bigl\{ \ket{\tau^{\diamond,v}}\coloneqq\frac{1}{\sqrt{2^N}}\bigotimes_{i=1}^N (\sigmax)^{v_i}\otimes (\sigmax)^{v_{N+i}}(\ket{0}_i\ket{+}_{N+i} + \ket{1}_i\ket{-}_{N+i})\bigm| v \in \{0,1\}^{2N}\Bigr\},
 \end{aligned}
 \end{equation}
where $\ket{(-)^{a}} \coloneqq (\ket{0} + (-1)^{a} \ket{1})/\sqrt{2}$ for $a \in \{0,1\}$.

\medskip
\textbf{Measurements.} Any distribution over tensor products of computational (Pauli-$Z$) or Hadamard (Pauli-$X$) basis measurements on $2N$ qubits; more precisely, any distribution over the following set of  measurements 
\begin{equation}\label{eq:measurements}
    \Bigl\{ \bigl\{ \Pi_q^u \coloneqq \ketbrasame{B_{q_1}^{u_1}}\otimes \ldots \otimes \ketbrasame{B_{q_{2N}}^{u_{2N}}} \bigm| u \in \{0,1\}^{2N}\bigr \} \Bigm| {q\in \{0,1\}^{2N}} \Bigr\},
\end{equation}
where
\begin{equation}\label{eq:bb84_states}
    \ket{B_0^0} \coloneqq \ket{0}, \quad \ket{B_0^1} \coloneqq \ket{1}, \quad \ket{B_1^0} \coloneqq \ket{+}, \quad \text{and} \quad \ket{B_1^1} \coloneqq \ket{-}.
\end{equation}
\medskip
Note that we can self-test product states and entangled states together with local measurements in the single-device setting, which is not possible in the nonlocal setting.

Our protocol generalizes the protocols in \cite{gheorghiuvidick2019rsp, mv2021selftest} and uses the Extended Noisy Trapdoor Claw-Free function Families, or ENTCFs, from~\cite{mahadev2018verification}. An ENTCF consists of two indistinguishable function-pair families, a claw-free family $\calF$ and an injective family $\calG$, and satisfies various properties under the LWE hardness assumption. 

In our protocol, the classical verifier first samples $\theta \in \{0,1,\ldots,2N\}\cup\{\diamond\}$ uniformly at random.  Then it generates the public keys and trapdoors of $2N$ function pairs from $\calF \cup \calG$ according to $\theta$ as follows.
\begin{enumerate}
    \item $\theta=0$: all pairs are from $\calG$.
    \item $\theta \in \{1,\ldots,2N\}$: the $\theta^{\thexp}$ pair is from $\calF$ and the remaining $2N-1$ pairs are from $\calG$.
    \item $\theta = \diamond$: all pairs are from $\calF$.
\end{enumerate}
The verifier sends the public keys to the device. The device then sends back $2N$ images, $y_1,\ldots,y_{2N}$, of these function pairs -- these play the role of a commitment. In the second round, the verifier either (i) checks the commitment by asking for preimages of the $y_i$ and accepts or rejects accordingly, or (ii) asks for an equation involving the preimages of the $y_i$. In case (ii), there is a final round where the verifier sends with probability $1/2$ a uniformly random $q\in \{0^{2N},1^{2N},0^{N}1^{N},1^N0^N\}$ and with probability $1/2$ a random $q\in \{0,1\}^{2N}$ according to some distribution $\mu$ of its choosing. The device sends back the result $u\in \{0,1\}^{2N}$ of performing some measurement $\{P_q^u\}_u$. The verifier lastly checks that $u$ is consistent with measuring $\ket{\tau^{\theta,v}}$ using $\{\Pi_q^u\}_u$ , where $v\in \{0,1\}^{2N}$ is some bitstring that the verifier can compute efficiently using the trapdoors, and accepts or rejects accordingly.

\begin{thm}[Informal]\label{thm:informal_selftest}
  Let $\lambda\in \mathbb{N}$ be a security parameter and let $N=\poly(\lambda)$. Assuming the LWE problem of size $\lambda$ cannot be solved in $\poly(\lambda)$ time, our protocol satisfies the following properties.
  \begin{description}
    \item[Completeness (see \cref{thm:completeness}).]
    Using $\poly(\lambda)$ qubits and quantum gates, a quantum device can prepare one of the $2N$-qubit states in $\{\ket{\tau^{\theta,v}} \mid \theta\in \{0,1,\ldots,2N\}\cup\{\diamond\}, v\in \{0,1\}^{2N}\}$ and measure it using $\{\Pi_q^u \mid u \in \{0,1\}^{2N}\}$ upon question $q\in \{0,1\}^{2N}$ to pass our protocol with probability $\geq 1-\negl$. Moreover, the verifier can be classical and run in $\poly(\lambda)$ time.
    \item[Soundness (see \cref{thm:soundness}).]
     If a quantum device passes our protocol in $\poly(\lambda)$ time with probability $\geq 1 - \epsilon$,
    then the device must have prepared a (sub-normalized) state $\sigma^{\theta,v}$, measured it using $\{P_q^u\}_u$, and received outcome $u$, such that
    \begin{align}
        &\sum_{v \in \{0,1\}^{2N}}\norm{V \sigma^{\theta,v} V^\dagger - 
        \ketbra{\tau^{\theta,v}}{\tau^{\theta,v}} \otimes \alpha^{\theta,v}}_1 \leq O(N^{7/4}\epsilon^{1/32}) \quad \text{and}\label{eq:intro_soundness_states}
        \\
        &\mathbb{E}_{q\leftarrow \mu} \Bigl[ \sum_{u,v \in \{0,1\}^{2N}}\norm{V P_q^u \, \sigma^{\theta,v} \, P_q^u V^\dagger - 
        \Pi_q^{u} \, \ketbra{\tau^{\theta,v}}{\tau^{\theta,v}} \, \Pi_q^{u} \otimes \alpha^{\theta,v}}_1 \Bigr] \leq O(N^2\epsilon^{1/32}),\label{eq:intro_soundness_measurements}
    \end{align}
   where $\theta\in \{0,1,\ldots,2N\}\cup\{\diamond\}$, $\mu$ is the distribution on $\{0,1\}^{2N}$ chosen by the verifier in our protocol, $u,v\in \{0,1\}^{2N}$ are known to the verifier, $V$ is an efficient isometry independent of $\{\theta, \mu, u, v\}$, and the $\alpha^{\theta,v}$s are some auxiliary states that are computationally indistinguishable from some fixed state $\alpha$.
  \end{description}
\end{thm}

We highlight the $\poly(N,\epsilon)$ soundness error (or robustness) that we achieve. Good robustness is critical if we want to use our protocol in practice because real quantum devices are imperfect. The more imperfect a device is, the more robust a protocol needs to be to control it.

\paragraph{Techniques.}

The completeness of our self-testing protocol follows straightforwardly from the properties of ENTCFs (see \cref{sec:completeness}). The main challenge is to prove  soundness. We give a high-level overview here and provide more details in \cref{sec:soundness}. We start by defining $4N$ observables of the device $\{X_i, Z_i \mid i \in [2N]\}$ using its measurement operators. The strategy is to characterize these observables as the standard $\sigmax_i$ and $\sigmaz_i$ Pauli observables on $2N$ qubits where $i$ indexes those qubits. Then, we characterize the device's states by their invariance under products of projectors corresponding to these observables and the device's measurements as products of these projectors. To characterize $X_i$ and $Z_j$, we first generalize techniques in \cite{mv2021selftest} to show that $X_i$ and $Z_j$ obey certain state-dependent commutation and anti-commutation relations (\cref{sec:commute,sec:anti_commute}). To carry out the generalization, it is important for the verifier to select $\theta$ from the set $\thetaset$ for two reasons. The first is that they allow us to bound the failure probability associated with \emph{each} $\sigma^\theta$ by $2N+2$ (the number of possible $\theta$s) times the \emph{average} failure probability over all $\theta$s. The second is that this restricted set of $\theta$s suffices for us to characterize $X_i$ and $Z_i$ as $\sigmax_i$ and $\sigmaz_i$. Intuitively, $\theta = 0$ is used to characterize $\{Z_1, \ldots, Z_{2N}\}$, $\theta \in [2N]$ is used to characterize $X_\theta$, and $\theta=\diamond$ is used to characterize EPR pairs. We give a more precise correspondence in \cref{tab:testing_correspondence}.

From \cref{sec:conj_inv} onwards, we introduce new techniques to handle products of projectors (corresponding to these observables). These techniques differ significantly from \cite{mv2021selftest} because their techniques are not susceptible to generalization to arbitrary $N$. These techniques also differ significantly from those used in nonlocal self-testing because we lack the perfect state-\emph{independent} commutation relations between observables on two spatially-separated devices. More specifically, we introduce a ``operator-state commutation'' lemma (\cref{prop:operator_state_commutation}) and a ``lifting'' lemma (part $5$ of \cref{lem:lifting}, and \cref{lem:lifting_proj}) that together give us the ability to ``commute an observable past a state'' by leveraging the cryptographic properties of ENTCFs based on the LWE hardness assumption. We then use this ability to handle products of projectors (see \cref{sec:swap_states,sec:swap_observables,sec:soundness_states_meas}). This ability is useful because, for example, $X_1Z_2X_3\psi = Z_2X_1X_3\psi$ (1) does not follow from the commutation relation $X_1Z_2\psi = Z_2 X_1 \psi$, where $\psi$ is some density operator; however, it does follow if we could commute $X_3$ past $\psi$ first because $X_1$ and $Z_2$ would then be directly next to $\psi$. Having all (1)-like relations implies that $X_i$ and $Z_i$ can be characterized as $\sigmax_i$ and $\sigmaz_i$ respectively, which follows from results in approximate representation theory~\cite{Vidick_Fsmp_2021,gowershatami2017rep}. We remark that the preceding discussion is for intuition only: in fact, we do not explicitly prove (1)-like relations and then apply representation theory. Instead, we do these two steps simultaneously to directly show that an explicit ``swap'' isometry (defined by \cref{fig:swap_isometry} in \cref{sec:swap}) approximately maps $X_i$ and $Z_i$ to $\sigmax_i$ and $\sigmaz_i$ respectively.

\paragraph{Applications.}

We present two applications of our result, the first is for device-independent (DI) quantum key distribution (QKD), and the second is for dimension testing. We stress that for both applications, we crucially rely on the characterization of \emph{measurements} in \cref{eq:intro_soundness_measurements} of \cref{thm:informal_selftest}. Our characterization of measurements significantly differentiates our work from that of \cite{gmp2022parallelrsp}, which only characterizes states.

\emph{DIQKD \textup{(\cref{sec:diqkd})}.} 
A DI protocol is one where the parties involved do not need to trust the inner working of the devices they use to be sure that the devices have successfully implemented the protocol. A QKD protocol is one for establishing information-theoretically secure keys between two parties. Previous DIQKD protocols rely on the nonlocal assumption. This assumption is usually justified experimentally by spatially separating two devices by a large distance, which is difficult to implement.  Recently, Metger et.~al.~\cite{metger2021diqkd} proposed a different setting for DIQKD: they replace the nonlocal assumption with the assumption that the two devices are computationally bounded. It is still assumed that the devices cannot communicate with an eavesdropper because otherwise the eavesdropper can learn the generated shared secret key. (We refer the reader to \cite[Section 4]{metger2021diqkd} for a discussion on why this assumption might still hold even when the devices themselves can communicate.) However, since their protocol sequentially repeats the self-test in \cite{mv2021selftest}, their soundness proof
relies on the IID assumption that the device behaves identically and independently at each repetition to argue that it has prepared and measured many EPR pairs.

Our DIQKD protocol consists of a random number of ``test rounds'' followed by a final ``generation round'', where both round types are based on our self-test. The $N$ EPR pairs certified in the generation round are used to generate $\Omega(N)$ shared keys. Because of the parallel nature of our self-test, our DIQKD protocol does not require the IID assumption. We sketch a soundness proof that uses a ``cut-and-choose'' argument from \cite[Theorem 4.33 (arXiv version v2)]{gmp2022parallelrsp} to upper bound the failure probability of the device in the generation round, conditioned on the protocol not aborting in the test rounds. This argument does not require an IID assumption \emph{between} rounds. Then, we use \cref{eq:intro_soundness_measurements} of \cref{thm:informal_selftest} to lower bound the key rate, which does not require an IID assumption \emph{within} any round. Hence we remove the IID assumption altogether. The application of our self-test to remove the IID assumption from DIQKD in the computational setting can be viewed as analogous to the application of a nonlocal self-test to remove the IID assumption from DIQKD in the usual nonlocal setting~\cite{ruv2013command}.

\emph{Dimension test \textup{(\cref{sec:dimension})}.} 
Our dimension test is a simplified version of our self-test and is inspired by the nonlocal dimension test in \cite{chao2020quantum} and its exposition in \cite[Section~2.5.2]{Vidick_Fsmp_2021}. The protocol in \cite{chao2020quantum} works as follows. The verifier chooses a random bit $\theta\in \{0,1\}$ and random bitstring $x\in \{0,1\}^n$ and sends $n$ qubits to the device such that the qubits encode $x$ in the computational basis ($\theta=0$) or in the Hadamard basis ($\theta=1$). After the device has received all $n$ qubits, the verifier sends $\theta$ to the device and asks it to return a bitstring $x'\in \{0,1\}^n$. If $x'=x$, the verifier certifies that the device has a large quantum dimension. Our protocol can be viewed as a version of this protocol, where the verifier classically delegates the preparation of the appropriate $n$-qubit states to the prover in a secure manner. Although our protocol is inspired by \cite{chao2020quantum}, our security proof uses \cref{thm:informal_selftest} and differs significantly from that in \cite{chao2020quantum}.

\begin{thm}[Informal, see \cref{cor:dimension}]\label{thm:informal_dimension}
Under the same computational assumptions as in \cref{thm:informal_selftest}, if a quantum device runs in $\poly(\lambda)$ time and passes our dimension test with probability $\geq 1 - \epsilon$, then its quantum dimension is at least
$(1 - O(N^2\epsilon^{1/32}))2^{N}$.
\end{thm}

Note that by quantum dimension, we mean the dimension of the device's quantum memory. We take the base-$2$ logarithm of the quantum dimension of a device as a count of its number of qubits. Therefore, the theorem also lower bounds the number of qubits by $N - O(N^2\epsilon^{1/32})$. Importantly, $O(N^2\epsilon^{1/32}) = \poly(N,\epsilon)$, which means that it suffices to run our protocol $\poly(N)$ times to estimate $\epsilon$ to an accuracy that is sufficient to ensure an $\Omega(N)$ lower bound on the qubit-count. As a single run of our protocol also only takes $\poly(N)$ time, the total time to implement a dimension test is $\poly(N)$, which is theoretically efficient. Intuitively, this theorem is proved using \cref{eq:intro_soundness_measurements} of \cref{thm:informal_selftest} to argue that the Hilbert space $\calH$ of the device must be able to accommodate all possible post-measurement states that could result from performing a Hadamard basis measurement of $N$ qubits in a computational basis state. Since there are $2^N$ such post-measurement states, and they are all orthogonal, we deduce a quantum dimension lower bound of $2^N$. A formal proof is more challenging because  \cref{eq:intro_soundness_measurements} of \cref{thm:informal_selftest} gives an approximation and we need to prove that the rank of a quantum state is robust against the approximation error.

Compared to nonlocal dimension tests \cite{brunner2008testing,cai2016new,coladangelo2020two}, the advantage of ours is that we do not need to assume spatial separation between multiple devices. Compared to prepare-and-measure dimension tests \cite{gallego2010device,chao2017overlapping,chao2018largeentanglement,chao2020quantum},
the advantage of ours is that the verifier does not need to be quantum -- all computations and communications are classical.
To the best of our knowledge, our dimension test is the first\footnote{More recently, \cite{mahadev2022efficient} also claims a dimension test using completely different methods.} that can test for an arbitrary quantum dimension in the computational setting. In fact, whether this is possible was recently raised as an open question by Vidick in \cite[pg.~84]{Vidick_Fsmp_2021}.

\paragraph{Discussion.}

One interesting direction is to further improve the efficiency and robustness of our protocol. When $N = \lambda$, one bottleneck in improving the efficiency is that sending (the public key of) one function pair already requires $\poly(\lambda) = \poly(N)$ bits of communication. In recent work, it has been shown that, instead of sending the public keys, the verifier can apply a \emph{succinct batch key generation algorithm} to reduce the cost of sending public keys \cite{bartusek2022succinct}. We expect that techniques in \cite{bartusek2022succinct}
can be used to shorten other messages of our protocol as well.
Turning to robustness, we note that there exists a nonlocal self-test \cite{nv2017quantumlinearity} which uses $\poly(N)$ bits of communication and achieves robustness $\poly(\epsilon)$. It might be possible to combine our techniques with those in \cite{nv2017quantumlinearity} to achieve similar robustness in the computational setting. Another interesting question to ask is what $\mathsf{MIP}^\ast$ protocols can be compiled into computation delegation protocols under computational assumptions. For comparison, it has been shown that classical $\mathsf{MIP}$ protocols sound against non-signalling provers can be turned into computation delegation protocols \cite{tauman2013,tauman2014}. It would also be interesting to see if a systematic way exists to translate nonlocal self-tests into computational ones. We note that \cite{kahanamokumeyer2021classicallyverifiable} suggests that the two settings might not be too different at a conceptual level by presenting a test of quantumness in the computational setting that closely resembles the nonlocal CHSH test~\cite{chsh}. 
Recently, Kalai et.~al. proposed a way to 
construct a proof-of-quantumness protocol from any nonlocal game with a classical and quantum separation
using quantum homomorphic encryption \cite{kalai2022quantum}.
However, it is unknown if the aforementioned protocols are quantumly sound.
Going beyond quantum dimension testing, it would be interesting to see if our protocol can be combined with those that test quantum circuit depth \cite{chia2022classical,arora2023quantum} to give a protocol that tests the quantum volume of a quantum computer.

\paragraph{Organization.} Our paper is organized as follows. In \cref{sec:prelims}, we give the preliminaries required. More specifically, we review ENTCFs and prove a variety of approximation lemmas. In \cref{sec:completeness}, we describe our protocol and prove that it can be passed with probability negligibly close to $1$ by an honest quantum device using $\poly(N)$ resources (\cref{thm:completeness}). In \cref{sec:soundness}, we prove that the soundness error of our protocol can be controlled to within $\poly(N,\fail)$ (\cref{thm:soundness}). In \cref{sec:diqkd}, we present a DIQKD protocol based on our self-testing protocol. In \cref{sec:dimension}, we present a simplified version of our self-testing protocol that can efficiently certify an arbitrary quantum dimension (\cref{cor:dimension}).

\paragraph{Acknowledgments.}
We especially thank Carl Miller, Tony Metger, and Thomas Vidick for many helpful discussions and correspondence. We also thank Anne Broadbent, Nai-Hui Chia, Shih-Han Hung, Yi Lee, Atul Mantri, Peter Yuen, and Jiayu Zhang for helpful discussions. HF acknowledges the support of the NSF QLCI program (grant OMA-2016245). DW acknowledges the support of the Army Research Office (grant W911NF-20-1-0015); the Department of Energy, Office of Science, Office of Advanced Scientific Computing Research, Accelerated Research in Quantum Computing program; and the National Science Foundation (grant DMR-1747426). QZ acknowledges the support of the Department of Defense through the QuICS Hartree Postdoctoral Fellowship. 

\paragraph{Note.}
During the preparation of the first arXiv version of this manuscript, we became aware of related independent work by Alexandru Gheorghiu, Tony Metger, and Alexander Poremba, who construct a protocol for remotely preparing a tensor product of random BB84 states and use this to remove the need for quantum communication from a number of quantum cryptographic protocols. We refer to their paper~\cite{gmp2022parallelrsp} for more details and thank them for their cooperation in publishing the first versions of our respective results at the same time. In the second version of the manuscript, we used ideas from \cite[Proof of Proposition 4.32 (arXiv version v2)]{gmp2022parallelrsp} in our proof of \cref{lem:ci_aux}. We also used ideas from \cite[Proof of Theorem 4.33 (arXiv version v2)]{gmp2022parallelrsp} in our construction of a DIQKD protocol in \cref{sec:diqkd}. The present manuscript is the third arXiv version and improves upon the second arXiv version in its presentation. In particular, we have improved the introduction to \cref{sec:soundness}.

\section{Preliminaries}
\label{sec:prelims}

\subsection{Notation}

$\mathbb{N}$ is the set of positive integers. For $k\in \mathbb{N}$, we write $[k]\coloneqq \{1,2,\ldots,k\}$. Except in \cref{sec:dimension}, we reserve the letter $N$ for the number of EPR pairs we self-test and so $2N$ is the number of qubits we self-test. In \cref{sec:dimension}, we reserve $N$ for the number of qubits we dimension-test. We do not use special fonts for vectors. Unless otherwise indicated, the (lowercase) letters $a,c$ are reserved for single bits, i.e., $a,c \in \{0,1\}$; $u,v$ for bitstrings in $\{0,1\}^{2N}$; $\epsilon,\delta$ for real numbers in $(0,1)$; and $n$ for a positive integer. When we use these reserved symbols as indices of a sum without specifying the range, the range should be taken as the entire domain of these symbols. For example, $\sum_a$ always means $\sum_{a\in \{0,1\}}$. For a set $X$ and a condition $C$ on elements of that set, we use the notation $\sum_{x \in X \mid C}$ to mean a sum over all $x \in X$ that satisfy condition $C$. The set $X$ can be implicit, so, for example, $\sum_{v | v_1=a}$ means a sum over all $v\in \{0,1\}^{2N}$ with $v_1=a$. For a finite set $X$, we use the notation $x\leftarrow_U X$ to mean that $x$ is sampled from $X$ uniformly at random. For a probability distribution $\mu$ on $X$, we use the notation $x\leftarrow_\mu X$ to mean that $x$ is sampled from $X$ according to $\mu$.
 
$\calH$ denotes a finite-dimensional Hilbert space. $\LH$ denotes the set of linear operators on $\calH$ and $\UH$ denotes the set of unitary operators on $\calH$. $\PosH$ denotes the positive semi-definite operators on $\calH$, $\PosH \coloneqq \{A\in \LH \mid A \geq 0 \}$. We sometimes refer to operators in $\PosH$ or vectors in $\calH$, not necessarily normalized, as (quantum) states. For operators $A, B \in \LH$, $A \geq B$ means $A - B \geq 0$, i.e, $A - B \in \PosH$. $\mathcal D(\calH)$ denotes the set of density operators, $\mathcal D(\calH) \coloneqq \{A\in \PosH \mid \tr[A]=1 \}$. All Hilbert spaces in this work are viewed as $\mathbb{C}^{m}$ for some $m\in \mathbb{N}$ under a fixed choice of basis: this is necessary for some notions we use to make sense, for example, quantum gates and the vector-operator correspondence in \cref{sec:dimension}.

We write $\lambda\in \mathbb{N}$ for the security parameter.
Most quantities in this work are dependent on $\lambda$. Therefore, for convenience, we often make the dependence implicit.  
A function $f: \mathbb N\rightarrow \mathbb{R}$ is said to be negligible if for any polynomial $p \in \mathbb{R}[x]$, $\lim_{\lambda\rightarrow \infty} f(\lambda)p(\lambda)=0$. We denote such functions by $\negl$.

For an operator $X \in \LH$, we write $\norm{X}_p \coloneqq \tr [\abs{X}^p]^{1/p}$, where $|X|\coloneqq \sqrt{X^{\dagger}X}$, for the Schatten $p$-norm. In this work, we mainly work with the trace norm $\norm{X}_1$, Frobenius norm $\norm{X}_2$ (also written as $\norm{X}_F$), and operator norm $\norm{X}_{\infty}$. For operators $A,B$ in $\LH$, we use $\langle A, B \rangle \coloneqq \tr[A^{\dagger} B]$ to denote the Hilbert-Schmidt inner product. The commutator and anti-commutator of $A$ and $B$ are defined as $[A,B] \coloneqq AB-BA$ and  $\{A,B\} \coloneqq AB+BA$ respectively. 

The single-qubit $Z$ and $X$ Pauli operators are denoted $\sigmaz\coloneqq\begin{psmallmatrix} 1 & 0 \\ 0 & -1 \end{psmallmatrix}$
and  $\sigmax \coloneqq \begin{psmallmatrix} 0& 1 \\ 1 & 0 \end{psmallmatrix}$ which have eigenstates $\{\ket{0} \coloneqq \begin{psmallmatrix} 1 \\ 0 \end{psmallmatrix}, \ket{1} \coloneqq \begin{psmallmatrix} 0 \\ 1 \end{psmallmatrix}\}$ and $\{\ket{(-)^0}\coloneqq \frac{1}{\sqrt{2}}(\ket{0}+\ket{1}),\ket{(-)^1} \coloneqq \frac{1}{\sqrt{2}}(\ket{0}-\ket{1}) \}$, respectively. Given $m\in \mathbb{N}$, and $i\in [m]$, we define $\sigmaz_i$ (resp.~$\sigmax_i$) to be linear operator on $(\mathbb{C}^2)^{\otimes m}$ that acts as $\sigmaz$ (resp.~$\sigmax_i$) on the $i$th tensor factor and identity on all other tensor factors. The value of $m$ should always be clear from the context.

An observable on $\calH$ refers to a Hermitian operator in $\LH$. We say an observable is a binary observable if it has two eigenvalues, $-1$ and $+1$. For a binary observable $O$, we define $O^{(0)}$ (resp. $O^{(1)}$) to be the projector onto the $+1$ (resp. $-1$) eigenspace of $O$. Equivalently, for $b\in\{0,1\}$, we define $O^{(b)}\coloneqq \frac{\1+(-1)^{b}O}{2}$. Note that $O = O^{(0)}-O^{(1)}$.
We say $\{ P^i  \mid i \in [n]\}$ is a projective measurement if $P^i \geq 0$ for all $i\in [n]$,
$P^i P^j = \delta_{i,j} P^i$ for all $i,j\in [n]$, and $\sum_{i=1}^n P^i = \1$.

\subsection{Extended noisy trapdoor claw-free function families}

In this sub-section, we summarize the properties that we employ of Extended Noisy Trapdoor Claw-free function Families (ENTCFs). Our discussion is based on the arXiv v2 version of \cite{mahadev2018verification}. For full details about the properties of ENTCFs, we refer to
\cite[Definitions~4.1--4.4]{mahadev2018verification}. For full details about how to construct ENTCFs under the LWE hardness assumption, we refer to \cite[Section 9]{mahadev2018verification}; in particular, throughout this work, we make the LWE hardness assumption as described in \cite[Definition 3.4]{mahadev2018verification} where the LWE parameters are set according to \cite[Section 9.1]{mahadev2018verification} as functions of the security parameter $\lambda$.

Let $\lambda \in \mathbb{N}$ be a security parameter. Let $\calX\subseteq\{0,1\}^{w}$ and $\calY$ be finite sets that depend on $\lambda$, where $w = w(\lambda)$ is some integer that is a polynomially-bounded function of $\lambda$. An ENTCF consists of two families of function pairs, $\calF$ and $\calG$.
Function pairs from these two families are labeled by public keys. The set of public keys for $\calF$ is denoted by $\calK_\calF$, and the set of public keys for $\calG$ is denoted by $\calK_\calG$. For $k \in \calK_\calF$, a function pair $(f_{k,0}, f_{k,1})$ from $\calF$ is called a \emph{claw-free} pair. For $k \in \calK_\calG$, a function pair
$(f_{k,0}, f_{k,1})$ from $\calG$
is called an \emph{injective} pair. For any $k \in \calK_\calF \cup \calK_\calG$, the functions $f_{k,0}$ and $f_{k,1}$ map an $x \in \calX$ to a probability distribution on $\calY$. Note that the keys and function pairs of an ENTCF are functions of $\lambda$. We use the terms ``efficient'' and ``negligible'' to refer to $\poly(\lambda)$-time and $\negl$ respectively.

\begin{property}\leavevmode
\begin{propertyenum}
    \item\label{property:efficient_generation} \emph{Efficient function generation property} \cite[Definitions 4.1 (1), 4.2 (1)]{mahadev2018verification}. There exist efficient classical probabilistic algorithms $\Gen_\calF$
    and $\Gen_\calG$ for $\calF$ and $\calG$ respectively with 
    \begin{equation}
        \Gen_\calF(1^\lambda) \to (k \in \calK_\calF, t_k) \quad \text{and} \quad \Gen_\calG(1^\lambda) \to (k \in \calK_\calG, t_k),
    \end{equation}
    where $t_k$ is known as a trapdoor. We write $\Gen_\calF(1^\lambda)_{\key}$ and $\Gen_\calG(1^\lambda)_{\key}$ for the marginal distributions of the public key from $\Gen_\calF(1^\lambda)$ and $\Gen_\calG(1^\lambda)$ respectively. 

    \item\label{property:injective_pair}
    \emph{(Disjoint) injective pair property} \cite[Definitions 4.1 (2), 4.2 (2)]{mahadev2018verification}. For all $k \in \calK_\calF\cup \calK_\calG$, $x,x'\in \calX$ with $x \neq x'$, and $b \in \{0,1\}$, we have $\Supp(f_{k,b}(x)) \medcap \Supp(f_{k,b}(x')) = \emptyset$. 
    
    For all $k \in \calK_\calF$, there exists a perfect matching $R_k\subseteq \calX \times \calX$ of $\calX$ such that $f_{k,0}(x)=f_{k,1}(x')$ (equal as distributions on $\calY$) if and only if $(x,x')\in R_k$. We call any pair $(x,x')\in R_k$ a \emph{claw}.
    In particular, for all $k\in \calK_\calF$, we have
    $\medcup_{x \in \calX} \Supp(f_{k,0}(x)) = \medcup_{x\in \calX} \Supp(f_{k,1}(x))$.

    In contrast, for all $k \in \calK_\calG$, we have
    $(\medcup_{x\in \calX} \Supp(f_{k,0}(x)))\medcap (\medcup_{x'\in \calX} \Supp(f_{k,1}(x'))) = \emptyset$.

    \item\label{property:efficient_range_superposition}
    \emph{Efficient range superposition property} \cite[Definitions 4.1~(3.c), 4.2~(3.b), 4.3 (1)]{mahadev2018verification}. Given $k\in \calK_\calF \cup \calK_\calG$, there exists an efficient quantum algorithm that prepares a state that is negligibly close to
    \begin{equation}
        \ket{\psi} \coloneqq \frac{1}{\sqrt{2\cdot |\calX|}} \sum_{b\in\{0, 1\}}\sum_{x\in \calX, y\in \calY}\sqrt{(f_{k,b}(x))(y)}\ket{b}\ket{x}\ket{y},
    \end{equation}
    in trace distance. (In the case $k\in \calK_\calF$, this property follows from applying \cite[Lemma 3.8]{mahadev2018verification} to \cite[Definition 4.1 (3.c)]{mahadev2018verification}, as done in \cite[start of Section 5.1]{mahadev2018verification}.)

    \item\label{property:adaptive_hardcore_bit}
    \emph{Adaptive hardcore bit property} \cite[Definition 4.1~(4)]{mahadev2018verification}. 
    There does not exist an efficient quantum algorithm that, given $k \leftarrow \Gen_\calF(1^\lambda)_{\key}$,
     can compute $b\in \{0,1\}$ and $x_b \in \calX$
    for some $b \in \{0,1\}$, $d \in \{0,1\}^{w} \backslash \{0^w\}$,\footnote{Formally, the computed $d$ needs to be in some set $S_{k,b,x_b}$ contained in $\{0,1\}^w\backslash\{0^w\}$ but, as $S_{k,b,x_b}$ still contains all but a negligible fraction of elements in $\{0,1\}^w$ and membership in $S_{k,b,x_b}$ can be checked efficiently given $(k,b,x_b,t_k)$, it behaves like $\{0,1\}^w\backslash\{0^w\}$ for our purposes, so we equate it to $\{0,1\}^w\backslash\{0^w\}$ for notational convenience.} and, with non-negligible advantage, a bit $d\cdot (x_0 \oplus x_1) \in \{0,1\}$ such that $(x_0, x_1) \in R_k$. In particular, this means no efficient quantum algorithm can compute a claw $(x_0,x_1)\in R_k$. This is why function pairs from $\calF$ are called claw-free.

    \item\label{property:injective_invariance}
    \emph{Injective invariance property}~\cite[Definition 4.3 (2)]{mahadev2018verification}.  There does not exist an efficient quantum algorithm that can distinguish between the distributions $\Gen_\calF(1^\lambda)_{\key}$ and $\Gen_\calG(1^\lambda)_{\key}$ with non-negligible advantage.

    \item\label{property:efficient_decoding}
    \emph{Efficient decoding property}~\cite[Definitions 4.1 (2, 3.a, 3.b), 4.2 (2, 3.a), 4.3 (1)]{mahadev2018verification}. In this paper, we define the following ``decoding maps'' that decode the output of functions from an ENTCF.  These follow \cite[Definition 2.1]{mv2021selftest} but we restate them here for completeness.
    \begin{enumerate}
        \item Let $m\in \mathbb{N}$, $m = \poly(\lambda)$. For $k \in (\calK_\calF \cup \calK_\calG)^m$, $y \in \calY^m$,
        $b \in \{0,1\}^m$, and $x \in \calX^m$, we define
        \begin{equation}
            \begin{aligned}
                \text{CHK}(k,y,b,x) \coloneqq
                \begin{cases}
                    0 &\text{ if } y_i \in \Supp( f_{k_i,b_i}(x_i)) \text{ for all }i\in [m], \\
                    1 &\text{otherwise.}
                \end{cases}
            \end{aligned}
        \end{equation}
        \item For $k \in \calK_\calG$ and $y \in \calY$, we 
        define
        \begin{equation}
        \begin{aligned}
            \hatb(k,y) \coloneqq
            \begin{cases}
                    0 &\text{ if } y \in \bigcup_{x \in \calX} \Supp(f_{k,0}(x)), \\
                    1 &\text{ if } y \in \bigcup_{x \in \calX} \Supp(f_{k,1}(x)), \\
                    \perp &\text{ otherwise.}
            \end{cases}
        \end{aligned}
        \end{equation}
        \item For $b\in \{0,1,\perp\}$, $k \in \calK_\calF \cup \calK_\calG$, and $y \in \calY$,
        we define
        \begin{equation}
            \begin{aligned}
                \hatx(b, k, y) \coloneqq
                \begin{cases}
                        \perp &\text{if } y \notin
                        \bigcup_{x \in \calX} \Supp(f_{k,b}(x)) \text{ or }b=\perp, \\
                        x
                        & \text{ such that }
                        y \in \Supp(f_{k,b}(x)).
                \end{cases}
            \end{aligned}
        \end{equation}
        In addition, for $k\in \calK_\calG$, we use the shorthand $\hatx(k,y) \coloneqq \hatx(\hatb(k,y), k, y)$.
        \item For $k \in \calK_\calF$, $y \in \calY$,
        and $d \in \{0,1\}^{w}$, we define 
        \begin{equation}
            \hath(k,y,d) \coloneqq 
            \begin{cases}
            d \cdot ( \hatx(0, k, y) \oplus
            \hatx(1, k,y)) &\text{if  $y\in \bigcup_{x \in \calX} \Supp(f_{k,0}(x))$ and  $d \neq 0^w$},
            \\
            \perp &\text{otherwise}.
            \end{cases}
        \end{equation}
\end{enumerate}
The efficient decoding property states that $\hatb$, $\hatx$, and $\hath$ can be computed efficiently given a trapdoor $t_k$ for $k$ by a classical deterministic algorithm and that $\text{CHK}$ can be computed efficiently even without a trapdoor by a classical deterministic algorithm (note $m = \poly(\lambda)$).
\end{propertyenum}

\begin{remark}
The adaptive hardcore bit and injective invariance properties are the only properties that require the LWE hardness assumption.
\end{remark}
\end{property}

\subsection{Efficient quantum operations and computational indistinguishability}

In this subsection, we record \cite[Definition 2.2]{mv2021selftest}, which formalizes the
notion of efficient quantum operations. We also append a definition for the efficiency of POVMs that is not present in \cite[Definition 2.2]{mv2021selftest}.
\begin{defn}[Efficient unitaries, isometries, observables, and measurements]\label{def:efficient_operators}{\ \\}
Let $\{\calH_\lambda \mid \lambda \in \N\}$, $\{\calH_{A,\lambda} \mid \lambda \in \N\}$, and $\{\calH_{B,\lambda} \mid \lambda \in \N\}$ be families of finite-dimensional Hilbert spaces where $\dim(\calH_{A,\lambda}) \leq \dim(\calH_{B,\lambda})$ for all $\lambda$. 
\begin{enumerate}
    \item We say a family of unitaries $\{U_\lambda \in \calL(\calH_\lambda) \mid \lambda \in \N\}$ is efficient if there exists a classical Turing machine $M$ that, on input $1^\lambda$, outputs a description of a quantum circuit with a fixed gate set that implements $U_\lambda$ in $\poly(\lambda)$ time.		
    
    \item We say a family of isometries $\{V_\lambda: \calH_{A,\lambda} \to \calH_{B,\lambda} \mid \lambda \in \N\}$ is efficient
    if there exists an efficient family of unitaries $\{U_\lambda \in \calL(\calH_{B,\lambda}) \mid \lambda \in \N\}$,
    such that $V_\lambda = U_\lambda(\1_{A,\lambda} \otimes \ket{0_{k(\lambda)}})$,
    where $\ket{0_{k(\lambda)}}$ denotes a fiducial state in an ancillary Hilbert space of dimension $k(\lambda) \coloneqq \dim(\calH_{B,\lambda})/\dim(\calH_{A,\lambda})$.  
    
    \item We say a family of binary observables $\{Z_\lambda \in \LH \mid \lambda \in \N\}$ is efficient,
    if each $\calH_{B,\lambda} \cong (\C^2)^{\otimes \poly(\lambda)}$ and there exists a family of efficient unitaries
    $\Set{U_\lambda \in \calL(\calH_{A,\lambda} \otimes \calH_{B,\lambda}) \mid \lambda \in \N}$ such that 
    for any $\ket{\psi}_{A,\lambda} \in \calH_{A,\lambda}$,
    \begin{equation}
        U_\lambda^\dagger (\sigma_Z \otimes \1) U_\lambda( \ket{\psi}_{A,\lambda} \otimes \ket{0}_{B,\lambda})
        = (Z_\lambda \ket{\psi}_{A,\lambda}) \otimes \ket{0}_{B,\lambda}.
    \end{equation}
    
    \item Let $\{\calA_\lambda \subset \N \mid \lambda \in \N\}$ be a family of finite sets. We say 
    a family of projective measurements 
    \begin{equation}
    \Set{ M_\lambda = \Set{M_\lambda^{i} \in \calL(\calH_{A,\lambda}) \mid i \in \calA_\lambda}
    \mid \lambda \in \N}
    \end{equation}
    is efficient if the family of isometries $\{V_\lambda \coloneqq \sum_{i \in \calA_\lambda} \ket{i} \otimes M_\lambda^{i} \mid \lambda \in \N \}$ is efficient. 
    
    \item Let $\{\calA_\lambda \subset \N \mid \lambda \in \N\}$ be a family of finite sets. We say that a family of POVMs
    \begin{equation}
        \Set{E_\lambda = \Set{E_\lambda^{i} \in \calL(\calH_{A,\lambda}) \mid i \in \calA_\lambda}
        \mid \lambda \in \N}
    \end{equation}
    is efficient if there exists an efficient family of isometries $\{V_\lambda: \calH_{A,\lambda} \to \calH_{B,\lambda} \mid \lambda \in \N\}$ and an efficient family of projective measurements $\{ M_\lambda = \{M_\lambda^{i} \in \calL(\calH_{B,\lambda}) \mid i \in \calA_\lambda \}
    \mid \lambda \in \N \}$ such that $E_{\lambda}^i = V_{\lambda}^{\dagger} M_{\lambda}^i V_{\lambda}$ for all $i \in \calA_\lambda$.
\end{enumerate}
\end{defn}

We formally define the notion of computational indistinguishability.
\begin{defn}[Computational indistinguishability]\label{def:comp_indist}
We say that two families of positive semi-definite operators $\{\sigma(\lambda)\}_{\lambda \in \mathbb{N}}$ and $\{\tau(\lambda)\}_{\lambda \in \mathbb{N}}$ are computationally distinguishable with advantage at most $\delta = \delta(\lambda)$ if the following holds. For all efficient families of POVMs $\{\{E_\lambda, \1-E_\lambda\} \mid \lambda \in \mathbb{N}\}$ (with respect to a \emph{fixed} polynomial in $\lambda$), there exists $\lambda_0\in \mathbb{N}$, such that for all $\lambda\geq \lambda_0$, we have

\begin{equation}\label{eq:comp_indist}
\abs{\tr[E_\lambda\sigma(\lambda)] - \tr[E_{\lambda} \tau(\lambda)]} \leq \delta(\lambda).
\end{equation}

In this case, we write $\sigma \csimeq_{\delta} \tau$. If instead of $\delta$, we have $O(\delta)$ on the right-hand side of \cref{eq:comp_indist}, we write $\sigma \capprox_{\delta} \tau$. If $\delta$ can be chosen to be a negligible function of $\lambda$, then we say that the two families of positive operators are computationally indistinguishable (without qualification).
\end{defn}

We sometimes write \cref{eq:comp_indist} in terms of an  efficient family of algorithms $\calA_{\lambda}$ that output a bit $b\in \{0,1\}$ corresponding to $\{E_{\lambda}, \1-E_{\lambda}\}$. In this case, $\tr[E_{\lambda} \psi(\lambda)]$ is written as $\Pr(\calA_\lambda \text{ outputs } 0  \text{ on input } \psi(\lambda))$. When working with computational indistinguishability, we often make the dependence on $\lambda$ implicit and abuse language by referring to states or POVMs instead of families of them.

\subsection{Approximation lemmas}

In this subsection, we prove various approximation lemmas that will be used in our soundness proof. These lemmas will be used to bound how close an arbitrary device is to being honest by its failure probability. We prove two types of approximation lemmas. The first consists of purely mathematical inequalities that hold unconditionally. The second consists of mathematical inequalities involving states and operators that hold assuming the states are computationally distinguishable with some advantage and the operators are efficient.

We first recall some facts about operators in $\LH$. These will be frequently used without further comment.
\begin{enumerate}
    \item For $A \in \LH$, $\abs{\tr[A]} \leq \norm{A}_1$ and if $A \geq 0$, $\tr[A] = \abs{\tr[A]} = \norm{A}_1$.
    \item For $A,B,C\in \LH$ and $p\in [1,\infty]$, $\|ABC\|_p\leq \|A\|_{\infty}\|B\|_p \|C\|_{\infty}$.
    \item (H\"older's inequality for Schatten $p$-norms). For 
    $A, B \in \LH$ and $p, q\in [1,\infty]$ with $1/p + 1/q = 1$, $\|AB\|_1\le  \|A\|_{p} \, \|B\|_q$. Note: (i) this is stronger than a common form with $\abs{\langle{A,B\rangle}}$ on the left-hand side, (ii) when $p=q=2$, this is also known as the Cauchy-Schwarz inequality for Schatten $2$-norms.
\end{enumerate}

We will use the following notions of approximation.
\begin{defn}\label{def:approximation}
In the following, the notation left of ``$\iff$'' is defined on its right.

\begin{enumerate}

\item \textbf{Complex vectors}. For $a,b \in \C^n$, we write
\begin{equation}
    a \simeq_\epsilon b \iff \normone{ a - b } \leq \epsilon \quad \text{and} \quad a \approx_\epsilon b \iff \normone{ a - b } \leq O(\epsilon).
\end{equation}

\item \textbf{State distance}.
For $\phi, \psi \in \LH$, we write
\begin{equation}
\phi \simeq_\epsilon \psi \iff \norm{\phi-\psi}_1^2 \leq \epsilon \quad \text{and} \quad \phi \approx_\epsilon \psi \iff \norm{\phi-\psi}_1^2 \leq O(\epsilon).
\end{equation}
The use of this notation is usually reserved for when $\phi,\psi$ are quantum states, i.e., elements of $\PosH$, hence the name ``state distance''.

\item \textbf{State-dependent operator distance}.
\newline
For $A, B \in \LH$ and $\psi \in \PosH$, we write $\norm{A}_{\psi}^2\coloneqq\tr[A^{\dagger}A\psi] = \norm{A\sqrt{\psi}}_2^2$ and
\begin{equation}
     A \simeq_{\epsilon, \psi} B \iff \norm{A-B}^2_{\psi} \leq \epsilon \quad \text{and} \quad
     A \approx_{\epsilon, \psi} B \iff \norm{A-B}^2_{\psi} \leq O(\epsilon).
\end{equation}
\end{enumerate}

\end{defn}

Note that $\simeq$ is ``more precise'' than $\approx$. For example, if $a,b\in \mathbb{C}$, then $a\simeq b\implies a\approx b$.  In these preliminaries, we choose to use $\simeq$ instead of $\approx$ to be more precise (all results still hold under changing $\simeq$ to $\approx$). This extra precision will occasionally be useful when proving our main results. However, we will usually use $\approx$ instead of $\simeq$ as $\approx$ allows us to hide constant factors and is therefore more convenient.

The following lemma relates the state-dependent norm $\norm{\cdot}_{\psi}$ to the operator norm $\norm{\cdot}$ and the trace of $\psi$. It is a generalization of \cite[Lemma 2.17]{mv2021selftest}.
\begin{lemma}
\label{lem:state_to_infty_norm}
Let $\psi \in \PosH$ and $A \in \LH$. Then
\begin{equation}
    \norm{A}_{\psi} \leq \norm{A}_{\infty} \cdot \sqrt{\tr[\psi]}.
\end{equation}
\end{lemma}
\begin{proof}
    The lemma follows from H\" older's inequality:
    \begin{equation}
        \norm{A}_{\psi}^2 = \tr[A^\dagger A\psi] 
        \leq \norm{A^\dagger A \psi}_1 
        \leq
        \norm{A^\dagger A}_{\infty}\cdot \norm{\psi}_1
        = \norm{A}_{\infty}^2 \cdot \tr[\psi].
    \end{equation}
\end{proof}

The next lemma is similar to 
\cite[Lemma 2.18(ii)]{mv2021selftest} except that we do not require the $n$ to be constant.

\begin{lemma}\label{lem:op_approx_compts}
Let $\psi_i \in \PosH$ for all $i \in [n]$ and $\psi \coloneqq \sum_{i=1}^n \psi_i$. Let $\epsilon \geq  0$. Let $A, B \in \LH$.  Then the following are equivalent:
\begin{enumerate}
    \item There exists $\epsilon_1,\ldots, \epsilon_n\geq 0$ with $\sum_{i=1}^n \epsilon_i \leq \epsilon$ such that $A\simeq_{\epsilon_i,\psi_i} B$ for all $i\in [n]$.
    \item $A\simeq_{\epsilon,\psi} B$.
\end{enumerate}
\end{lemma}

\begin{proof}
For $1\implies 2$, consider
\begin{equation}
    \tr[(A-B)^\dagger (A-B) \psi]
    = \sum_{i=1}^n\tr[(A-B)^\dagger (A-B) \psi_i]
    \leq \sum_{i=1}^n \epsilon_i \leq \epsilon.
\end{equation}

For $2\implies 1$, define $\epsilon_i \coloneqq \normstate{A-B}{\psi_i}^2 \geq 0$ for $i\in [n]$, so that $A \simeq_{\epsilon_i,\psi_i} B$ by definition. Then
\begin{equation}
    \sum_{i=1}^n \epsilon_i =  \sum_{i=1}^n\tr[(A-B)^\dagger (A-B) \psi_i] =  \tr[(A-B)^\dagger (A-B) \psi] \leq \epsilon.
\end{equation}
\end{proof}

The following replacement lemma will be frequently used in our analysis to replace the device's operators by their ideal counterparts at the cost of introducing some error. Its first and second parts are similar to \cite[Lemma 2.21]{mv2021selftest} but strengthened to include the trace of the state $\psi$. Its third and fourth parts allow us to replace states and operators in the presence of  projective measurements. Importantly, we keep the error independent of the number of projectors constituting the projective measurement. Intuitively, this should be possible because $\sum_i P^i = \1$ for $\{P^i\}_i$ a projective measurement. 
\begin{lemma}[Replacement lemma]\label{lem:replace}
\begin{enumerate}
    \item[]
    \item Let $\psi \in \PosH$ and $A, B, C \in \LH$. If $A \simeq_{\epsilon, \psi} B$
    and $\norm{C}_{\infty} = c$ for some constant $c$, then
    \begin{equation}
        \tr[CA\psi] \simeq_{c\sqrt{\tr[\psi]\cdot\epsilon}}
        \tr[CB\psi] \quad \text{and} \quad
        \tr[AC\psi] \simeq_{c\sqrt{\tr[\psi]\cdot\epsilon}}
        \tr[BC\psi].
    \end{equation}

    \item Let $\psi, \psi' \in \LH$
    and $A \in \LH$. If $\psi \simeq_{\epsilon} \psi'$ and $\norm{A}_{\infty} = c$
    for some constant $c$, then
    \begin{equation}
        \tr[A\psi] \simeq_{c\sqrt{\epsilon}} \tr[A\psi'].
    \end{equation}

    \item Let $\psi\in \PosH$, and $A,B\in \LH$ be Hermitian. Let $\{X_1,\ldots,X_n\}$ be a set of mutually commuting binary observables and let $\{Y_1,\ldots,Y_n\}$ be another such set. If $A\simeq_{\epsilon,\psi} B$, then, for all $i\in [n]$,
    \begin{equation}
        \sum_{u\in \{0,1\}^n} \abs{\tr[(A-B) X_1^{(u_1)} X_2^{(u_2)} \cdots X_n^{(u_n)} Y_i^{(u_i)} Y_{i+1}^{(u_{i+1})} \cdots Y_n^{(u_n)} \psi]} \leq \sqrt{\tr[\psi]\cdot \epsilon}.
    \end{equation}
    
    \item Let $\psi, \psi' \in \LH$ be Hermitian, and $\{P^i\}_{i\in [n]}$ be a projective measurement on $\calH$. If $\psi \simeq_{\epsilon} \psi'$, then 
    \begin{equation}
        \sum_{i=1}^n \abs{\tr[P^i (\psi-\psi')]} \leq \sum_{i=1}^n \normone{P^i (\psi-\psi')P^i]} \leq \normone{\psi-\psi'}  \leq  \sqrt{\epsilon}.
    \end{equation}
\end{enumerate}
\end{lemma}

\begin{proof}
The first equation of the first part follows from
\begin{equation}
    \abs{\tr[C(A-B)\psi]} = \abs{\langle C, A-B \rangle_{\psi}} 
    \leq \norm{C^\dagger}_{\psi} \cdot \norm{A-B}_{\psi} 
    \leq c\sqrt{\tr[\psi]} \sqrt{\epsilon},
\end{equation}
where, in the last inequality, we use \cref{lem:state_to_infty_norm} to bound
$\norm{C^\dagger}_{\psi} \leq c \sqrt{\tr[\psi]}$. The second equation of the first part can be shown analogously.

The second part follows from $\abs{\tr[A(\psi - \psi')]} \leq \norm{A}_{\infty} \cdot \norm{\psi - \psi'}_1$.

Consider the third part. We first write $P^{u} \coloneqq X_1^{(u_1)} X_2^{(u_2)} \cdots X_n^{(u_n)}$,  $\tilde{P}^{u} \coloneqq X_i^{(u_i)} X_{i+1}^{(u_{i+1})} \cdots X_n^{(u_n)}$, and $Q^{u} \coloneqq Y_i^{(u_i)} Y_{i+1}^{(u_{i+1})} \cdots Y_n^{(u_n)}$ for convenience. Note that $(P^u)^{\dagger} = (P^u)^2 = P^u$ by the commutativity of the $X_i$s. Similarly, $(Q^u)^{\dagger} = (Q^u)^2 = Q^u$ by the commutativity of the $Y_i$s. Then, the third part follows from
\begin{equation}
\begin{aligned}
    &\quad \sum_{u} \abs{\tr[(A-B)P^{u} Q^{u}\psi]} 
    \\
    &= \sum_{u} \abs{\langle P^{u}(A-B) \sqrt{\psi}, P^{u} Q^{u} \sqrt{\psi} \rangle} 
    &&\text{($(P^u)^{\dagger}P^u = P^u$)}
    \\
    & \leq \sum_u \norm{P^{u}(A-B) \sqrt{\psi}}_2\norm{P^{u}Q^{u}\sqrt{\psi}}_2
    &&\amatrix{\text{Cauchy-Schwarz}}{\text{for Schatten $2$-norms}}
    \\
    & = \sum_u \sqrt{\tr[P^{u}(A-B)\psi (A-B)]}\sqrt{\tr[P^u Q^u \psi Q^u]}
    &&\text{($(P^u)^{\dagger}P^u = P^u$; $A,B$ Hermitian)}
    \\
    & \leq \sqrt{\sum_u \tr[P^{u}(A-B)\psi (A-B)]} \cdot \sqrt{\sum_u \tr[P^u Q^u \psi Q^u]}
    &&\text{(Cauchy-Schwarz)}
    \\
    & = \sqrt{\tr[(A-B)^2\psi]} \cdot \sqrt{\sum_{u_{i},u_{i+1},\ldots, u_n} \tr[\tilde{P}^u Q^u \psi Q^u]}
    &&\text{$\Bigl(\sum_{u_1,\ldots,u_{i-1}}P^u = \tilde{P}^u\Bigr)$}
    \\
    & \leq \sqrt{\epsilon} \cdot \sqrt{\sum_{u_{i},u_{i+1},\ldots,u_n} \tr[Q^u \psi Q^u]} = \sqrt{\tr[\psi]\cdot \epsilon}
    &&\text{(lemma conditions).}
\end{aligned}
\end{equation}

Finally, consider the fourth part. We write the Hermitian operator $\sigma \coloneqq \psi - \psi'$ in terms of the decomposition $\sigma = R - S$ where $R, S$ are positive semi-definite operators with $RS=0$, so that $\abs{\sigma} = R+S$. Then, the fourth part follows from
\begin{equation}
\begin{aligned}
    \sum_{i=1}^n \abs{\tr[P^i \sigma]}  &= \sum_{i=1}^n \abs{\tr[P^i \sigma P^i]} \leq  \sum_{i=1}^n \normone{P^i \sigma P^i}  \leq \sum_{i=1}^n (\normone{P^i R P^i} + \normone{P^i S P^i} ) 
    \\
    & = \sum_{i=1}^n (\tr[P^i R P^i] + \tr[P^i S P^i]) = \tr[R+S] = \tr[\abs{\sigma}] = \normone{\psi-\psi'} \leq \sqrt{\epsilon}.
\end{aligned}
\end{equation}
\end{proof}

The next lemma is elementary and states that the trace norm of a sum of orthogonal operators equals the sum of their trace norms. We give a proof for completeness.
\begin{lemma}\label{lem:normone_orthogonal}
    Suppose $A_1,A_2,\ldots,A_n\in \LH$ are Hermitian and $A_iA_j=0$ for all $i\neq j$, then
    \begin{equation}
        \bignorm{\sum_{i=1}^n A_i}_1 = \sum_{i=1}^n \normone{A_i}.
    \end{equation}
\end{lemma}
\begin{proof}
    It clearly suffices to prove the lemma  for the case $n=2$. In this case, let us write $A\coloneqq A_1$ and $B\coloneqq A_2$. First, note that $AB = 0 = (AB)^{\dagger} = B^{\dagger}A^{\dagger} = BA$ so $A$ and $B$ commute. Therefore, we can simultaneously diagonalize $A$ and $B$ and write $A = \sum_i \lambda_i \ketbrasame{i}$ and $B = \sum_i \mu_i \ketbrasame{i}$ for some $\lambda_i,\mu_i\in \mathbb{R}$ and $\{\ket{i}\}_{i\in [\dim(\calH)]}$ an orthonormal basis of $\calH$. Now, $AB=0$ implies $\sum_i \lambda_i\mu_i \ketbrasame{i} = 0$ and so the sets $S\coloneqq \{i| \lambda_i\neq 0\}$ and $T\coloneqq \{i| \mu_i\neq 0\}$ are disjoint. Therefore, the lemma follows from
    \begin{equation*}
        \normone{A+B} = \sum_i \abs{\lambda_i + \mu_i} = \sum_{i\in S} \abs{\lambda_i} + \sum_{i\in T} \abs{\mu_i} = \normone{A} + \normone{B}. \qedhere
    \end{equation*}
\end{proof}

The next lemma is similar to \cite[Lemma 2.22]{mv2021selftest} but also strengthened to include the trace of $\psi$.
\begin{lemma}\label{lem:operator_to_state}
Let $A, B, C \in \LH$ with $\norm{C}_\infty = c$ for some constant $c$, and let $\psi \in \PosH$.
Then the following holds:
\begin{equation}\label{eq:operator_to_state1}
    A \simeq_{\epsilon, \psi} B
    \implies  A \psi C \simeq_{c^2 \tr[\psi]\epsilon} B \psi C
    \text{ and } 
    C \psi A^\dagger 
    \simeq_{c^2 \tr[\psi]\epsilon} C \psi B^\dagger.
\end{equation}
In particular, let $P, Q \in \LH$ be such that $\norminfty{P},\norminfty{Q}\leq 1$, then
\begin{equation}\label{eq:operator_to_state2}
	P\simeq_{\epsilon,\psi}Q \implies P\psi P \simeq_{4\tr[\psi] \epsilon}Q\psi Q.
\end{equation}
\end{lemma}

\begin{proof}
The proof of  \cref{eq:operator_to_state1} is the same as that of \cite[Lemma 2.22]{mv2021selftest} except we use the bound $\norm{C}_{\psi}^2\leq c^2 \tr[\psi]$ (\cref{lem:state_to_infty_norm}). \cref{eq:operator_to_state2} follows from the first part via
\begin{align*}
    P\simeq_{\epsilon,\psi}Q 
    &\implies P\psi P \simeq_{\norm{P}_\psi^2 \epsilon} P\psi Q \quad \text{and} \quad P\psi Q \simeq_{\norm{Q}_\psi^2\epsilon} P\psi Q 
    &&\text{(\cref{eq:operator_to_state1})}
    \\
    &\implies P\psi P \simeq_{\epsilon\cdot (\norm{P}_\psi + \norm{Q}_\psi)^2}Q\psi Q
    &&\text{(triangle inequality)}
    \\
    &\implies P\psi P \simeq_{4\tr[\psi]\epsilon}Q\psi Q &&\text{($\norminfty{P},\norminfty{Q} \leq 1$)}. \qedhere
\end{align*}
\end{proof}

The next lemma is used to bound the distance between post-measurement states when similar measurements are applied to the same state or when the same measurement is applied to similar states.
\begin{lemma}[Post-measurement approximation lemma]\label{lem:post_meas_approx}
\begin{enumerate}
    \item[]
    \item Let $\psi_j \in \PosH$ for $j\in [n]$ and $\psi \coloneqq \sum_{j=1}^n \psi_j$. Let $\{P^i\}_{i\in [m]}$ and $\{Q^i\}_{i\in [m]}$ be two projective measurements on $\calH$. If $\sum_{i,j} \norm{P^i-Q^i}_{\psi_j}^2\leq \epsilon$, then
    \begin{equation}
        \sum_{i\in [m],j\in [n]} P^i \psi_j P^i \otimes \ketbrasame{i,j} \simeq_{4\tr[\psi]\cdot \epsilon} \sum_{i\in [m],j\in [n]} Q^i \psi_j Q^i \otimes \ketbrasame{i,j}.
    \end{equation}
    \item Let $\psi_j, \psi_j' \in \PosH$ for $j\in [n]$ and $\psi \coloneqq\sum_{j=1}^n\psi_j$ and $\psi' \coloneqq\sum_{j=1}^n\psi_j'$. Let $\{P^i\}_{i\in [m]}$ be a projective measurement on $\calH$. If $\sum_{j=1}^m \normone{\psi_j-\psi_j'} \leq \epsilon$, then
    \begin{equation}
        \sum_{i\in[m],j\in[n]} P^i\psi_j P^i \otimes \ketbrasame{i,j} \simeq_{\epsilon^2} \sum_{i\in [m],j\in [n]} P^i\psi_j' P^i \otimes \ketbrasame{i,j}.
    \end{equation}
\end{enumerate}
\end{lemma}
\begin{proof}
    The first part of the lemma follows from:
    \begin{align*}
        &\bignorm{\sum_{i,j} P^i \psi_j P^i \otimes \ketbrasame{i,j} - \sum_{i,j} Q^i \psi_j Q^i \otimes \ketbrasame{i,j}}_1
        \\ 
        &= \sum_{i,j} \normone{P^i \psi_j P^i - Q^i \psi_j Q^i} &&\text{(\cref{lem:normone_orthogonal})}
        \\
        &\leq \sum_{i,j} \normone{(P^i-Q^i) \psi_j P^i} + \normone{Q^i \psi_j (P^i-Q^i)} &&\text{(triangle inequality)}
        \\
        &\leq \sum_{i,j} \norm{(P^i-Q^i)\sqrt{\psi_j}}_2 (\norm{\sqrt{\psi_j}P^i}_2 + \norm{\sqrt{\psi_j}Q^i}_2) &&\amatrix{\text{Cauchy-Schwarz}}{\text{for Schatten $2$-norms}}
        \\
        &= \sum_{i,j} \norm{P^i-Q^i}_{\psi_j} \Bigl(\sqrt{\tr[P^i\psi_j]} + \sqrt{\tr[Q^i\psi_j]}\Bigr) &&\text{(definitions)}
        \\
        &\leq \sqrt{\sum_{i,j} \norm{P^i-Q^i}_{\psi_j}^2} \cdot \Biggl(\sqrt{\sum_{i,j} \tr[P^i\psi_j]} + \sqrt{\sum_{i,j} \tr[Q^i\psi_j]} \Biggr) &&\text{(Cauchy-Schwarz)}
        \\
        &\leq 2 \sqrt{\epsilon} \sqrt{\tr[\psi]} &&\text{(lemma conditions)}.
    \end{align*}
    
    The second part of the lemma follows from
    \begin{equation}
        \bignorm{\sum_{i,j} P^i\psi_j P^i \otimes \ketbrasame{i,j} - \sum_{i,j} P^i\psi_j' P^i \otimes \ketbrasame{i,j}}_1 = \sum_{i,j} \normone{P^i\psi_j P^i - P^i\psi_j'P^i}
        \leq \sum_j \normone{\psi_j - \psi_j'} \leq \epsilon,
    \end{equation}
    where we used \cref{lem:normone_orthogonal} for the first equality, part 4 of the replacement lemma (\cref{lem:replace}) for the first inequality, and the lemma condition for the second inequality.
\end{proof}

The next two lemmas are proved in \cite{mv2021selftest}. We omit the proofs.
\begin{lemma}[{\cite[Lemma 2.23]{mv2021selftest}}]\label{lem:mv_223}
    Let $\calH_1$ and $\calH_2$ be Hilbert spaces with
    $\dim(H_1) \leq \dim(H_2)$, $V$ be an isometry: $\calH_1 \to \calH_2$, and $A$ and $B$ be binary observables on $\calH_1$ and $\calH_2$ respectively.
    Then, the following holds for all $\psi \in \Pos(\calH_1)$:
    \begin{equation}
    \begin{aligned}
        &VAV^\dagger \simeq_{\epsilon, V\psi V^\dagger} B
        \implies A \simeq_{\epsilon, \psi} V^\dagger BV, \\
        &A \simeq_{\epsilon, \psi} V^\dagger B V 
        \implies VAV^\dagger \simeq_{\sqrt{\epsilon}, V\psi V^\dagger} B.
    \end{aligned}
    \end{equation}
\end{lemma}

\begin{lemma}[{\cite[Lemma 2.24]{mv2021selftest}}]\label{lem:mv_224}
    Let $\calH_1$ and $\calH_2$ be Hilbert spaces with
    $\dim(H_1) \leq \dim(H_2)$, $V$ be an isometry: $\calH_1 \to \calH_2$, and $A$ and $B$ be binary observables on $\calH_1$ and $\calH_2$ respectively.
    Then, the following holds for all $\psi \in \Pos(\calH_1)$ and $b \in \{0,1\}$:
    \begin{equation}
    \begin{aligned}
        &A \simeq_{\epsilon, \psi} V^\dagger B V 
        \implies A^{(b)} \simeq_{\epsilon, \psi } V^\dagger B^{(b)} V,
        \\
        &  B \simeq_{\epsilon, V\psi V^\dagger} VAV^\dagger
        \implies B^{(b)} \simeq_{\epsilon, V \psi V^\dagger} V A^{(b)} V^\dagger.
    \end{aligned}
    \end{equation}
\end{lemma}

The next lemma will be used to characterize the states of the quantum device given a characterization of its observables. The $i\in [n]$ will index qubits.
\begin{lemma}\label{lem:linear_error_accumulation}
    Let $P_1,\ldots,P_n\in \PosH$ be projectors and $A\in \LH$ be such that 
    \begin{equation}
    \normone{A-P_i A P_i} \leq \epsilon_i,
    \end{equation}
    for all $i\in [n]$. Then, writing $P_{1:n} \coloneqq P_1 P_{2} \cdots P_n$, we have
    \begin{equation}
        \normone{A- P_{1:n}^{\dagger} A P_{1:n}} \leq \sum_{i=1}^n \epsilon_i.
    \end{equation}
\end{lemma}

\begin{proof}
We argue by induction on $n\geq 0$. The base case, $\normone{A-A}=0$, clearly holds. 

Now, for $i\in [n+1]$, we write $P_{i:n+1} \coloneqq P_i P_{2} \cdots P_{n+1}$. Then,
\begin{align*}
    &\quad \normone{A - P_{1:n+1}^{\dagger} A P_{1:n+1}} 
    \\
    &\leq \normone{A -  P_{2:n+1}^{\dagger} A P_{2:n+1}} + \normone{P_{2:n+1}^{\dagger} A P_{2:n+1} - P_{1:n+1}^{\dagger} A P_{1:n+1}} &&\text{(triangle inequality)}
    \\
    &\leq \sum_{i=2}^{n+1} \epsilon_i + \normone{P_{2:n+1}^{\dagger} A P_{2:n+1} - P_{2:n+1}^{\dagger}P_1 A P_1 P_{2:n+1}}
    &&\amatrix{\text{inductive hypothesis}}{\text{and $P_{1:n+1} = P_1 P_{2:n+1}$}}
    \\
    &\leq \sum_{i=2}^{n+1}\epsilon_i + \normone{A - P_1 A P_1} \leq \sum_{i=1}^{n+1} \epsilon_i
    && \amatrix{
        \text{$\norminfty{P_{2:n+1}}\leq1$ then}}{
        \text{lemma conditions}},
\end{align*}
which completes the proof.
\end{proof}

Unlike the preceding lemmas, all remaining lemmas of this subsection concern computational indistinguishability. In particular, this means they only hold with respect to efficient operations.
\begin{lemma}[Triangle inequality for computational indistinguishability]\label{lem:ci_triangle}
    Let $\{\rho_i \mid i \in [n]\} \subseteq \PosH$ be such that $\rho_i \csimeq_{\epsilon_i} \rho_{i+1}$ for all $i \in [n-1]$, then
    \begin{equation}
        \rho_1 \csimeq_{\sum_{i=1}^{n-1} \epsilon_i} \rho_n.
    \end{equation}
\end{lemma}

\begin{proof}
    Let $\{E, \1-E\}$ be an efficient POVM. The lemma follows from
    \begin{equation*}
        \abs{\tr[E \rho_1] - \tr[E \rho_n]} \leq \sum_{i=1}^{n-1} 
        \abs{\tr[E \rho_i] - \tr[E \rho_{i+1}]} \leq  \sum_{i=1}^{n-1} \epsilon_i. \qedhere
    \end{equation*}
\end{proof}

\begin{lemma}[Partitioning property of computational indistinguishability]\label{lem:ci_partitioning}
    Let $\sigma,\tau \in \Pos(H)$. Let $\{\Pi^i\}_{i\in [n]}$ be an efficient projective measurement. If $\sigma \csimeq_{\delta} \tau$, then there exists $\delta_i\geq 0$ such that 
    \begin{equation}
        \Pi^i \sigma \Pi^i \simeq_{\delta_i} \Pi^i \tau \Pi^i,
    \end{equation}
    for all $i\in [n]$, and $\sum_{i=1}^n \delta_i \leq 2\delta$.
\end{lemma}
\begin{proof}
    For $i\in[n]$, let 
    \begin{equation}
        \alpha_i \coloneqq \max \tr[A(\Pi^i \sigma \Pi^i - \Pi^i \tau\Pi^i)] \quad \text{and} \quad  \beta_i \coloneqq \max \tr[-B(\Pi^i \sigma \Pi^i - \Pi^i \tau\Pi^i)],
    \end{equation}
    where the max is taken over all efficient POVM elements $A$ and $B$. Because $A$ and $B$ can be zero, we see $\alpha_i,\beta_i\geq 0$. Let $A_i$ and $B_i$ be the corresponding maximizers. 

    By definition, $\delta_i \coloneqq \max\{\alpha_i,\beta_i\}$ equals how computationally indistinguishable $\Pi^i \sigma \Pi^i$ is from  $\Pi^i \tau\Pi^i$. As $\alpha_i,\beta_i\geq 0$, we have $\delta_i \leq \alpha_i+ \beta_i$, so it suffices to upper bound $\sum_{i=1}^n \alpha_i+\beta_i$.

    Consider the following efficient algorithm for distinguishing $\sigma$ and $\tau$: 
    \begin{enumerate}
        \item Measure $\{\Pi^i\}_i$ on the input state and record the outcome.
        \item If the outcome is $i$ then measure $\{A_i, \1-A_i\}$.
        \item Output $0$ if the result corresponds to $A_i$ and $1$ otherwise.
    \end{enumerate}
    The probability that this algorithm outputs $0$ on any input $\psi\in \PosH$ is
    \begin{equation}
        \sum_{i=1}^n \tr[\Pi^i\psi\Pi^i] \, \tr\Bigl[A_i \cdot \frac{\Pi^i\psi\Pi^i}{\tr[\Pi^i\psi\Pi^i]}\Bigr] = \sum_{i=1}^n \tr[A_i \Pi^i \psi \Pi^i].
    \end{equation}
    Therefore, by the definition of $\sigma \csimeq_{\delta} \tau$, we deduce
    \begin{equation}
        \sum_{i=1}^n \alpha_i = \sum_{i=1}^n \tr[A_i (\Pi^i \sigma \Pi^i - \Pi^i \tau \Pi^i)] \leq \delta.
    \end{equation}
    We can similarly show that $\sum_{i=1}^n\beta_i \leq \delta$. Therefore, $\sum_{i=1}^n \delta_i\leq 2\delta$ as required. 
\end{proof}

We record the following lemma.
\begin{lemma}[{\cite[Lemma 2.6]{mv2021selftest}}]
\label{lem:sum_unitary}
Let $U_1, U_2$ be efficient unitaries on $\calH$. Then, $(U_1 + U_2)^\dagger (U_1 + U_2)$ and $(U_1 - U_2)^\dagger (U_1 - U_2)$
are observables and there exists an efficient quantum algorithm that given a state $\psi \in \calD(\calH)$ outputs a bit $b$ with
\begin{equation}
    \pr[b=0|\psi] =
    \frac{1}{4} \tr[(U_1 +U_2)^\dagger (U_1+U_2)\psi] \quad \text{and} \quad
    \pr[b=1|\psi] =
    \frac{1}{4} \tr[(U_1 -U_2)^\dagger (U_1-U_2)\psi].
\end{equation}
\end{lemma}
\begin{remark}
The proof in \cite{mv2021selftest} uses a purification step that we do not know how to make efficient. However, they could have worked directly with mixed states to derive the result (see the proof of \cref{lem:lifting_proj}).
\end{remark}

We end this section with two lemmas that allow us to use the computational indistinguishability between two states $\psi$ and $\psi'$ to ``lift'' a true statement about $\psi$ onto $\psi'$.

\begin{lemma}[{Lifting lemma, see \cite[Lemma 2.25 (i--v)]{mv2021selftest}}]\label{lem:lifting}
    Let $\psi, \psi' \in \calD(\calH)$ be such that
    $\psi \csimeq_\delta \psi'$.
    \begin{enumerate}
        \item Let $A$ be an efficient binary observable
        on $\calH$. Then, $\tr[A\psi] \simeq_{2\delta} \tr[A\psi']$.
        \item Let $A,B$ be efficient binary observables on 
        $\calH$. Then,
        \begin{align}
            A \simeq_{\epsilon, \psi} B \implies 
            A \simeq_{\delta+\epsilon, \psi'} B.
        \end{align}
        \item Let $A,B$ be efficient binary observables on 
        $\calH$. Then,
        \begin{align}
            [A, B] \simeq_{\epsilon,\psi} 0
            \implies [A, B] \simeq_{\delta+\epsilon, \psi'} 0.
        \end{align}
        \item Let $A,B$ be efficient binary observables on 
        $\calH$. Then,
        \begin{align}
            \{A, B\} \simeq_{\epsilon,\psi} 0
            \implies \{A, B\} \simeq_{\delta+\epsilon, \psi'} 0.
        \end{align}
        \item Let $\calH'$ be another Hilbert space with $\dim(\calH') \geq \dim(\calH)$, $U$ an efficient unitary on $\calH$, $B$ an efficient binary observable on $\calH'$, and $V:\calH \rightarrow \calH'$ an efficient isometry. Then
        \begin{equation}\label{eq:lifting_vbvu}
            \real \tr[V^{\dagger}BV \, U \psi] \simeq_{2\delta} \real \tr[V^{\dagger}BV \, U \psi'] 
            \quad \text{and} \quad
            \real \tr[U \, V^{\dagger}BV \psi] \simeq_{2\delta} \real \tr[U \, V^{\dagger}BV \psi'].
        \end{equation}
    \end{enumerate}
\end{lemma}
\begin{proof}
The first four parts of the lemma (and their proofs) correspond to \cite[Lemma 2.25 (i--iv)]{mv2021selftest} verbatim, except that we record the constant factors arising from their proof. The fifth part is implicit in \cite[Proof of Lemma 2.25(v)]{mv2021selftest}. We present a full proof of the fifth part for completeness.

First, note that the second equation in \cref{eq:lifting_vbvu} follows from the first by setting $V$ to $VU^{\dagger}$, which is efficient. Therefore, it suffices to prove the first equation in \cref{eq:lifting_vbvu}. Let $W \in \calL(\calH')$ be an efficient unitary such that
$V = W(\1 \otimes \ket{0}_{\calH''})$ where $\dim(\calH'') =
\dim(\calH')/\dim(\calH)$. $\calH''$ exists because we can add dimension to $\calH'$.
Since $U \otimes \1_{\calH''}$ and $W^\dagger B W$ are efficient unitaries,
we can apply \cref{lem:sum_unitary} to see that
there exists an efficient algorithm that outputs a bit $b$ with
\begin{equation}
    \pr[b=0|\psi] = \frac{1}{4}\tr[
    (U\otimes\1_{\calH''} - W^\dagger B W)^\dagger
    (U\otimes\1_{\calH''} - W^\dagger B W) (\psi \otimes \ketbra{0}{0}_{\calH''})].
\end{equation}
Since $\psi \csimeq_\delta \psi'$, this means
\begin{equation}
\begin{aligned}
    &\tr[
    (U\otimes\1_{\calH''} - W^\dagger B W)^\dagger
    (U\otimes\1_{\calH''} - W^\dagger B W) (\psi' \otimes \ketbra{0}{0}_{\calH''})] \\
    \simeq_{4\delta}
    &\tr[
    (U\otimes\1_{\calH''} - W^\dagger B W)^\dagger
    (U\otimes\1_{\calH''} - W^\dagger B W) (\psi \otimes \ketbra{0}{0}_{\calH''})].
\end{aligned}
\end{equation}
On the other hand, we have
\begin{align}
    &\tr[
    (U\otimes\1_{\calH''} - W^\dagger B W)^\dagger
    (U\otimes\1_{\calH''} - W^\dagger B W) (\psi' \otimes \ketbra{0}{0}_{\calH''})]
    = 2 - 2\real \tr[V^\dagger B VU \psi'], \\
    &\tr[
    (U\otimes\1_{\calH''} - W^\dagger B W)^\dagger
    (U\otimes\1_{\calH''} - W^\dagger B W) (\psi \otimes \ketbra{0}{0}_{\calH''})]
    = 2 - 2\real \tr[V^\dagger B VU \psi].
\end{align}
Therefore,
\begin{equation}
    \real \tr[V^\dagger B VU \psi] 
    \simeq_{2\delta} \real \tr[V^\dagger B VU \psi'],
\end{equation}
as required.
\end{proof}

In our work, we will need a new type of lifting lemma for handling projective measurements.

\begin{lemma}[Lifting-under-projections lemma]\label{lem:lifting_proj}
   Let $\{\Pi^u\}_{u\in \{0,1\}^n}$ and $\{P^u\}_{u\in \{0,1\}^n}$ be two efficient projective measurements on $\calH$. Let $\psi, \psi' \in \PosH$ be such that $\psi\csimeq_\delta \psi'$ for some $\delta\geq 0$. Then, there exists $\alpha_u, \beta_u  \geq 0$ for $u\in \{0,1\}^n$ such that 
    \begin{equation}\label{eq:sum_proj_partu}
    \begin{aligned}
        (\Pi^u + P^u) \psi (\Pi^u + P^u) \csimeq_{\alpha_u}  (\Pi^u + P^u) \psi' (\Pi^u + P^u),
        \\
        (\Pi^u - P^u) \psi (\Pi^u - P^u) \csimeq_{\beta_u}  (\Pi^u - P^u) \psi' (\Pi^u - P^u),
    \end{aligned}
    \end{equation}
    where $\sum_{u\in \{0,1\}^n} (\alpha_u + \beta_u) \leq 4\delta$. 
    
    Moreover, let $\{Q^w\}_{w\in \{0,1\}^m}$ be another efficient projective measurement such that $Q^w$ commutes with $\Pi^u$ for all $w\in \{0,1\}^m$ and $u\in \{0,1\}^n$. Then, we have
    \begin{equation}
        \sum_{w \in \{0,1\}^m, \, u\in \{0,1\}^n} \abs{\real \tr[Q^w \Pi^u P^u \psi] - \real \tr[Q^w \Pi^u P^u \psi']} \leq 2\delta.
    \end{equation}
\end{lemma}

\begin{proof}
    Let $\rho^{\pm}_u \coloneqq  (\Pi^u \pm P^u) \psi (\Pi^u \pm P^u) $ and $\rho'^\pm_u \coloneqq  (\Pi^u \pm P^u) \psi' (\Pi^u \pm P^u)$. Let
    \begin{equation}
        \alpha^+_u \coloneqq \max_{A_u} \tr[A_u(\rho^{+}_u - \rho'^{+}_u)] \geq 0  \quad \text{and} \quad \beta^+_u \coloneqq \max_{B_u} \tr[B_u (\rho^{-}_u - \rho'^{-}_u)] \geq 0,
    \end{equation}
    where the max is taken over efficient POVM elements $A_u$ and $B_u$. Let $A_u^*$ and $B_u^*$ denote the respective maximizers. By the linear combination of unitaries technique \cite{childswiebe2012lcu} (also see \cref{lem:sum_unitary}) and the fact that $\{\Pi^u\}_{u\in \{0,1\}^n}$ and $\{P^u\}_{u\in \{0,1\}^n}$ are efficient, we see that the following isometry is efficient:
    \begin{equation}
        V \coloneqq \frac{1}{2} \Bigl( \ket{0} \otimes \sum_{u\in \{0,1\}^n} \ket{u} \otimes (\Pi^u + P^u) + \ket{1} \otimes \sum_{u\in \{0,1\}^n} \ket{u} \otimes (\Pi^u - P^u) \Bigr).
    \end{equation}

    In addition, the POVM $\{\Gamma, \1-\Gamma\}$ is efficient, where 
    \begin{equation}
        \Gamma \coloneqq \sum_{u\in \{0,1\}^n} \ketbrasame{0,u} \otimes A^*_u + \sum_{u\in \{0,1\}^n} \ketbrasame{1,u} \otimes B^*_u.
    \end{equation}
    Therefore, there exists an efficient algorithm that, given any $\sigma \in \PosH$, outputs $0$ with probability
    \begin{equation}  
         \tr[\Gamma \calV \sigma \calVdag] = \frac{1}{2} \Bigl( \sum_{u\in \{0,1\}^n} \tr[A_u^*(\Pi^u + P^u)\sigma(\Pi^u + P^u)] + \tr[B_u^*(\Pi^u - P^u)\sigma(\Pi^u - P^u)]\Bigr).
    \end{equation}
    Therefore, by the definition of $\psi \csimeq_\delta \psi'$, we have $ \tr[\Gamma \calV \rho \calVdag] -  \tr[\Gamma \calV \rho' \calVdag] \leq \delta$. That is,
    \begin{equation}
        \frac{1}{2}\sum_{u\in \{0,1\}^n} (\alpha^+_u + \beta^+_u) \leq \delta.
    \end{equation}
    Now, let 
    \begin{equation}
        \alpha^-_u \coloneqq \max_{A_u} -\tr[A_u(\rho^{+}_u - \rho'^{+}_u)] \geq 0  \quad \text{and} \quad \beta^-_u \coloneqq \max_{B_u} -\tr[B_u (\rho^{-}_u - \rho'^{-}_u)] \geq 0,
    \end{equation}
    where the maximization is again over efficient POVM elements $A_u$ and $B_u$. We can similarly show that 
    \begin{equation}
        \frac{1}{2}\sum_{u\in \{0,1\}^n} (\alpha^-_u + \beta^-_u) \leq \delta.
    \end{equation}
    
    But $\alpha_u \coloneqq \max\{\alpha^+_u, \alpha^-_u\}$  equals how computational indistinguishable $\rho^+_u$ and $\rho'^+_u$ are. Likewise $\beta_u \coloneqq \max\{\beta^+_u, \beta^-_u\}$ equals how computational indistinguishable $\rho^-_u$ and $\rho'^-_u$ are. Therefore, we obtain
    \begin{equation}
        \frac{1}{2}\sum_{u\in \{0,1\}^n} (\alpha_u + \beta_u) \leq \frac{1}{2}\sum_{u\in \{0,1\}^n} (\alpha^+_u + \beta^+_u + \alpha^-_u + \beta^-_u) \leq 2\delta.
    \end{equation}
    Hence the first part of the lemma. 
    
    Now consider the second, ``moreover'', part of the lemma. Let $\rho^{\pm}_{w,u} \coloneqq  Q^w (\Pi^u \pm P^u) \psi (\Pi^u \pm P^u) Q^w$ and  $\rho'^\pm_{w,u} \coloneqq  Q^w (\Pi^u \pm P^u) \psi' (\Pi^u \pm P^u) Q^w$. Since $\{Q^w\}_{w\in \{0,1\}^m}$ is efficient, we can use the partitioning property of computational indistinguishability (\cref{lem:ci_partitioning}) to deduce that there exists $\alpha_{w,u},\beta_{w,u}\geq0$ for $u\in \{0,1\}^n, w\in \{0,1\}^m$ such that
    \begin{equation}
        \rho^{+}_{w,u} \csimeq_{\alpha_{w,u}} \rho^{+}_{w,u} \quad \text{and} \quad \rho^{-}_{w,u} \csimeq_{\beta_{w,u}} \rho^{-}_{u,w},
    \end{equation}
    where $\sum_{w\in \{0,1\}^m} \alpha_{w,u} \leq 2\alpha_u$ and $\sum_{w\in \{0,1\}^m} \beta_{w,u} \leq 2\beta_u$. 
    
    The second part of the lemma then follows from the first part by
    \begin{equation}
    \begin{aligned}
        &\sum_{w \in \{0,1\}^m, \,u\in \{0,1\}^n} \abs{\real \tr[Q^w \Pi^u P^u \psi] - \real \tr[Q^w\Pi^u P^u \psi'] }
        \\
        &= \sum_{w \in \{0,1\}^m, \,u\in \{0,1\}^n} \Bigl| \frac{1}{4} \Bigl(\tr[\rho^+_{w,u}]- \tr[\rho^-_{w,u}] - (\tr[\rho'^+_{w,u}]- \tr[\rho'^-_{w,u}]) \Bigr) \Bigr|
        \\
        &\leq \frac{1}{4} \sum_{w \in \{0,1\}^m, \,u\in \{0,1\}^n} (\alpha_{w,u} + \beta_{w,u}) \leq \frac{1}{2} \sum_{u\in \{0,1\}^n} (\alpha_u + \beta_u) \leq 2 \delta,
    \end{aligned}
    \end{equation}
    where the first equality uses the fact that $Q^w$ commutes with $\Pi^u$ for all $w\in \{0,1\}^m$ and $u\in \{0,1\}^n$.
\end{proof}

\section{Completeness of self-testing protocol}
\label{sec:completeness}

In this section, we introduce our self-testing protocol in Fig.~\ref{fig:protocol}. We then prove \cref{thm:completeness} by describing an efficient honest quantum device that passes our protocol with probability $\geq 1 - 2N \cdot \negl$ and noting that the verification can be efficiently performed classically.

\begin{figure}
\centering
\scalebox{0.785}{
\begin{tabular}{p{17cm}}
\hline
\smallskip
1. Input: $\lambda \in \mathbb{N}$. Set $N = \poly(\lambda)$. Given a distribution $\mu$ on $\{0,1\}^{2N}$.  Sample $\theta\leftarrow_U \thetaset$ uniformly at random. Sample $2N$ key-trapdoor pairs $(k_1,t_{k_1}),\ldots,(k_{2N},t_{k_{2N}})$ from an ENTCF according to $\theta$ as follows:
\begin{enumerate}[leftmargin=50pt]
    \item[$\theta \in \mathrm{[} 2N \mathrm{]} $:] the $\theta$-th key-trapdoor pair is sampled from $\Gen_{\calF}(1^{\lambda})$ and the remaining $2N-1$ pairs are all sampled from $\Gen_{\calG}(1^{\lambda})$.
    \item[$\theta = 0$:] all the key-trapdoor pairs are sampled from
    $\Gen_{\calG}(1^{\lambda})$.
    \item[$\theta = \diamond$:] all the key-trapdoor pairs are sampled from
    $\Gen_{\calF}(1^{\lambda})$.
\end{enumerate}
Send the keys $k = (k_1,\ldots,k_{2N})$ to the device.

\\

 2.    Receive $y = (y_1,\ldots,y_{2N})\in \calY^{2N}$ from the device.
    
\\

  3.  Sample round type ``preimage" or ``Hadamard" uniformly at random and send to the device.
  
\\

   \underline{Case ``preimage''}: receive 
   \begin{equation*}
   (b,x) = (b_1,\ldots,b_{2N}, x_1,\ldots,x_{2N})
   \end{equation*}
   from the device, where $b \in \{0,1\}^{2N}$ and $x\in \{0,1\}^{2Nw}$.

  If $\CHK(k_i,y_i,b_i,x_i)=0$ for all $i\in [2N]$, \textbf{accept}, else \textbf{reject}. 
  
\\

    \underline{Case ``Hadamard''}:
    receive 
    \begin{equation*}
        d=(d_1,\ldots,d_{2N}) \in \{0,1\}^{2Nw}
    \end{equation*}
    from the device.
    
\\

    4. With probability $1/2$, sample $q\leftarrow_U \{0^{2N}, 1^{2N}, 0^N1^N, 1^N0^N\}$ uniformly at random,
    and with probability $1/2$
    sample $q\leftarrow_{\mu} \{0,1\}^{2N}$ according to the distribution $\mu$. Send $q$ to the device.
    
\\

    \hspace{4.15mm} Receive $u\in \{0,1\}^{2N}$ from the device.
    \begin{enumerate}[\text{Case} A.,leftmargin=55pt]
     \item $\theta = 0$ and
     \begin{enumerate}[leftmargin=5pt]
         \item[if] $q_i = 0$ and $\hatb(k_i,y_i)\neq u_i$ for some $i \in [2N]$, \textbf{reject},
         \item[else] \textbf{accept}.
     \end{enumerate}
        \item $\theta \in [2N]$ and
        \begin{enumerate}[leftmargin=5pt]
            \item[if] $q_i = 0$ and $\hatb(k_i,y_i)\neq u_i$ for some $i\neq \theta$, \textbf{reject},
            \item[if] $q_\theta = 1$ and  $\hath(k_\theta,y_\theta,d_\theta) \oplus 
            \hatb(k_{\theta+N}, y_{\theta+N})\neq u_\theta$, \textbf{reject}, \item[else] \textbf{accept}.
        \end{enumerate}
        \item $\theta = \diamond$ and
        \begin{enumerate}[leftmargin=5pt]
            \item[if] $q_i = 0$, $q_{N+i} = 1$ and $u_i \oplus u_{N+i} \neq 
            \hat{h}(k_{N+i},y_{N+i}, d_{N+i})$ for some $i \in [N]$, \textbf{reject}, 
            \item[if] $q_i=1$, $q_{N+i} = 0$ and $u_i \oplus u_{N+i} \neq 
            \hat{h}(k_{i},y_{i}, d_{i})$ for some $i \in [N]$, \textbf{reject},
            \item[else] \textbf{accept}.
        \end{enumerate}
    \end{enumerate}
\\
    \hline
\end{tabular}
}
\caption{A protocol that self-tests EPR pairs of a computationally efficient device.}
\label{fig:protocol}
\end{figure}

\begin{thm}\label{thm:completeness}
There exists a quantum device using $2N\cdot \poly(\lambda)$ qubits and quantum gates that is accepted by our self-testing protocol with probability $\geq 1 - 2 N \cdot \negl$. Moreover, there exists a classical verifier that runs in $2N \cdot \poly(\lambda)$ time.
\end{thm}

For the rest of this paper, given input security parameter $\lambda\in \mathbb{N}$, we set $N = \poly(\lambda)$ to be a fixed polynomially bounded function of $\lambda$, so that both the device and the verifier above are efficient and the device passes the protocol with probability negligibly close to $1$.

\begin{proof}
In the first round, for each $i\in [2N]$, by the efficient range superposition property of ENTCFs (\cref{property:efficient_range_superposition}), the (honest quantum) device uses $k_i$ to prepare a state $\ket{\psi_i'}$ that is $\negl$-close to
\begin{equation}\label{eq:state}
    \ket{\psi_i} \coloneqq \frac{1}{\sqrt{2\cdot |\calX|}} \sum_{b\in\{0, 1\}}\sum_{x\in \calX, y\in \calY}\sqrt{(f_{k_i,b}(x))(y)}\ket{b}\ket{x}\ket{y},
\end{equation}
in trace distance, which uses $\poly(\lambda)$ qubits and quantum gates. Therefore, the device prepares the tensor product $\bigotimes_{i=1}^{2N}\ket{\psi_i'}$ using $2N\cdot\poly(\lambda)$ qubits and quantum gates. 

We have $\|\bigotimes_{i=1}^{2N} \ket{\psi_i}- \bigotimes_{i=1}^{2N} \ket{\psi_i'}\|_1\leq 2N\cdot \negl$. Therefore, the output distributions arising from all subsequent measurements made on $\bigotimes_{i=1}^{2N} \ket{\psi_i'}$ are the same as those made on $\bigotimes_{i=1}^{2N} \ket{\psi_i}$ up to $2N \cdot \negl$ in total variation distance. Therefore, for the rest of this proof, we can assume that the device has actually prepared $\ket{\psi}\coloneqq \bigotimes_{i=1}^{2N} \ket{\psi_i}$ and 
reduce the lower bound on the success probability by $2N\cdot \negl$ at the end.

Then, the device measures the (image) $y$ register of $\ket{\psi_i}$ and sends the outcome to the verifier. By the (disjoint) injective pair property of ENTCFs (\cref{property:injective_pair}), after the $y$ measurement, the state $\ket{\psi_i}$ collapses to $\ket{\phi_i}\ket{y_i}$, where
\begin{equation}
   \begin{aligned}
        \ket{\phi_i} \coloneqq \begin{cases}
            \ket{\hatb(k_i, y_i)}\ket{\hatx(k_i, y_i)}  &\text{ if } k_i \in \calK_\calG, \\
            \frac{1}{\sqrt{2}} (\ket{0}\ket{\hatx_0(k_{i}, y_i)}+\ket{1}\ket{\hatx_1(k_{i}, y_i)}) &\text{ if } k_{i} \in \calK_\calF.
        \end{cases}
\end{aligned} 
\end{equation}
In the following, we use the shorthand $\hatb_i\coloneqq \hatb(k_i,y_i) \in \{0,1\}$ and, for $a\in \{0,1\}$, $\hatx_{a,i}\coloneqq\hatx(a, k_i, y_i)\in \calX$.

In the second round, there are two cases, ``preimage'' or ``Hadamard''. In the ``preimage'' case, the device measures the $b$ and $x$ registers of each $\ket{\phi_i}$ in the computational basis and sends the outcome to the device. This will always be accepted by the device using the definition of $\CHK$.

In the ``Hadamard'' case, the device measures the $x$ register of each $\ket{\phi_i}$ in the Hadamard basis and sends the outcome $d=(d_1,d_2,\ldots,d_{2N})$ to the verifier. After this measurement, $\ket{\phi_i}$ collapses to $\ket{\alpha_i}\ket{d_i}$, where, if $\theta\in [2N]$, then
\begin{equation}
    \ket{\alpha_i} = 
    \begin{cases}
    \ket{\hatb_i} &\text{if $i\neq \theta$},
    \\
     (\ket{0} + (-1)^{d_\theta\cdot (\hatx_{0,\theta} \oplus \hatx_{1,\theta})}\ket{1})/\sqrt{2} &\text{if $i=\theta$};
    \end{cases}
\end{equation}
if $\theta=0$, then $\ket{\alpha_i} = \ket{\hatb_i}$; and if $\theta=\diamond$, then
$\ket{\alpha_i} = (\ket{0} + (-1)^{d_i\cdot (\hatx_{0,i} \oplus \hatx_{1,i})}\ket{1})/\sqrt{2}$.

In the following, we use the shorthand $\hath_i
\coloneqq d_i\cdot (\hatx_{0,i} \oplus \hatx_{1,i})\in \{0,1\}$ and 
\begin{equation}\label{eq:hardcore_bitstring}
    \hath' \coloneqq (\hath_{N+1},\ldots,\hath_{2N}, \hath_1,\hath_2,\ldots,\hath_{N})\in \{0,1\}^{2N}.
\end{equation}
For $v\in \{0,1\}^{2N}$, we also define the state
\begin{equation}\label{eq:def_psi}
    \ket{\psi^{v}} \coloneqq \frac{1}{\sqrt{2^N}}\bigotimes_{i=1}^N (\sigmax)^{v_i}\otimes (\sigmax)^{v_{N+i}}(\ket{0}_i\ket{+}_{N+i} + \ket{1}_i\ket{-}_{N+i}),
\end{equation}
which consists of $N$ (locally-rotated) EPR pairs.

Then, the device applies $N$ controlled-$\sigmaz$ gates between the $i$-th and $(N+i)$-th qubits of $\bigotimes_{i=1}^{2N}\ket{\alpha_i}$ for all $i\in [N]$ (note that the controlled-$\sigmaz$ gate is independent of which qubit is the control and which qubit is the target). The device has now prepared the $2N$-qubit state
\begin{equation}\label{eq:honest_final_state}
    \ket{\alpha}\coloneqq 
    \begin{cases}
    \ket{\hatb_1,\ldots,\hatb_{\theta-1}}
    \ket{(-)^{\hatb_{\theta+N} \oplus \hath_{\theta}}}
    \ket{\hatb_{\theta+1},\ldots,\hatb_{2N}} &\text{if $\theta\in [2N]$, $\theta \leq N$},\\
    \ket{\hatb_1,\ldots,\hatb_{\theta-1}}
    \ket{(-)^{\hatb_{\theta-N} \oplus \hath_{\theta}}}
    \ket{\hatb_{\theta+1},\ldots,\hatb_{2N}} &\text{if $\theta\in [2N]$, $\theta > N$},
    \\
    \ket{\hatb_1,\ldots,\hatb_{2N}} &\text{if $\theta=0$},\\
    \ket{\psi^{\hath'}}  &\text{if $\theta=\diamond$}.
    \end{cases}
\end{equation}

In the ``Hadamard'' case, there is a third and final round where the verifier sends a bitstring $q\in \{0,1\}^{2N}$ to the device. The device performs the following $q$-dependent measurements. For $i \in [2N]$, if $q_i = 0$, measure the $i$th qubit of $\ket{\alpha}$ in the computational basis, otherwise, measure the $i$th qubit of $\ket{\alpha}$ in the Hadamard basis.
The device finally sends the outcome $u\in \{0,1\}^{2N}$ of these measurements to the verifier. The expressions on the right-hand side of \cref{eq:honest_final_state} imply that the device passes the last checks made by the verifier with probability  $\geq 1 - 2N \cdot \negl$. (The only way the device fails these last checks is if $d_i = 0^w$ for some $i\in [2N]$, which happens with probability at most $2N \cdot \negl$).

The ``moreover'' part of the theorem follows directly from the efficient function generation and the efficient decoding properties of ENTCFs (\cref{property:efficient_generation,property:efficient_decoding}).
\end{proof}

\section{Soundness of self-testing protocol}
\label{sec:soundness}

In this section, we show that our self-testing protocol achieves $\poly(N,\fail)$ soundness error. Unlike the proof of completeness in \cref{sec:completeness}, we use the adaptive hardcore bit and injective invariance properties of ENTCFs to prove soundness in this section. Therefore, it is necessary for us to make the LWE hardness assumption throughout this section.

We start with \cref{sec:devices} where we mathematically model quantum devices. Essentially, a quantum device is a four-tuple $D=(S,M,\Pi,P)$ where $S$ is a set of states, $M$, $\Pi$, and $\{P_q \mid q\in \{0,1\}^{2N}\}$ are the measurements $D$ performs to obtain the $d$, $(b,x)$, and $u$ given question $q$, respectively, as defined in our protocol (\cref{fig:protocol}). We use $P_q$ to define observables of the quantum device called $X_{q,i}$ and $Z_{q,i}$ for $i\in [2N]$ that should act as $\sigmax_i$ and $\sigmaz_i$, respectively. (Recall that $\sigmax_i$ is the Pauli $X$ operator acting on qubit $i$ and $\sigmaz_i$ is the Pauli $Z$ operator acting on qubit $i$.) For different choices of $(\theta,q, i) \in (\thetaset) \times \{0,1\}^{2N}\times  [2N]$, our goal is to characterize the actions of the observables $X_{q,i}$ and $Z_{q,i}$ on the state $\sigma^\theta$, which is the post-$M$-measurement state. Unlike in nonlocal self-testing, where we only need to characterize one state, here we need to characterize multiple states and observables. Our strategy is to first decompose 
\begin{equation}
    \sigma^\theta \approx \sum_{v \in  \{0,1\}^{2N}} \sigma^{\theta,v},
\end{equation}
where $\sigma^{\theta,v}$ are (subnormalized) states that correspond to certain $(y,d)$ measurement outcomes. We then show that $X_{q,i}$ and $Z_{q,i}$ act on $\sigma^{\theta,v}$ similarly to how $\sigmax_i$ and $\sigmaz_i$ act on certain ideal states up to some error that can be bounded by the probability that the device fails our protocol for a given $(\theta,q)$. More specifically, the ideal state corresponding to $\sigma^{\theta,v}$ is $\tau^{\theta,v}$ as defined in \cref{def:tau_thetav}. In \cref{tab:testing_correspondence}, we 
 summarize which $(\theta,q)$ pairs are used to characterize which observables. For example, $(\theta, q) = (0,0^{2N})$ is used to characterize $Z_{0^{2N},i}$ for all $i\in [2N]$ while $(\theta, q) = (\diamond, 0^N1^N)$ is used to characterize $Z_{0^N1^N,i}\cdot X_{0^N1^N,i+N}$ for all $i\in [N]$.
\begin{table}[ht]
    \centering
    \renewcommand*{\arraystretch}{1.2}
    \begin{tabular}{c|c|c|c}
        $(q_i, q_{i+N})$ with $i \leq N$ & $\theta = 0$ &  $\theta \in [2N]$ &
        $\theta=\diamond$ \\
        \hline
         $(0,0)$ & $Z_{q,i}$ and $Z_{q,i+N} $ & - & -  \\
         $(1,1)$ & - & $X_{q,\theta}$ if $\theta \in \{i,i+N\}$ & - \\
        $(0,1)$ & - & - & $Z_{q,i} \cdot X_{q,i+N} $ \\
        $(1,0)$ & - & - & $X_{q,i} \cdot Z_{q,i+N} $ \\ 
    \end{tabular}
    \caption{Correspondence between the $(\theta,q)$ used in our protocol and the observables tested.}\label{tab:testing_correspondence}
\end{table}

Sampling $\theta\leftarrow_U \thetaset$ and $q$ as in our protocol allows us to effectively characterize the actions of $X_{q,i}$ and $Z_{q,i}$ on $\sigma^{\theta}$ for all $(\theta,q,i) \in (\thetaset) \times \{0^{2N},1^{2N}, 0^N1^N,1^N0^N\} \times [2N]$ using the overall failure probability of the device. If we naively sampled $\theta \leftarrow_U \{0,1\}^{2N}$, we would not be able to effectively bound the failure probability of the device given some $\theta$ by the overall failure probability of the device since the former probability can be smaller than the latter by a factor of $2^{2N}$.

Then, in \cref{prop:tilde_equal_nontilde}, we use the computational indistinguishability of elements in  $\{\sigma^\theta \mid \theta\in \thetaset\}$ to argue that, for all $q\in \{0,1\}^{2N}$, the observables $X_{q,i}$ and $Z_{q,i}$ act on $\sigma^\theta$ similarly to how $X_i \coloneqq X_{1^{2N},i}$ and $Z_i \coloneqq Z_{0^{2N},i}$ act on $\sigma^\theta$ for all $\theta \in\thetaset$. This allows us to restrict attention for most of the subsequent analysis to $X_i$ and $Z_i$ since results that hold for them automatically hold for $X_{q,i}$ and $Z_{q,i}$ via \cref{prop:tilde_equal_nontilde}.


For self-testing, we not only need to characterize the action of a single observable on $\sigma^\theta$ as sketched above but we also need to characterize the actions of \emph{products} of observables.
In \cref{sec:commute,sec:anti_commute}, we establish the observables' $\sigma^\theta$-state-dependent commutation and anti-commutation relations. These relations involve products of two observables. Our proof generalizes and refines techniques used in \cite{mv2021selftest,gheorghiuvidick2019rsp}. In particular, we need to define error parameters associated with each element in $\{\sigma^{\theta,v} \mid v\in \{0,1\}^{2N}\}$. We then use these error parameters collectively to bound the overall approximation error associated with $\sigma^\theta$.

In \cref{sec:conj_inv}, we prove observable-state commutation relations for certain pairs of observables and states. These relations are necessary for showing that products of \emph{more than} two observables, each being either $X_{q,i}$ or $Z_{q,i}$, satisfy the same state-dependent (anti-)commutation relations as $\sigmaz_i$ and $\sigmax_i$, respectively. Recall that the previous (anti-)commutation relations only concern products of two observables. However, we also want to show, for example,
\begin{equation}\label{eq:example_operator_state_commutation}
Z_1 X_3 Z_2 \sigma^3 \approx X_3 Z_1 Z_2 \sigma^3,
\end{equation}
for some appropriate approximation error. \cref{eq:example_operator_state_commutation} does not follow from the $\sigma^3$-state-dependent commutation relation between $Z_1$ and $X_3$ because $Z_1$ and $X_3$ are not directly next to $\sigma^3$ due to the obstructing $Z_2$. However, \cref{eq:example_operator_state_commutation} does follow once we first use an observable-state commutation relation to commute $Z_2$ past $\sigma^3$. We view our use of observable-state commutation relations as one of the main technical contributions of this work. These techniques should be independently useful in any future work on computational self-testing involving more than two qubits.

In \cref{sec:swap}, we give the formula of the swap isometry $\calV$ and then describe an efficient quantum circuit that implements it. Our swap isometry can be viewed as a special case of the swap isometry proposed in \cite[Figure 2]{yang2013robust} in the nonlocal setting. In particular, it is not the obvious generalization of the swap isometry used in \cite[Proof of Lemma 4.28]{mv2021selftest} which is more difficult to analyze. 

In \cref{sec:swap_observables,sec:swap_states}, we analyze the effect of the swap isometry on the observables and states of the device. More specifically, in \cref{sec:swap_observables}, we show that $\calV$ maps the $X_{q,i}$ and $Z_{q,i}$ observables approximately to $\sigmax_i$ and $\sigmaz_i$.
In \cref{sec:swap_states}, we show that $\calV$ maps the state $\sigma^{\theta,v}$ of the device to a state of the form $\tau^{\theta,v}\otimes \alpha^{\theta,v}$, where 
$\tau^{\theta,v}$ is the ideal state and 
$\alpha^{\theta,v}$ is some junk state such that $\alpha^{\theta,v}$ is close to being computationally indistinguishable from a fixed state $\alpha$ for all $\theta\in \thetaset$ and $v\in \{0,1\}^{2N}$.

Lastly, in \cref{sec:soundness_states_meas}, we put everything together to give our main soundness result, \cref{thm:soundness}. The main task is to characterize the measurement operator $P_q = \{P_q^u \mid u \in \{0,1\}^{2N}\}$, where each $P_q^u$ is approximately a product of $2N$ binary projectors of the form $Z_{q,i}^{(u_i)}$ and $X_{q,i}^{(u_i)}$. Recall that we have 
previously characterized $Z_{q,i}$ and $X_{q,i}$ individually in \cref{sec:swap_observables}. To characterize the products of their corresponding projectors, we use the operator-state commutation relation to sequentially replace each projector in the product by its ideal counterpart. 

The analysis in \cref{sec:swap_observables,sec:swap_states,sec:soundness_states_meas} crucially relies on results in
\cref{sec:commute,sec:anti_commute,sec:conj_inv} which allow us to bound the soundness error in \cref{thm:soundness} by $O(\poly(N, \epsilon))$.
If we directly generalized the soundness analysis of \cite{mv2021selftest}, we would obtain an extremely loose $O(2^N \epsilon^{1/2^N})$ bound on the soundness error.\footnote{See \cite[End of Section 1.3 (arXiv version v2)]{gmp2022parallelrsp} for an explanation of why the technique in \cite{mv2021selftest} would lead to such a loose bound.}

\subsection{Quantum devices}
\label{sec:devices}

In this subsection, we model a general quantum device
that can be used by a device in our protocol specified in \cref{sec:completeness}.
Our definition is based on the definition given in
\cite[Section 4.1]{mv2021selftest}.

\begin{defn}\label{def:device}
    A device $D = (S, M, \Pi, P)$ is specified by Hilbert spaces named $\calH_D$, $\calH_Y$, and $\calH_R$, with $\dim(\calH_Y) = \abs{\calY}^{2N}$ and $\dim(\calH_R) = 2^{2Nw}$, and the following.
    \begin{enumerate}
        \item A set $S \coloneqq \{ \psi^\theta \mid \theta \in \thetaset \} \subset \calD(\calH_D \otimes \calH_Y)$ of states where each state $\psi^\theta$ is classical on $\calH_Y$:
        \begin{align}
            \psi^{\theta} \coloneqq \sum_{y \in \calY^{2N}} \psi^{\theta}_{y} \otimes \ketbra{y}{y}.
        \end{align}
        The state $\psi^\theta_y$ models the device's state immediately after returning $y \in \calY^{2N}$ to the verifier if the verifier initially sampled $\theta \in \thetaset$. More precisely, $\psi^{\theta}_y$ (and hence $\psi^{\theta}$) is a function of the public keys $k\in (\calK_{\calF} \cup \calK_{\calG})^{2N}$ that the verifier sampled according to $\theta$, as described in the protocol. We choose to make the $k$-dependence implicit for notational convenience.  

        \item A projective measurement $\Pi$
        for the preimage test
        on $\calH_D \otimes \calH_Y$:
        \begin{align}
            \Pi \coloneqq \Biggl\{\Pi^{b,x} \coloneqq \sum_{y\in \calY^{2N}} 
            \Pi_y^{b,x} \otimes \ketbra{y}{y} \Biggm| b \in \{0,1\}^{2N}, x \in \calX^{2N} \Biggr\}.
        \end{align}
        The measurement outcome $b,x$ is the device's 
        answer for the preimage test.

        \item A projective measurement $M$ on $\calH_D \otimes \calH_Y$ for the device's first answer in the Hadamard test:
        \begin{align}
            M \coloneqq \Biggl\{M^d \coloneqq \sum_{y\in \calY^{2N}} M_y^d \otimes \ketbrasame{y} \Biggm| d \in \{0,1\}^{2Nw}\Biggr\}.
        \end{align}
        We write $\sigma^\theta(D)$ for the classical-quantum state that results from measuring $M$ on $\psi^{\theta}$ followed by writing measurement outcome $d$ into another classical register whose Hilbert space is denoted by $\calH_R$. That is,
        \begin{equation}
            \sigma^\theta(D) \coloneqq \sum_{y\in \calY^{2N}, \, d\in \{0,1\}^{2Nw}} \sigma_{y,d}^{\theta}(D) \otimes \ketbra{y,d}{y,d} \in \calH_D\otimes\calH_Y\otimes\calH_R,
        \end{equation}
        where $\sigma_{y,d}^\theta(D) \coloneqq M_y^{d} \psi^{\theta}_y M_y^{d}$.
        \item Projective measurements $P_q$ on $\calH_D \otimes \calH_Y \otimes \calH_R$ for the device's second answer in the Hadamard test when asked question $q \in \{0,1\}^{2N}$:
        \begin{align}
            P_q \coloneqq 
            \Biggl\{ P_q^u = \sum_{y\in \calY^{2N}, d\in \{0,1\}^{2Nw}} P_{q, y, d}^{u} \otimes \ketbra{y,d}{y,d} \Biggm| u \in \{0,1\}^{2N} \Biggr\}.
        \end{align}
        The measurement outcome $v$ is the device's answer for the question $q$.
    \end{enumerate}
\end{defn}
\begin{remark}
    We stress that the states $\psi^{\theta}_y$ depend on the public keys $k\in (\calK_{\calF} \cup \calK_{\calG})^{2N}$ that the verifier sampled according to $\theta$. In particular, all subsequent states of the device also depend on $k$. When we make a statement about such $k$-dependent states labeled by $\theta$, that statement is often understood as holding in expectation over $k$ sampled according to $\theta$. The expectation over $k$ is necessary when the statement is derived using the adaptive hardcore bit or injective invariance properties of ENTCFs, which only hold in expectation over $k$.
\end{remark}

Henceforth, unqualified sums over each of the symbols $b,x,y,d$ always refer to sums over $b\in \{0,1\}^{2N}$, $x\in \calX^{2N}$, $y\in\calY^{2N}$, and $d\in \{0,1\}^{2Nw}$ respectively, unless otherwise stated.

In this work, we focus on efficient quantum devices, which are defined below.
\begin{defn}
    A device $D = (S, \Pi, M, P)$ is efficient
    if all the states in $S$ can be
    efficiently prepared and all the measurements
    $\Pi, M$, and $P$ are efficient.
\end{defn}

As in the nonlocal self-testing, we want to show that each projector $P_q^u$ behaves like a tensor product of projectors on $2N$ systems. Therefore, we define the marginal observables on each of those systems as done in \cite[Definition 4.4]{mv2021selftest}.
\begin{defn}[Marginal observables]\label{def:observables}
Let $D = (S, \Pi, M, P)$ be a device. For $i\in [2N]$ and $q \in \{0,1\}^{2N}$, we define the binary observables
\begin{equation}
\begin{aligned}
    &Z_{q,i}(D) \coloneqq \sum_v (-1)^{v_i} P_{q}^v \quad \text{if } q_i = 0 \quad \text{and}\\
    &X_{q,i}(D) \coloneqq \sum_v (-1)^{v_i} P_{q}^v \quad \text{if } q_i = 1.
\end{aligned}
\end{equation}
For $i\in [2N]$, $y\in \calY^{2N}$, and $d\in \{0,1\}^{2Nw}$, we also define the binary observables
\begin{equation}
    \begin{aligned}
        &Z_{q, i,y,d}(D) \coloneqq \sum_v (-1)^{v_i} P_{q,y,d}^v \quad \text{if } q_i = 0 \quad \text{and}\\
        &X_{q, i,y,d}(D) \coloneqq \sum_v (-1)^{v_i} P_{q,y,d}^v \quad \text{if } q_i = 1.
    \end{aligned}
\end{equation}
Note that $Z_{q,j}(D)$ commutes with $X_{q,k}(D)$ for $j \neq k$ according to these definitions.
\end{defn}
In later parts of this section, we will repeatedly work with $Z_{0^{2N}, i}(D)$, $Z_{0^{2N}, i, y,d}(D)$, $X_{1^{2N}, i}(D)$, and $X_{1^{2N}, i,y,d}(D)$, where $i\in [2N]$, $y\in \calY^{2N}$, and $d\in \{0,1\}^{2Nw}$. Therefore, we define the abbreviations
\begin{equation}
\label{eq:XZ_operator}
    \begin{alignedat}{4}
    &Z_i(D) &&\coloneqq Z_{0^{2N},i}(D),  \quad &&Z_{i,y,d}(D) &&\coloneqq Z_{0^{2N},i,y,d}(D),
    \\
    &X_i(D) &&\coloneqq X_{1^{2N},i}(D),  \quad &&X_{i,y,d}(D) &&\coloneqq X_{1^{2N},i,y,d}(D).
\end{alignedat}
\end{equation}

We also define the abbreviations
\begin{equation}
    \label{eq:tilde_operator}
    \begin{aligned}
    \tildeZ_i(D) &\coloneqq Z_{0^N1^N,i}(D) &&\text{for all $i\in [N]$},
    \\
    \tildeZ_i(D) &\coloneqq Z_{0^N1^N,i}(D) &&\text{for all $i\in \{N+1,\dots,2N\}$},
    \\
    \tildeX_i(D) &\coloneqq X_{1^N0^N,i}(D) &&\text{for all $i\in [N]$},
    \\
    \tildeX_i(D) &\coloneqq X_{0^N1^N,i}(D) &&\text{for all $i\in \{N+1,\dots,2N\}$}.
\end{aligned}
\end{equation}

Our soundness analysis will characterize the states $\sigma^{\theta,v}$ in the following definition as the states that we are self-testing. In the following definition, for $i\in [2N]$, we write
\begin{equation}
    \modulo(i + N, 2N) \coloneqq \begin{cases}
        i + N &\text{if $i \leq N$},
        \\
        i - N &\text{if $i > N$}.
    \end{cases}
\end{equation}

\begin{defn}[Hadamard round post-$d$-measurement states $\sigma^{\theta,v}$]\label{def:sig_thetav}
Let $D$ be a device. For $\theta\in \thetaset$ and $v\in \{0,1\}^{2N}$, we define the state
\begin{equation}
\sigma^{\theta, v}(D) \coloneqq \sum_{(y,d)\in \Sigmaset(\theta,v)}\sigma_{y,d}^{\theta}(D) \otimes \ket{y,d}\bra{y,d} \in \calH_D\otimes \calH_Y \otimes \calH_R,
\end{equation}
where, 
\begin{equation}
    \Sigmaset(\theta,v)
    \coloneqq
    \begin{cases}
    \bigl\{(y,d)\bigm| \hatb(k_i,y_i) = v_i \text{ for all } i \neq \theta \text{ and } \hath(k_\theta, y_{\theta}, d_{\theta}) = v_{\theta} \oplus v_{\modulo(\theta+N, 2N)} \bigr\} &\text{if $\theta\in [2N]$},
    \\
    \bigl\{(y,d) \bigm| \hatb(k_i,y_i) = v_i \text{ for all } i\bigr\} &\text{if $\theta=0$},
    \\
    \bigl\{(y,d) \bigm| \hath(k_i,y_i,d_i) = v_{\modulo(i+N,2N)} \text{ for all } i \bigr\} &\text{if $\theta=\diamond$.}
    \end{cases}
\end{equation}
In all cases, $(y,d)$ ranges over $\calY^{2N}\times \{0,1\}^{2Nw}$, $i$ ranges over $[2N]$, and the state $\sigma^{\theta, v}(D)$ implicitly depend on keys $k \in (\calK_{\calF}\cup \calK_{\calG})^{2N}$ chosen according to $\theta$ as described in the protocol.
\end{defn}

Note that $\sum_v \sigma^{\theta, v}(D) \leq \sigma^\theta(D)$ by definition because the $\Sigmaset(\theta,v)$s partition a subset of $\calY^{2N}\times \{0,1\}^{2Nw}$. In particular, taking traces on both sides gives $\sum_v\tr[\sigma^{\theta,v}(D)]\leq 1$. In the honest case,
\begin{equation}
    \tr_{Y,R}[\sigma^{\theta, v}(D)]
    = \begin{cases}
    2^{-2N} \ketbrasame{v_1 \ldots v_{\theta-1} (-)^{v_\theta} v_{\theta+1}
    \ldots v_{2N}} & \text{for $\theta\in [2N]$,} 
    \\
    2^{-2N} \ketbrasame{v} & \text{for $\theta=0$,} 
    \\
    2^{-2N} \ketbrasame{\psi^{v}} & \text{for $\theta=\diamond$,}
    \end{cases}
\end{equation}
where $\ket{v}\coloneqq \ket{v_1 v_2 \ldots v_{2N}}$ and we recall the definition of $\ket{\psi^v}$ from \cref{eq:def_psi}.

In the soundness proof, we will use quantities called $\gammaP$ and $\gammaH$ to bound how far away the $Z_{q,i}(D), X_{q,i}(D)$ observables and $\sigma^{\theta,v}$ states are from the self-tested observables and states. We define them below.
\begin{defn}[$\gammaP$ and $\gammaH$]\label{def:gamma}
Let $D$ be a device. We define the following quantities that all relate to the failure probabilities of $D$.

\begin{enumerate}
\item \textbf{Preimage test.} For $\theta\in \thetaset$, we define
\begin{align}
\label{eq:t_theta}
t_{\theta}(D)\coloneqq \tr\Bigl[\sum_{y} \, \sum_{(b,x)\in \Preimage(k,y)}\Pi^{b,x}_y \psi^\theta_y\Bigr] = \tr\Bigl[\sum_{y,b,x \mid \CHK(k,y,b,x)=0}\Pi_y^{b,x}\psi_y^{\theta}\Bigr],
\end{align}
where an implicit expectation is taken over the keys $k\in (\calK_{\calF}\cup \calK_{\calG})^{2N}$ that are sampled according to $\theta$ as described in the protocol, and
\begin{equation}
\Preimage(\theta,y) \coloneqq 
\{(b,x) \in \{0,1\}^{2N}\times \calX^{2N} \,|\, b_i = \hatb(k_i,y_i), x_i =\hatx(b_i,k_i,y_i) \text{ for all } i \}.
\end{equation}
Note that the set $\Preimage(\theta,y)$ can be empty due to $\hatb$ or $\hatx$ returning $\perp$; in that case, $\sum_{(b,x)\in \Preimage(\theta,y)} \Pi^{b,x}_y$ is taken to mean $0$. We then define
\begin{equation}
    \gammaP(D) \coloneqq 1- \min \{t_{\theta}(D) \mid \theta\in \thetaset\}.
\end{equation}

\item \textbf{Hadamard test.} For $\theta \in \thetasetint$, $q \in \{0,1\}^{2N}$ and $i \in [2N]$, we define
\begin{alignat}{2}
\label{eq:r_theta} &r_{\theta,q,i}(D) \coloneqq \begin{cases}
    \tr\Bigl[\sum_v Z_{q,i}^{(v_i)} \sigma^{\theta,v}\Bigr]  &\text{if } q_i = 0 \\
    1 &\text{otherwise}
\end{cases},
\\
\label{eq:s_theta} &s_{\theta,q,i}(D) \coloneqq 
\begin{cases}
    \tr\Bigl[\sum_v X_{q,i}^{(v_i)} \sigma^{\theta,v}\Bigr] & \text{if } q_i = 1 \\
    1 & \text{otherwise}
\end{cases}.
\end{alignat}
For $\theta = \diamond$ and $i\in [N]$, we define
\begin{align}
\label{eq:r_diamond} & r_{\diamond, q, i}(D) \coloneqq \begin{cases}
    \sum_v \tr \Bigl[(Z_{q,i} X_{q, N+i})^{(v_i)}\sigma^{\diamond,v} \Bigr]
&\text{if } q_i = 0 \text{ and } q_{N+i} = 1 \\
1 & \text{otherwise}
\end{cases},\\
\label{eq:s_diamond}  & s_{\diamond,q, i}(D) \coloneqq \begin{cases}
    \sum_v \tr \Bigl[(X_{q,i} Z_{q, N+i})^{(v_{N+i})}\sigma^{\diamond,v} \Bigr]
&\text{if } q_i = 1 \text{ and } q_{N+i} = 0 \\
1 & \text{ otherwise}
\end{cases}.
\end{align}
We then define 
\begin{equation}\label{eq:gammaTq}
     \gammaHq(D) \coloneqq 1 - \min \{\min\{r_{\theta,q,i}, s_{\theta,q,i} \mid \theta \in  \thetasetint, i \in [2N]\}, \min\{r_{\diamond,q,i}, s_{\diamond,q,i} \mid i \in [N]\} \}.
\end{equation}
\end{enumerate}
\end{defn}

Since our soundness proof proceeds by first characterizing the
operators $X_i(D), Z_i(D), \tilde{X}_i(D)$ and $\tilde{Z}_i(D)$ (defined in \cref{eq:XZ_operator,eq:tilde_operator}),
we write
\begin{equation}\label{eq:gammaT_convenience}
    \gammaH(D) \coloneqq \max\{\gammaHq(D) \mid q \in \{0^{2N}, 1^{2N}, 0^N1^N, 1^N0^N\}\}.
\end{equation}
Henceforth, we assume\footnote{This assumption holds when $D$ fails our protocol with sufficiently small probability -- see \cref{prop:gamma_bound_by_fail}.} that $\gammaP(D), \gammaH(D) < 1$ only for notational convenience so that we can simplify, for example, $O(\gammaH(D) + \gamma(D)^r)$ to $O(\gamma(D)^r)$ when $0<r<1$. For the same reason, we assume $N\gammaH(D)$ is non-negligible in $\lambda$ so that we can simplify, for example, $O(N\gammaH(D) + \negl)$ to $O(N\gammaH(D))$.

We define the following failure probabilities of a device, which can be estimated by running the protocol multiple times.
\begin{defn}[Failure probabilities]
\label{def:failure_prob}
Let $D$ be a device. For $q\in \{0,1\}^{2N}$, we define 
\begin{align}
    &\failP(D) \coloneqq \pr(D \text{ fails preimage test} \mid \text{ case: preimage}),\\
    &\failHq(D) \coloneqq \pr(D \text{ fails Hadamard test} \mid \text{ case: Hadamard and question } q).
\end{align}
Then, 
\begin{equation}
    \fail(D) \coloneqq \failP(D)/2 + \Bigl(\sum_{q \in \{0^{2N},1^{2N},0^N1^N, 1^N0^N\}} \frac{1}{4} \failHq(D) + \sum_{q\in \{0,1\}^{2N}} \mu(q)\failHq(D)\Bigr)/4.
\end{equation}
is the average failure probability.
\end{defn}
Henceforth, when $D$ is clear from the context,
we mostly drop the $D$-dependence from the quantities
\begin{equation}
    \sigma^\theta, \sigma^\theta_{y,d}, \sigma^{\theta,v}, Z_i, X_i,\tildeZ_i,\tildeX_i, Z_{i,y,d}, X_{i,y,d}, Z_{q,i}, X_{q,i}, t_{\theta}, r_{\theta,q,i},s_{\theta,q,i},\gammaHq,\gammaH, \failP, \failHq, \fail.
\end{equation}

While it is easier to bound the soundness error using $\gammaP,\gammaHq$ and $\gammaH$, they 
are not immediately observable to the verifier. However, we can bound them by the observable failure probabilities, $\failP, \failHq, \text{ and }\fail$, as follows.

\begin{proposition}[$\gamma$ bounded by failure probability $\epsilon$]\label{prop:gamma_bound_by_fail}
Let $D$ be a device. Then, for all $q \in \{0,1\}^{2N}$,
\begin{align}
\label{eq:gammaPleq}\gammaP \leq (2N&+2) \failP \leq 2(2N+2)\fail, \\
\label{eq:gammaTqleq}
\gammaHq &\leq (2N+2) \failHq.
\end{align}
Moreover,
\begin{align} 
    \label{eq:gammaTsubsetq}
    \gammaH \coloneqq \max\{\gammaHq \mid q \in \{0^{2N}, 1^{2N}, 0^N1^N, 1^N0^N\}\} \leq 16(2N+2) \fail.
\end{align}

\end{proposition}

\begin{proof}
As $t_{\theta}\leq 1$ for all $\theta$, \cref{eq:gammaPleq} follows from 
\begin{equation}
\failP = 1 - \frac{1}{2N+2}\Bigl(t_{\diamond} + \sum_{\theta=0}^{2N} t_{\theta} \Bigr) \geq 1-\frac{1}{2N+2}(2N + 1 + \min\{t_{\theta} \mid \theta\in \thetaset \}) = \frac{\gammaP}{2N+2}.
\end{equation}

For $\theta \in \thetaset$ and $q\in \{0,1\}^{2N}$, let $p_{\theta,q}$ be the success probability of the device when the verifier chooses $\theta$ and asks question $q$. Since $\Pr(\medcap_i A_i)\leq \min_i\{\pr(A_i)\}$ for any events $A_i$, we obtain
\begin{align}
    &p_{0,q} \leq \min\{r_{0,q,i} \mid i \in [2N]\}, \\
    &p_{\theta,q} \leq \min\{ r_{\theta,q,i}, s_{\theta,q,\theta} \mid i \in [2N]\} \quad \text{ for } \theta \in [2N], \\
    &p_{\diamond,q} \leq \min\set{r_{\diamond,q,i}, s_{\diamond,q,i} \mid i \in [N]}.
\end{align}
Recall that
\begin{align}
    \gammaHq &= 1 - \min \set{\min\set{r_{\theta,q,i}, s_{\theta,q,i} \mid \theta \in \{0\} \cup [2N], i \in [2N]}, \min\set{r_{\diamond,q,i}, s_{\diamond,q,i} \mid i \in [N]}  } \\
    &= 1 - \min \set{p_{\theta, q} \mid \theta \in \thetaset}.
\end{align}
Then
\begin{align*}
    1 - \failHq = \frac{1}{2N+2} \sum_{\theta \in \thetaset} 
    p_{\theta,q} 
    \leq \frac{2N+1}{2N+2} + \frac{\min\set{p_{\theta,q} \mid \theta \in \thetaset}}{2N+2}, 
\end{align*}
from which \cref{eq:gammaTqleq} follows.

Because in the Hadamard test, with probability $1/2$, $q$ is chosen from $\set{0^{2N}, 1^{2N}, 0^N1^N, 1^N0^N}$ uniformly at random, we have
\begin{align*}
    \frac{1}{4}(\fail_{H,0^{2N}} + \fail_{H,1^{2N}} + \fail_{H,0^{N}1^N} + \fail_{H,1^N0^N}) \leq 4\fail.
\end{align*}
That is, for $q \in \set{0^{2N}, 1^{2N}, 0^N1^N, 1^N0^N}$, $\failHq \leq 16\fail$.
Then \cref{eq:gammaTsubsetq} follows from \cref{eq:gammaTqleq}.
\end{proof}

\begin{defn}\label{def:zeta_chi_def}
For $\theta \in \thetasetint$, $q\in \{0,1\}^{2N}$, $i\in [2N]$, and $v\in \{0,1\}^{2N}$, we define
\begin{align}
    \zeta(\theta, q, i, v) \coloneqq \normstate{Z_{q,i} - (-1)^{v_i}I}{\sigma^{\theta,v}}^2 && \text{if $q_i = 0$,}
    \\
    \chi(\theta, q, i, v) \coloneqq \normstate{X_{q,i} - (-1)^{v_i}I}{\sigma^{\theta,v}}^2 && \text{if $q_i = 1$}.
\end{align}
For $q\in \{0,1\}^{2N}$, $i\in [2N]$, and $v\in \{0,1\}^{2N}$, we define
\begin{align}
    \zetadiamond(q, i, v) &\coloneqq \normstate{Z_{q,i}X_{q,N+i} - (-1)^{v_i}I}{\sigma^{\diamond,v}}^2 && \text{if $q_i = 0$ and $q_{N+i} = 1$},
    \\
    \chidiamond(q, i, v) &\coloneqq \normstate{X_{q,i}Z_{q,N+i} - (-1)^{v_{N+i}}I}{\sigma^{\diamond,v}}^2 && \text{if $q_i = 1$ and $q_{N+i} = 0$}.
\end{align}
\end{defn}
The values $\{\zeta,\zetadiamond,\chi,\chidiamond\}$ in the definition above can be viewed as being of order $\gammaHq/2^{2N}$ (see \cref{lem:zeta_chi_bounds_new}). We will later bound the errors of various approximate operator-state relations
by sums of $2^{2N}$ terms, where each is of the form $\zeta$, $\zetadiamond$, $\chi$, or $\chidiamond$.

For convenience, we say an equation involving $\zeta(\theta,q,i,v)$, $\chi(\theta,q,i,v)$, $\zetadiamond(q,i,v)$, or $\chidiamond(q,i,v)$ holds for all ``appropriate'' $(\theta,q,i,v)$ to mean that the equation holds for all $(\theta,q,i,v)$ in the appropriate domain as specified in \cref{def:zeta_chi_def}.

For later convenience, for all appropriate $(\theta,i,v)$, we define
\begin{alignat}{4}\label{eq:zetachi_convenience}
&\zeta(\theta,i,v) &&\coloneqq \zeta(\theta,0^{2N},i,v), 
\quad &&\chi(\theta,v) &&\coloneqq \chi(\theta,1^{2N},\theta,v), 
\\
&\zetadiamond(i,v) &&\coloneqq \zetadiamond(0^N1^N,i,v),
\quad &&\chidiamond(i,v) &&\coloneqq \chidiamond(1^N0^N,i,v).
\end{alignat}

By the definition of $\simeq$, for all appropriate $(q,i,\theta,v)$, we have
\begin{equation}
\begin{alignedat}{4}
&Z_{q,i} &&\simeq_{\zeta(\theta,q,i,v),\sigma^{\theta,v}} (-1)^{v_i}I  &&\quad \text{and} \quad 
X_{q,i} &&\simeq_{\chi(\theta,q,i,v),\sigma^{\theta,v}} (-1)^{v_i}I,
\\
&Z_{q,i}X_{q,N+i} &&\simeq_{\zetadiamond(q,i,v),\sigma^{\diamond,v}} (-1)^{v_i}I  &&\quad \text{and} \quad 
X_{q,i}Z_{q,N+i} &&\simeq_{\chidiamond(q,i,v),\sigma^{\diamond,v}} (-1)^{v_{N+i}}I.
\end{alignedat}
\end{equation}
In addition, by elementary algebra, we can express
\begin{equation}\label{eq:zeta_chi_as_norm}
\begin{alignedat}{3}
&\zeta(\theta,q,i,v)/4 &&= \tr[\sigma^{\theta,v}] - \tr[Z_{q,i}^{(v_i)}\sigma^{\theta,v}] &&= \norm{Z_{q,i}^{(v_i)}-\1}^2_{\sigma^{\theta,v}},
\\
& \chi(\theta,q,i,v)/4 &&= \tr[\sigma^{\theta,v}] - \tr[X_{q,i}^{(v_i)}\sigma^{\theta,v}] &&= \norm{X_{q,i}^{(v_i)}-\1}^2_{\sigma^{\theta,v}},
\\
&\zetadiamond(q,i,v)/4 &&= \tr[\sigma^{\diamond,v}] - \tr[(Z_{q,i} X_{q,N+i})^{(v_i)}\sigma^{\diamond,v}] &&= \norm{(Z_{q,i} X_{q,N+i})^{(v_i)}-\1}^2_{\sigma^{\diamond,v}},
\\
&\chidiamond(q,i,v)/4 &&= \tr[\sigma^{\diamond,v}] - \tr[(X_{q,i} Z_{q,N+i})^{(v_{N+i})}\sigma^{\diamond,v}] &&= \norm{(X_{q,i} Z_{q,N+i})^{(v_{N+i})}-\1}^2_{\sigma^{\diamond,v}}.
\end{alignedat}
\end{equation}

The lemma below follows directly from \cref{def:observables,def:gamma}.
\begin{lemma}\label{lem:sandwich_argument}

For $\theta \in \thetasetint$, $q \in \{0,1\}^{2N}$, and $i \in [2N]$ such that $q_i = 0$ and $i \neq \theta$, we have
\begin{equation}
    1-\gammaHq \leq r_{\theta,q,i} = \sum_{v} \tr[Z_{q,i}^{(v_i)} \sigma^{\theta,v}] \leq \sum_{v} \tr[\sigma^{\theta,v}] \leq \tr[\sigma^{\theta}] = 1.
\end{equation}

For $\theta \in [2N]$ and $q \in \{0,1\}^{2N}$ such that $q_\theta = 1$,
we have
\begin{equation}
1-\gammaHq \leq s_{\theta,q,\theta} = \sum_{v} \tr[X_{q,\theta}^{(v_\theta)} \sigma^{\theta,v}] \leq \sum_{v} \tr[\sigma^{\theta,v}] \leq \tr[\sigma^{\theta}] = 1.
\end{equation}

For $q \in \{0,1\}^{2N}$ and $i\in [N]$ such that $q_i = 0$ and $q_{N+i} = 1$, we have
\begin{alignat}{6}
&1-\gammaHq &&\leq r_{\diamond,q,i} &&= \sum_{v} \tr[(Z_{q,i}X_{q,N+i})^{(v_i)} \sigma^{\diamond,v}] &&\leq \sum_{v} \tr[\sigma^{\theta,v}] &&\leq \tr[\sigma^{\theta}] &&= 1.
\end{alignat}

For $q \in \{0,1\}^{2N}$ and $i\in [N]$ such that $q_i = 1$ and $q_{N+i}=0$, we have
\begin{alignat}{6}
&1-\gammaHq &&\leq s_{\diamond,q,i} &&= \sum_{v} \tr[(X_{q,i}Z_{q,N+i})^{(v_{N+i})} \sigma^{\diamond,v}] &&\leq \sum_{v} \tr[\sigma^{\theta,v}] &&\leq \tr[\sigma^{\theta}] &&= 1.
\end{alignat}

\end{lemma}

The following lemma bounds exponential sums of $\zeta(\theta,q,i,v)$, $\chi(\theta,q,\theta,v)$, $\chidiamond(q,i,v)$, and $\zetadiamond(q,i,v)$ over $v\in \{0,1\}^{2N}$ using $\gammaHq$.

\begin{lemma}\label{lem:zeta_chi_bounds_new}
For all appropriate $(\theta,q,i,v)$ with $i \neq \theta$, we have
\begin{equation}
    \sum_{v} \zeta(\theta,q,i,v)  \leq 4 \gammaHq. 
\end{equation}

For all appropriate $(\theta,q,i,v)$ with $i=\theta$, we have
\begin{equation}
    \sum_v \chi(\theta,q,i,v) \leq 4\gammaHq. 
\end{equation}

For all appropriate $(q,i,v)$, we have
\begin{equation}
    \sum_v \zetadiamond(q,i,v) \leq 4\gammaHq \quad \text{and} \quad  \sum_v \chidiamond(q,i,v) \leq 4\gammaHq.
\end{equation}

In particular, we have
\begin{equation}
    \sum_v \zeta(\theta,i,v) \leq 4\gammaH, \quad \sum_v \chi(\theta,v) \leq 4\gammaH, \quad \sum_v \zetadiamond(i,v) \leq 4\gammaH, \quad \text{and} \quad \sum_v \chidiamond(i,v) \leq 4\gammaH.
\end{equation}
\end{lemma}

\begin{proof}
Consider the first inequality. Fix $\theta\in \thetasetint$, $q\in \{0,1\}^{2N}$, and $i\in [2N]$ with $q_i = 0$ and $i\neq\theta$.  Summing the expression for $\zeta(\theta,q,i,v)$ in \cref{eq:zeta_chi_as_norm} over $v\in \{0,1\}^{2N}$ and using \cref{lem:sandwich_argument}, we obtain
\begin{equation}
\sum_{v}\zeta(\theta,q,i,v) \leq 4  - 4 \sum_{v} \tr[Z_{q,i}^{(v_i)}\sigma^{\theta,v}] \leq 4 \gammaHq.
\end{equation}
The argument is analogous for the remaining three inequalities before the ``in particular'' part. 

The ``in particular'' part follows from the definitions of $\zeta(\theta,i,v)$, $\chi(\theta,v)$, $\zetadiamond(i,v)$, and $\chidiamond(i,v)$ in  \cref{eq:zetachi_convenience} together with the definition of $\gammaH$ in \cref{eq:gammaT_convenience}.
\end{proof}

Note that the above lemmas control certain parameters of a device using its failure probability. But we have yet to leverage our computational assumptions. The following lemma, which we will frequently use later together with the lifting lemmas (\cref{lem:lifting,lem:lifting_proj}), allows us to use the injective invariance property of ENTCFs (\cref{property:injective_invariance}) to further control the device.

\begin{lemma}[Indistinguishability of $\{\psi^{\theta}\}_{\theta}$ $(\mathrm{resp.}~\{\sigma^{\theta}\}_{\theta})$]~\label{lem:indistinguishability} Any pair of states in $\{\psi^{\theta}\}_{\theta\in \thetaset}$ \newline $(\mathrm{resp.}~\{\sigma^{\theta}\}_{\theta\in \thetaset})$ of an efficient device $D$ are computationally indistinguishable.
\end{lemma}

\begin{proof}
For $v\in \{0,1\}^{2N}$, let $\calD(v) \coloneqq (k_1,k_2,\ldots,k_{2N})$ be the distribution on $2N$-tuples of public keys such that each $k_i$
is independently distributed according to $\Gen_{\calG}(1^{\lambda})_{\key}$ if $v_i=0$ or $\Gen_{\calF}(1^{\lambda})_{\key}$ if $v_i=1$.

For $\theta\in \thetaset$, let
\begin{equation}
\calD^{\theta} \coloneqq 
\begin{cases}
    \calD(\str(\theta)) &\text{if $\theta\in [2N]$,}
    \\
    \calD(0^{2N}) &\text{if $\theta = 0$,}
    \\
    \calD(1^{2N}) &\text{if $\theta = \diamond$,}
\end{cases}
\end{equation}
where $\str(\theta) \in \{0,1\}^{2N}$ is the $2N$-bit string with a $1$ at position $\theta$ and $0$s elsewhere.

The injective invariance property of ENTCFs (\cref{property:injective_invariance}) states that the distributions $\Gen_{\calG}(1^{\lambda})_{\key}$ and $\Gen_{\calF}(1^{\lambda})_{\key}$ are computationally indistinguishable. Therefore, it is clear that for any $u,v\in\{0,1\}^{2N}$, with $u,v$ differing by exactly one bit, $\calD(u)$ and $\calD(v)$ are computationally indistinguishable. Therefore, any pair of distributions in $\{\calD^{\theta}\}_{\theta\in\thetaset}$ are computationally indistinguishable up to $2N \, \negl = \negl$ by \cref{lem:ci_triangle} and  our setting $N = \lambda$. This is because, for any $\theta, \theta'\in \thetasetint$, $\str(\theta)$ and $\str(\theta')$ differ by at most $2$ bits while $\str(\theta)$ and $1^{2N}$ differ by at most $2N$ bits. As $D$ can prepare the state $\psi^{\theta}$ (resp.~$\sigma^{\theta}$) efficiently given keys drawn from $\calD^{\theta}$, all pairs of states in $\{\psi^{\theta}\}_{\theta \in \thetaset}$  (resp.~$\{\sigma^{\theta}\}_{\theta \in \thetaset}$) must be computationally indistinguishable.
\end{proof}

In later parts of the soundness proof, we use $\gamma_P$, $\gammaHq$ and $\gammaH$ to characterize $\sigma^\theta$-dependent operator relations. 
First, we show $\sum_v\sigma^{\theta,v}$ is $\gammaP$-close to $\sigma^\theta$ so that in later proofs we can replace
$\sigma^\theta$ by $\sum_v\sigma^{\theta,v}$, which is easier to analyze---see \cref{def:zeta_chi_def} and the lemmas following it.

\begin{defn}[Valid $y$]\label{def:valid_y}
    Let $y\in \calY^{2N}$ and $k \in (\calK_\calF \cup \calK_\calG)^{2N}$. We say $y$ is valid (with respect to $k$) if $\hatb(k_i,y_i) \neq \perp$ for all $i\in [2N]$; equivalently, 
    \begin{equation}
        y_i \in \bigcup_x \bigl(\Supp(f_{k_i,0}(x)) \cup \Supp(f_{k_i,1}(x))\bigr), \quad \text{for all $i\in [2N]$.}
    \end{equation}
    Otherwise, we say $y$ is invalid. 
\end{defn}

\begin{lemma}\label{lem:sum_sigma_v}
    Let $D$ be a device. For all $\theta \in \thetaset$, $\norm{\sigma^\theta - \sum_v \sigma^{\theta,v}}_1 \leq \gammaP$. 
\end{lemma}

\begin{proof}
    We have
    \begin{equation}
    \begin{aligned}
    \bignorm{\sigma^\theta - \sum_v \sigma^{\theta, v}}_1 
    = & \, \bignorm{ \sum_{\text{invalid } y} 
    \; \sum_d M_y^d \psi_y^\theta M_y^d \otimes \ketbrasame{y,d} }_1 \\
    = & \, \Bigl| \tr\Bigl[ \sum_{\text{invalid } y}   
    \; \sum_d M_y^d \psi_y^\theta M_y^d \otimes \ketbrasame{y,d}\Bigr] \Bigr| \\
    =& \, \sum_{\text{invalid } y}  \tr[\psi_y^\theta] \leq 1-t_{\theta} \leq \gammaP,
\end{aligned}
\end{equation}
where the last two inequalities follow from the definitions of $t_\theta$ and $\gammaP$ respectively.
\end{proof}

\cref{lem:zeta_chi_bounds_new,lem:indistinguishability,lem:sum_sigma_v} imply that, on the state $
\sigma^{\theta}$, $Z_{q,i}$ is close to $Z_i \coloneqq Z_{0^{2N},i}$ and $X_{q,i}$ is close to $X_i\coloneqq X_{1^{2N},i}$. More precisely:
\begin{proposition}
\label{prop:tilde_equal_nontilde}
For all $q\in \{0,1\}^{2N}$, $\theta\in [2N]\cup \{0,\diamond\}$, and $i\in[2N]$, we have
\begin{equation}
\begin{aligned}
    Z_{q,i} \simeq_{8\gammaHq + 8\gammaH + 4 \gammaP, \sigma^\theta} Z_{i} && \text{if $q_i=0$,}
    \\
    X_{q,i} \simeq_{8\gammaHq + 8\gammaH + 4 \gammaP, \sigma^\theta} X_{i} && \text{if $q_i=1$.}
\end{aligned}
\end{equation}
\end{proposition}
\begin{proof}
We first prove the equation involving $Z$. It suffices to prove this equation for $\theta=0$ because of the computational indistinguishability of the $\sigma^\theta$s (\cref{lem:indistinguishability}) and part 2 of the lifting lemma (\cref{lem:lifting}). Let $q\in \{0,1\}^{2N}$ and $i\in [2N]$ be such that $q_i=0$. Then, we have
\begin{equation}
\begin{aligned}
    \norm{Z_{q,i} - Z_i}^2_{\sigma^{0,v}} 
    &\leq 
    (\norm{Z_{q,i} - (-1)^{v_i}I}_{\sigma^{0,v}} + \norm{Z_i - (-1)^{v_i}I}_{\sigma^{0,v}})^2
    \\
    &\leq 2 \zeta(0,q,i,v) + 2\zeta(0,0^{2N},i,v).
\end{aligned}
\end{equation}

Therefore, by \cref{lem:op_approx_compts} and the bounds on $\zeta$ in \cref{lem:zeta_chi_bounds_new}, which is applicable because $i\neq 0$, we obtain
\begin{equation}\label{eq:tilde_equal_nontilde_sumv_new}
    Z_{q,i}\simeq_{8\gammaHq + 8\gamma_{H,0^{2N}}, \sum_v \sigma^{0,v}} Z_i.
\end{equation}

Now, using \cref{lem:sum_sigma_v}, we obtain
\begin{equation}
    \Bigl| \tr[(Z_{q,i}-Z_i)^2\sigma^0] - \tr[(Z_{q,i}-Z_i)^2\sum_v\sigma^{0,v}] \Bigr| \leq \norminfty{Z_{q,i}-Z_i}^2 \cdot \gammaP \leq 4 \gammaP.
\end{equation}
Therefore, by \cref{eq:tilde_equal_nontilde_sumv_new}, we have
\begin{equation}
    \tr[(Z_{q,i}-Z_i)^2\sigma^0] \leq  \tr[(Z_{q,i} - Z_i)^2\sum_v\sigma^{0,v}] + 4\gammaP \leq 8\gammaHq + 8\gamma_{H,0^{2N}} + 4 \gammaP,
\end{equation}
which, noting $\gamma_{H,0^{2N}}\leq \gammaH$ by definition, completes the proof of the first equation of the lemma. 

Now consider the second equation of the lemma involving $X$.  It again suffices to only consider $\theta=i$ by the reasoning at the start. The proof is then analogous to the proof of the first equation, except we use the bounds on $\chi$ in \cref{lem:zeta_chi_bounds_new}, which is applicable because $\theta=i$. We omit the details. 
\end{proof}

\subsection{Reduction to perfect device}
\label{sec:reduction}

In this subsection, we follow the strategy of \cite[Lemma 4.13]{mv2021selftest} by first showing
that any efficient device with $\gammaP(D)<1$ is close to a ``perfect'' efficient device. A perfect device is one that, for any $\theta\in \thetaset$ chosen by the verifier, can always pass the preimage test except with $\negl$ small probability.
Then, we restrict attention to efficient perfect devices for the rest of the soundness proof before \cref{sec:soundness_states_meas}.  

\begin{defn}[Perfect device]\label{def:perfect_device}
Let $D = (S,\Pi,M,P)$ be an device. We say $D$ is perfect if $\gammaP(D) = \negl$.
\end{defn}

\begin{proposition}
\label{prop:reduce_to_pd}
Let $D = (S, \Pi, M, P)$ be an efficient device with $\gammaP(D)<1$ and $S = \{\psi^\theta \mid \theta\in \thetaset \}$.
Then, there exists an efficient perfect device $\tildeD = (\tildeS, \Pi, M, P)$ which uses the same measurements $\Pi, M, P$ as $D$ and has states $\tildeS = \{ \tildepsi^\theta \mid \theta\in \thetaset \}$ that satisfy: for all $\theta\in \thetaset$,
\begin{equation}
        \norm{\psi^\theta - \tildepsi^\theta}_1 \leq \sqrt{\gamma_P(D)}.
\end{equation}
\end{proposition}

\begin{proof}
    $\tildeD$ can efficiently prepare each state $\tildepsi^\theta$ in $\tildeS$ as follows. $\tildeD$ first follows $D$ to prepare $\psi^\theta$ using a given set of public keys $k \in (\calK_{\calF}\cup \calK_{\calG})^{2N}$ sampled according to $\theta$ as described in the protocol. $\tildeD$ then applies the efficient unitary $U_\Pi$ associated with the efficient measurement $\Pi$, as per \cref{def:efficient_operators}, to create the state
    \begin{equation}\label{eq:unitary_Pi}
        \phi^{\theta} \coloneqq  U_\Pi (\ketbrasame{0_{2N+2Nw}}_{\anc} \otimes \psi^{\theta}) U_{\Pi}^{\dagger} = \sum_{y,b,x,b',x'} \ket{y,b,x}\bra{y,b',x'} \otimes 
        \Pi^{b,x}_y \psi_y^\theta \Pi^{b',x'}_y.   
    \end{equation}

    On this state, $\tildeD$ evaluates the
    $\CHK$ function---which is efficient given only $k$ and not the trapdoor $t_k$, see the efficient decoding property of ENTCFs (\cref{property:efficient_decoding})---on the classical register holding  $(y, b, x)$ to create the state
    \begin{equation}\label{eq:chk_phitheta}
        \sum_{y, b, x, b', x'}
        \ket{\CHK(k,y,b,x)}\bra{\CHK(k,y,b',x')} \otimes
        \ket{y,b,x}\bra{y,b',x'} \otimes 
        \Pi^{b,x}_y \psi_y^\theta \Pi^{b',x'}_y.
    \end{equation}
    $\tildeD$ next measures the first (single-bit) register of the state in \cref{eq:chk_phitheta}. 
    The probability that the measurement outcome is $0$ is equal to $\tr[\Lambda\phi^{\theta}]$, where we write $\Lambda$ for the projector
    \begin{equation}
        \Lambda \coloneqq \sum_y \sum_{b,x \mid \CHK(k,b,x,y) = 0 } \ketbrasame{y,b,x}.
    \end{equation}
    Note that $\tr[\Lambda\phi^{\theta}] = t_{\theta}$, where we recall $t_{\theta}$ from \cref{def:gamma}.
    
    If the measurement outcome is $0$, $\tildeD$ continues by applying $U_{\Pi}^{\dagger}$ to the post-measurement states, and then traces out the $\anc$ register to output the state
    \begin{equation}
        \tildepsi^{\theta} \coloneqq \tr_{\anc} \Bigg[U_\Pi^\dagger \, \frac{\Lambda \phi^{\theta} \Lambda}{\tr[\Lambda\phi^{\theta}]}
       \, U_\Pi \Bigg].
    \end{equation}
    But, using \cref{eq:unitary_Pi}, we can write
    \begin{equation}
        \psi^{\theta} = \tr_{\anc} \Big[ U_{\Pi}^{\dagger} \phi^{\theta} U_{\Pi} \Big].
    \end{equation}
    Therefore, since the trace norm cannot increase under partial trace and is unitarily invariant, we obtain
    \begin{equation}
        \normone{\tildepsi^{\theta} - \psi^{\theta}} \leq \bignorm{\frac{\Lambda \phi^{\theta} \Lambda}{\tr[\Lambda\phi^{\theta}]} - \phi^{\theta}}_1 \leq 2\sqrt{1-\tr[\Lambda \phi^{\theta}]} = \sqrt{1-t_{\theta}} \leq \sqrt{\gammaP},
    \end{equation}
    where we used the gentle measurement lemma in the second inequality, see, e.g.,~\cite[Lemma 9.4.1 in arXiv v8]{wilde2017textbook}.

    Now, $\tildeD$ might not obtain $0$ when it measures the first register of the state in \cref{eq:chk_phitheta}. In this case, it repeats the above procedure up to (a fixed) $\poly(\lambda)$ times, stopping and outputting $\tildepsi^{\theta}$ the first time $0$ is measured and aborting if $0$ is never measured within those $\poly(\lambda)$ repeats. But, the probability of aborting is $(1-t_{\theta})^{\poly(\lambda)}\leq \gammaP^{\poly(\lambda)} = \negl$ as required.
\end{proof}

We finish this subsection by proving several lemmas that hold for efficient perfect devices. \cref{lem:sigma_equal_sumv,lem:collapse} are used implicitly in \cite{mv2021selftest}. The other \cref{lem:uniformity_sigma_marginals,lem:trace_xsigma_zero} are generalizations of explicit lemmas in \cite{mv2021selftest} and provide a precise reference in their proofs.
\begin{lemma}\label{lem:sigma_equal_sumv}
Let $D$ be a perfect device. For all $\theta\in \thetaset$, we have  
\begin{equation}
\sigma^\theta\simeq_{\negl} \sum_v \sigma^{\theta,v}.
\end{equation} 
\end{lemma}
\begin{proof}
The lemma follows from \cref{lem:sum_sigma_v} and the definition of a perfect device (\cref{def:perfect_device}).
\end{proof}

\begin{lemma}[Collapsing property]\label{lem:collapse}
Let $D$ be an efficient perfect device. For all $i \in [2N]$ and $\theta \in \thetaset$, we have
\begin{align*}
    \psi^\theta \csimeq_{\negl} \sum_{b,x} \Pi^{b,x}
    \psi^\theta \Pi^{b,x}.
\end{align*}
\end{lemma}

\begin{proof}
    For $c\in\{0,1\}$, let
    $\calY_{i,c} \coloneqq \{y_i \in \calY \mid 
    \hatb(k_i,y_i) = c\}$. Note that $y \in \calY^{2N}$ is valid (see \cref{def:valid_y}) if and only if $y_i \in \calY_{i,0} \cup \calY_{i,1}$ for all $i \in [2N]$.

    We first prove the lemma in the case $\theta = 0$, so $k\in \calK_\calG^{2N}$. As the device is perfect and $\theta=0$, with probability $\geq 1- \negl$, we have 
    \begin{equation}\label{eq:perfect_psiys}
        \psi_y^0 =
        \begin{cases}
        0 & \text{for all invalid $y\in \calY^{2N}$,}
        \\
        \Pi_{y}^{\hatb(k,y), \hatx(k,y)} \psi^0_y \Pi_{y}^{\hatb(k,y), \hatx(k,y)}
        & \text{for all valid $y \in \calY^{2N}$.}
        \end{cases}
    \end{equation}
    Therefore, with probability $\geq 1 - \negl$, we have
    \begin{align*}
        \sum_{b,x} \Pi^{b,x}
    \psi^0 \Pi^{b,x}
        &=  \sum_{b,x,y} \Pi^{b,x}_y
        \psi^0_y \Pi^{b,x}_y\otimes \ketbrasame{y}
        \\
        &= \sum_{\text{valid } y} \Pi_{y}^{\hatb(k,y), \hatx(k,y)}
        \psi^0_y \Pi_{y}^{\hatb(k,y), \hatx(k,y)} \otimes \ketbrasame{y}
        \\
        &= \sum_{\text{valid } y}\psi_y^0 \otimes \ketbrasame{y} = \sum_{y}\psi_y^0 \otimes \ketbrasame{y} = \psi^0.
    \end{align*}

    When $\theta \neq 0$, the computational indistinguishability of the $\psi^{\theta}$s (\cref{lem:indistinguishability}) implies
    \begin{align}
        \sum_{b,x}
        \Pi^{b,x} \psi^\theta \Pi^{b,x}
        \csimeq_{\negl}
        \sum_{b,x}
        \Pi^{b,x} \psi^0 \Pi^{b,x}
        \simeq_{\negl} \psi^0 \csimeq_{\negl}
        \psi^\theta.
    \end{align}
    The conclusion follows by the triangle inequality for computational indistinguishability (\cref{lem:ci_triangle}).
\end{proof}

\begin{lemma}\label{lem:uniformity_sigma_marginals}
Let $D$ be an efficient perfect device. For all $\theta\in[2N]$, we have
\begin{equation}\label{eq:uniformity_sigma_marginals}
\sum_{v}(-1)^{v_{\theta}} \tr[\sigma^{\theta,v}] \simeq_{\negl} 0.
\end{equation}
\end{lemma}

\begin{proof} The proof is similar to \cite[Lemma~4.15]{mv2021selftest}. It suffices to prove the lemma for $\theta=1$ as the proof for other $\theta$ is analogous. Assume for contradiction that \cref{eq:uniformity_sigma_marginals} is not true. Then, there exists a non-negligible function $\notnegl$ such that, for infinitely many $\lambda\in \mathbb{N}$, we have
\begin{equation}\label{eq:uniformity_sigma_marginals_nonnegl}
\Bigl| \sum_{v | v_1=0}\tr[\sigma^{1,v}] - \sum_{v | v_1 = 1}\tr[\sigma^{1,v}] \Bigr|> \notnegl.
\end{equation}
We also assume that, for a fixed $\lambda$, the term inside the absolute value sign in \cref{eq:uniformity_sigma_marginals_nonnegl} is $\geq 0$. The proof is analogous if otherwise.

Consider the following efficient algorithm $\calA$ which breaks the adaptive hardcore bit property of ENTCFs (\cref{property:adaptive_hardcore_bit}):
\begin{enumerate}
\item Input: a key $k_1$ sampled from $\Gen_{\calF}(1^{\lambda})_{\key}$.
\item  Independently sample $2N-1$ keys $k_2,\ldots,k_{2N}$,  each from $\Gen_{\calG}(1^{\lambda})_{\key}$. (Let $\calD^1$ denote the distribution on the keys sampled in the first two steps of $\calA$.)
\item Prepare the state $\psi^{1}$, then make measurement obtaining $y$, then make preimage measurement obtaining $(b,x)$, and then make Hadamard measurement obtaining $d$. 
\item Output: $(y_1, b_1,x_1,d_1)$ and print ``$f_{k_1,b_1}(x_1) = y_1 \text{ and } \hath(k_1,y_1,d_1) = 0$''.
\end{enumerate}

Because $D$ is perfect, we have, with probability $\geq 1 - \negl$,
\begin{equation}\label{eq:uniformity_sigma_marginals_preimage}
f_{k_1,b_1}(x_1) = y_1.
\end{equation}

Let $\calA'$ be the same as $\calA$ except that it does not make a preimage measurement and produces no output. By the collapsing property (\cref{lem:collapse}), the state of $\calA$ immediately after the preimage measurement is computationally indistinguishable from its state just before. Therefore, by definition, the state of $\calA$ after the Hadamard measurement is computationally indistinguishable from the state of $\calA'$ after the Hadamard measurement. But note that $\calA'$ behaves exactly like $D$ up to when the Hadamard measurement is made in the Hadamard case. Therefore, the state of $\calA$ after the Hadamard measurement is computationally indistinguishable from the state of $D$ after the Hadamard measurement in the Hadamard case.

Now, using the definitions of $\sigma^{1,v}$ and $\Sigmaset(1,v)$ from \cref{def:sig_thetav}, we see that \cref{eq:uniformity_sigma_marginals_nonnegl} means
\begin{equation}
\pr_{k\sim \calD^1, D}\bigl((y,d)\in \medcup_{v|v_1=0}\Sigmaset(1,v)\bigr) > \pr_{k\sim \calD^1, D}\bigl((y,d)\in \medcup_{v|v_1=1}\Sigmaset(1,v)\bigr) + \notnegl.
\end{equation}
Then, as argued above, we can replace $D$ by $\calA$ by introducing a negligible term as follows
\begin{equation}
\pr_{k\sim \calD^1, \calA}\bigl((y,d)\in \medcup_{v|v_1=0}\Sigmaset(1,v)\bigr) > \pr_{k\sim \calD^1, \calA}\bigl((y,d)\in \medcup_{v|v_1=1}\Sigmaset(1,v)\bigr) + \notnegl + \negl.
\end{equation}
Now, we can rewrite the probabilities appearing on the left- and right-hand sides to see that
\begin{equation}
\begin{aligned}
&\pr\bigl(\hath(k_1,y_1,d_1) = 0 \text{ and } \medcap_{i=2}^{2N} (\hatb(k_i,y_i)\in\{0,1\})\bigr) \\
&> \pr\bigl(\hath(k_1,y_1,d_1) = 1 \text{ and } \medcap_{i=2}^{2N} (\hatb(k_i,y_i)\in\{0,1\})\bigr) + \notnegl + \negl.  
\end{aligned} 
\end{equation}
But 
\begin{equation}
    \pr\bigl(\medcap_{i=2}^{2N} \, (\hatb(k_i,y_i)\in\{0,1\})\bigr) \geq 1 - \negl,
\end{equation}
by the definition of $\hatb$ and the fact $D$ is perfect. Therefore, as $\pr(A\cap B ) \geq \pr(A)+\pr(B) - 1$ for any events $A$ and $B$, we have
\begin{equation}~\label{eq:uniformity_sigma_marginals_generalized_bit}
\pr(\hath(k_1,y_1,d_1) = 0) > \pr(\hath(k_1,y_1,d_1) = 1) + \notnegl + \negl.
\end{equation}
\cref{eq:uniformity_sigma_marginals_preimage,eq:uniformity_sigma_marginals_generalized_bit} together contradict the adaptive hardcore bit property of ENTCFs (\cref{property:adaptive_hardcore_bit}).
\end{proof}

\begin{lemma}\label{lem:trace_xsigma_zero}
Let $D$ be an efficient perfect device. For all $i,\theta\in [2N]$, we have
\begin{equation}
    \tr[X_i\sigma^{\theta}] \simeq_{2\gammaH + \negl} 0.
\end{equation}
\end{lemma}
\begin{proof}
The proof is similar to \cite[Corollary 4.16]{mv2021selftest}. Because $D$ is efficient, $X_i$ is an efficient binary observable. Therefore, it suffices to consider $\theta = i$ by using the computational indistinguishability of the $\sigma^\theta$s (\cref{lem:indistinguishability}) and part $1$ of the lifting lemma (\cref{lem:lifting}). Moreover, we only consider $\theta = i = 1$ as the argument is analogous otherwise. Then, the lemma follows from
\begin{equation}
\tr[X_1\sigma^1] = \sum_v \tr[(-1)^{v_1} (2 X_{1}^{(v_1)} - \1)\sigma^{1, v}] = \sum_v (-1)^{v_1} \tr[\sigma^{1,v}]-\frac{1}{2}\sum_{v}(-1)^{v_1}\chi(1,v) \simeq_{2\gammaH} \negl,
\end{equation}
where the middle equality uses \cref{eq:zeta_chi_as_norm} and the last approximation uses \cref{lem:zeta_chi_bounds_new,lem:uniformity_sigma_marginals}.
\end{proof}

\subsection{Commutation relations}
\label{sec:commute}

\begin{proposition}\label{prop:commutation}
Let $D$ be an efficient perfect device. For all $i,j,\theta\in [2N]$ with $i\neq j$, we have
\begin{equation}\label{eq:commutation}
    [Z_i, Z_j] = 0, \quad
    [X_i, X_j] = 0, \quad \mathrm{and} \quad
    [Z_i, X_j] \approx_{\gammaH+\negl,\sigma^{\theta}} 0.
\end{equation}
\end{proposition}

Our proof is similar to that of \cite[Proposition 4.24]{mv2021selftest} except it is more refined in that we use $\zeta_v$ and $\chi_v$ to bound various $\sigma^{\theta,v}$-dependent approximation errors instead of $\gammaH$. For example, in \cref{eq:zeta_chi_example}, We cannot use $\gammaH$ to bound the errors as it would lead to a $O(2^{2N}\gamma_H)$ approximation error in \cref{eq:commutation} due to \cref{eq:zeta_chi_sum} where a sum over $v\in \{0,1\}^{2N}$ is taken.

\begin{proof}
The first two equations follow directly from \cref{def:observables}. Consider the last equation. Because $D$ is efficient, $Z_i$ and $X_j$ are efficient binary observables. Therefore, by the computational indistinguishability of the $\sigma^\theta$s (\cref{lem:indistinguishability}) and part $3$ of the lifting lemma (\cref{lem:lifting}), it suffices to prove the last equation for $\theta=j$. It also suffices to consider $i=1$ and $j=\theta=2$ as the proof is analogous for other $i,j,\theta$ with $i\neq j$. 

For the purposes of this proof, we write $\zeta_v$ for $\zeta(2,1,v)$ and $\chi_v$ for $\chi(2,v)$ for convenience. By definitions, we have
$Z_{1}\approx_{\zeta_{v},\sigma^{2,v}} (-1)^{v_1}I$ and $X_{2}\approx_{\chi_{v},\sigma^{2,v}} (-1)^{v_2}I$. Therefore,
\begin{equation}\label{eq:zeta_chi_example}
    Z_1X_2 \approx_{\chi_v, \sigma^{2,v}} (-1)^{v_2} Z_1 \approx_{\zeta_v,\sigma^{2,v}} (-1)^{v_2+v_1}I\approx_{\chi_v,\sigma^{2,v}} X_2 \,  (-1)^{v_1} \approx_{\zeta_v,\sigma^{2,v}} X_2 Z_1.
\end{equation}
Therefore, by the triangle inequality, we have 
\begin{equation*}
    Z_1X_2  \approx_{\chi_v + \zeta_v, \sigma^{2,v}} X_2 Z_1.
\end{equation*}
By linearity of the trace we can sum the above equation up over $v \in \{0,1\}^{2N}$ to deduce
\begin{equation}\label{eq:zeta_chi_sum}
Z_1X_2  \approx_{\sum_v{\chi_v + \zeta_v}, \sum_{v}\sigma^{2,v}} X_2 Z_1.
\end{equation}
Using $\sum_v \chi_v\leq 4\gammaH$ and $\sum_v \zeta_v\leq 4\gammaH$ from \cref{lem:zeta_chi_bounds_new}, and using $\sigma^2 \simeq_{\negl} \sum_{v}\sigma^{2,v}$ from \cref{lem:sigma_equal_sumv}, which uses the perfectness of $D$, we obtain 
\begin{equation}
    Z_1X_2  \approx_{\gammaH+\negl, \sigma^{2}} X_2 Z_1,
\end{equation}
which completes the proof.
\end{proof}

\subsection{Anti-commutation relations}
\label{sec:anti_commute}

We prove \cref{prop:anticommutation} in this subsection. The results and their proofs in this subsection are generalizations of those in \cite[Section 4.6]{mv2021selftest}. 
More specifically, \cite{mv2021selftest} deals with the two-qubit case and proves results for $\sigma^{(\theta_1,\theta_2)}$ with $(\theta_1,\theta_2) \in \{0,1\}^2$.
The naive generalization to the $2N$-qubit case would consider $(\theta_1, \ldots, \theta_{2N}) \in \{0,1\}^{2N}$. However, such generalization would not allow us to effectively control $\gammaH$ using the failure probability $\failHq$, since \cref{prop:gamma_bound_by_fail} would give the exponentially loose bound $\gammaH \leq 2^{2N} \failH$. To resolve this problem, we observe that it suffices to consider
\begin{equation}
(\theta_1, \ldots, \theta_{2N}) \in
\{z \in \{0,1\}^{2N} \mid \text{Hamming weight of } z \text{ is at most } 1 \} \cup \{1^{2N}\},
\end{equation}
to prove \cref{prop:anticommutation}. This is why we choose $\theta \in \thetaset$. Recall from \cref{fig:protocol} that our $\theta = 0$ corresponds to $(\theta_1,\dots, \theta_{2N}) = 0^{2N}$, our
$\theta \in [2N]$ corresponds to the string with $1$ at index $\theta$ and $0$ everywhere else,
and our $\theta = \diamond$ corresponds to $1^{2N}$.

\begin{proposition}\label{prop:anticommutation}
Let $D$ be an efficient perfect device. For all $i,\theta\in[2N]$, we have
\begin{equation}
    \{Z_i,X_i\}\approx_{\sqrt{\gammaH}+\negl, \sigma^\theta}0.
\end{equation}
\end{proposition}

We need a few lemmas first.

\begin{lemma}\label{lem:anticommutation_xzsigmaz}
Let $D$ be an efficient device. For all $i,\theta \in [2N]$, we have
\begin{equation}\label{eq:anticommutation_xzsigmaz}
    \sum_{a}\tr[X_i Z_i^{(a)}\sigma^{\theta}Z_{i}^{(a)}]\approx_{\sqrt{\gammaH}+\negl} 0.
\end{equation}
\end{lemma}

\begin{proof}
The proof generalizes that of \cite[Lemma 4.18]{mv2021selftest} and is more refined due to the use of $\zeta$. 

We prove the lemma for $i=1$ as the proof is analogous for other $i$. Now, the states in $\{\sum_{a}Z_1^{(a)}\sigma^{\theta} Z_1^{(a)} \mid \theta\in [2N]\}$ are computationally indistinguishable because they can be efficiently prepared from $\sigma^{\theta}$ by measuring $Z_1$ ($Z_1$ is efficient because $D$ is efficient) and the $\sigma^\theta$s are computationally indistinguishable (\cref{lem:indistinguishability}). In addition, because $D$ is efficient, $X_1$ is an efficient binary observable. Therefore, by part $1$ of the lifting lemma (\cref{lem:lifting}), it suffices to prove the current lemma for $\theta = 2$.

For $a\in \{0,1\}$, we define
\begin{equation}
\begin{aligned}
\sigma(a) \coloneqq \sum_{v | v_1=a} \sigma^{2, v} \quad \text{and} \quad
\zeta(a) \coloneqq \sum_{v | v_1=a} \zeta(2,1,v).
\end{aligned}
\end{equation}
Note that $\sigma^2=\sigma(0)+\sigma(1)$ by definition and $\zeta(0)+\zeta(1)\le 4\gammaH$ by \cref{lem:zeta_chi_bounds_new}.

By \cref{eq:zeta_chi_as_norm}, we have
\begin{equation}
  \zeta(a)/4=\tr[\sigma(a)] - \tr[Z_1^{(a)}\sigma(a)],
\end{equation}
which can be re-expressed as
\begin{equation}
 Z_1^{(a)} \approx_{\zeta(a),\sigma(a)}I, \quad Z_1^{(\overline{a})} \approx_{\zeta(a),\sigma(a)} 0 .
\end{equation}
Therefore, by \cref{lem:operator_to_state}, we have
\begin{equation}
Z_1^{(a)}\sigma(a) Z_1^{(a)} \approx_{\zeta(a)} \sigma(a), \quad Z_1^{(\overline{a})}\sigma(a) Z_1^{(\overline{a})}\approx_{\zeta(a)} 0.
\end{equation}
Therefore, by the triangle inequality (recall that the definition of $\approx$ for states involves a square), we have
\begin{equation}
\sum_{a,v} Z_1^{(a)}\sigma^{2,v} Z_1^{(a)}=\sum_a Z_1^{(a)} (\sigma(0)+\sigma(1))Z_1^{(a)}\approx_{\left(\sqrt{\zeta(0)}+\sqrt{\zeta(1)}\right)^2}\sigma^2.
\end{equation}
Recalling $\zeta(0)+\zeta(1)\le 4\gammaH$ by \cref{lem:zeta_chi_bounds_new}, the above equation means
\begin{equation}
\sum_{a,v} Z_1^{(a)}\sigma^{2,v} Z_1^{(a)}\approx_{\gammaH}\sigma^2.
\end{equation}
Finally, using the above equation and \cref{lem:trace_xsigma_zero}, we obtain
\begin{equation}
\sum_a \tr[X_1 Z_1^{(a)}\sigma^2Z_1^{(a)}] =\tr\Bigl[X_1\sum_{a,v}Z_1^{(a)}\sigma^{2, v}Z_1^{(a)}\Bigr] \approx_{\sqrt{\gammaH}} \tr[X_1\sigma^2] \approx_{\gammaH} 0,
\end{equation}
which completes the proof.
\end{proof}

To state and prove the next lemma, we need the following definition.
\begin{defn}
Let $D$ be a device. For keys $k\in (\calK_{\calF} \cup \calK_{\calG})^{2N}$, $b\in \{0,1\}^{2N}$, $y\in \calY^{2N}\textbf{}$, we define $\hatx(b,k,y)\in \calX^{2N}$ component-wise by  $\hatx(b,k,y)_i\coloneqq \hat{x}(b_i,k_i,y_i)$ for all $i\in [2N]$. Then, we define
\begin{equation}\label{eq:hatpi_by}
\hatPi^{b}_{y} \coloneqq \Pi^{b,\hatx(b,k,y)}_y \quad \text{and} \quad \hatPi^{b} \coloneqq \sum_y \hatPi^{b}_y\otimes \ketbra{y}{y},
\end{equation}
which can be thought of as projectors onto the correct preimage answers. 
Note that if $\hatx(b_i, k_i, y_i) = \perp$ for some $i\in [2N]$, then $\hatPi_y^b = 0$.
Moreover, for $a\in \{0,1\}$ and $i\in [2N]$,  we define
\begin{equation}
\hatPi^{i, a}_{y} \coloneqq \sum_{b|b_i = a} \hatPi^{b}_{y}, \quad
\hatPi^{i, a} \coloneqq \sum_{y} \hatPi^{i,a}_{y}\otimes \ketbra{y}{y},
\quad \text{and} \quad
Z^{(a)}_{i,d} \coloneqq \sum_y Z^{(a)}_{i,y,d} \otimes \ketbra{y}{y}.
\end{equation}
\end{defn}

We also note that, by previous definitions, we have
\begin{equation}
    Z_{i,y,d}^{(a)} = \sum_{v \mid v_i = a} P_{0^{2N},y,d}^v \quad \text{and} \quad 
    Z_i^{(a)} = \sum_d Z^{(a)}_{i,d} \otimes \ketbrasame{d}.
\end{equation}
\begin{lemma}
\label{claim:nomatch_hatpi}
    Let $D$ be a device. Let $\theta\in \thetasetint$
    and $b, b'\in \{0,1\}^{2N}$ be such that $b_j \neq b'_j$
    for some $j \in [2N]$ with $j \neq \theta$. Then, $\hatPi^b_y \psi^\theta_y \hatPi^{b'}_y = 0$, for all $y \in \calY^{2N}$.
\end{lemma}

\begin{proof}
    This proof generalizes that of \cite[Lemma 4.20]{mv2021selftest}.
    We consider the following three cases, which cover all possibilities:
    \begin{enumerate}
        \item $\hatb(k_j,y_j) = b_j$. In this case, we have $y_j\in \medcup_{x\in \calX} \Supp(f_{k_j,b_j}(x))$. Now, $k_j\in \calK_\calG$ as $j\neq \theta$. Therefore, $b_j'\neq b_j$ implies $y_j\notin \medcup_{x\in \calX} \Supp(f_{k_j,b_j'}(x))$. So $\hatx(b_j',k_j,y_j) = \perp$ and $\hatPi^{b'}_y=0$.
        \item $\hatb(k_j,y_j) = b_j'$. In this case, by an analogous argument, we obtain $\hatPi^{b}_y=0$.
        \item $\hatb(k_j,y_j) = \perp$. In this case, we have $\hatx(b_j,k_j,y_j) = \hatx(b_j',k_j,y_j) = \perp$. So $\hatPi^{b}_y=\hatPi^{b'}_y=0$.
    \end{enumerate}
    In all cases, we see  $\hatPi^b_y \psi^\theta_y \hatPi^{b'}_y =0$. Hence the lemma. 
\end{proof}

\begin{lemma}
\label{lem:anticommutation_mpi_zm}
Let $D$ be an efficient perfect device. For all $i,\theta\in[2N]$, we have
\begin{equation}
\label{eq:z_and_Pi}
\sum_{a,y,d} \| M_{y}^{d}\hatPi^{i,a}_y - Z^{(a)}_{i,y,d} M_{y}^{d}\|^2_{\psi^{\theta}_{y}} \simeq_{2\gammaH + \negl} 0.
\end{equation}
\end{lemma}

\begin{proof}
The proof of this lemma generalizes that of \cite[Lemma 4.19]{mv2021selftest}.

We prove the lemma for $i=1$ as the proof for other $i$ is analogous. Expanding the left-hand side of \cref{eq:z_and_Pi}, we obtain
\begin{equation}
\begin{aligned}
    &\sum_{a,y,d} \| M_{y}^{d}\hatPi^{1,a}_y  - Z^{(a)}_{1,y,d} M_{y}^{d}\|^2_{\psi^{\theta}_{y}}\\
    =& \sum_{a,y,d} \tr[
    (\hatPi^{1,a}_y M_{y}^{d} - M_{y}^{d}Z^{(a)}_{1,y,d})
    (M_{y}^{d}\hatPi^{1,a}_y  - Z^{(a)}_{1,y,d} M_{y}^{d})\psi^\theta_y] \\
    =& \sum_{a,y,d} \tr[\hatPi^{1,a}_yM_{y}^{d}\hatPi^{1,a}_y \psi^\theta_y]
    + \sum_{a,y,d} \tr[M_{y}^{d}Z^{(a)}_{1,y,d} M_y^{d}
    \psi^\theta_y]
    - \sum_{a,y,d} \tr[M_y^d Z^{(a)}_{1,y,d} M_y^d (
   \hatPi^{1,a}_y \psi_y^\theta + \psi_y^\theta\hatPi^{1,a}_y)].
\end{aligned}
\end{equation}
Since $\psi^\theta = \sum_y \psi_y^\theta \otimes \ketbra{y}{y}$,
$M^d = \sum_y M_y^d \otimes \ketbra{y}{y}$,
$\hatPi^{a}_{i} = \sum_{y} \hatPi^{i,a}_y \otimes \ketbra{y}{y}$
and
$Z^{(a)}_{i,d} = \sum_y Z^{(a)}_{1,y,d} \otimes \ketbra{y}{y}$,
we can simplify the left-hand side of \cref{eq:z_and_Pi}
as
\begin{equation}\label{eq:mpi_zm_threeterms}
    \sum_{a,d} \tr[ \hatPi^{1,a} M^{d} \hatPi^{1,a}  \psi^\theta] 
    + \sum_{a,d} \tr[ M^{d}Z^{(a)}_{1,d} M^{d}
    \psi^\theta]
    - \sum_{a,d} \tr[Z^{(a)}_{1,d} M^d (
    \hatPi^{1,a} \psi^\theta + \psi^\theta \hatPi^{1,a})M^d ].
\end{equation}
Because the device is perfect, according to \cref{def:failure_prob}, the first term is 
\begin{equation}
  \sum_{a,d} \tr[ \hatPi^{1,a} M^{d} \hatPi^{1,a}  \psi^\theta] = 
  \sum_a \tr\Bigl[\hatPi^{1,a}  \Bigl( \sum_d M^d \Bigr) \hatPi^{1,a}  \psi^\theta \Bigr]
  =
  \sum_a \tr[ \hatPi^{1,a}  \psi^\theta ]
    =  1 - \negl.
\end{equation}
The second term is
\begin{equation}
 \sum_d \tr\Bigl[ M^d \Bigl(\sum_a Z^{(a)}_{1,d}\Bigr) M^d \psi^\theta\Bigr] = 1.
\end{equation}
For the third term, we first notice that
\begin{equation}
    \norm{\hatPi^{1,0} + \hatPi^{1,1} - \1}^2_{\psi^\theta} = 1 - \tr[(\hatPi^{1,0} + \hatPi^{1,1})\psi^\theta] \leq \negl,
\end{equation}
that is
\begin{equation}\label{eq:sumpi_psi_replacement}
    \hatPi^{1,0} + \hatPi^{1,1} \simeq_{\negl, \psi^\theta} \1.
\end{equation}
Then, we have
\begin{equation}\label{eq:mpi_zm_thirdterm_twoterms}
\begin{aligned}
    &\sum_{a} \tr\Bigl[ \Bigl(\sum_d M^d Z^{(a)}_{1,d} M^d\Bigr) (
    \hatPi^{1,a} \psi^\theta + \psi^\theta \hatPi^{1,a} ) \Bigr] \\
    &\simeq_{\negl}
    \sum_{a,d}\tr[M^d Z^{(a)}_{1,d} M^d (
    \hatPi^{1,a} \psi^\theta
    (\hatPi^0_1 + \hatPi^1_1)
    + (\hatPi^0_1 + \hatPi^1_1) \psi^\theta \hatPi^{1,a} )M^d ] \\
    &= 2 \sum_{a,d} \tr[
    Z^{(a)}_{1,d} M^d \hatPi^{1,a} \psi^\theta
    \hatPi^{1,a} M^d] 
    + \sum_{a,d}
    \tr[
    Z^{(a)}_{1,d} M^d (\hatPi^{1,a}\psi^\theta
    \hatPi^{1,1-a}
    +\hatPi^{1,1-a}\psi^\theta
    \hatPi^{1,a})M^d],
\end{aligned}
\end{equation}
where the first approximation follows from the replacement lemma (\cref{lem:replace}) with \cref{eq:sumpi_psi_replacement}, $\norminfty{\hatPi^{1,a}} \leq 1$, and $\norm{\sum_d M^d Z_{1,d}^{(a)} M^d }_\infty \leq 1$. The last bound follows from  $0 \leq \sum_d M^d Z_{1,d}^{(a)} M^d 
\leq \sum_d M^d =  \1$ which is implied by $0 \leq Z_{1,d}^{(a)} \leq \1$.

Now, the second term appearing on the last line \cref{eq:mpi_zm_thirdterm_twoterms} is zero, which follows from: $\hatPi^{1,a}\psi^\theta
    \hatPi^{1,1-a}
    +\hatPi^{1,1-a}\psi^\theta
    \hatPi^{1,a}$ is independent of $a$, $\sum_a Z_{1,d}^{(a)}=\1$, $\sum_d M^d =\1$, and $\hatPi^{1,0}\,\hatPi^{1,1}=0$.

We argue that the first term appearing on the last line \cref{eq:mpi_zm_thirdterm_twoterms} is close to $2$. For $\theta = 0$, we have
\begin{align}
    \sum_d M^d \hatPi^{1,a} \psi^0 \hatPi^{1,a}  M_y^d \otimes \ketbrasame{d} 
    &= \sum_y \sum_d \, \sum_{b,b'\mid b_1=b_1'=a} M_y^d \, \hatPi^{b}_y \psi_y^0 \hatPi^{b'}_y \, M^d \otimes \ketbrasame{y,d}
    \\
    &= \sum_y \sum_d \sum_{b \mid b_1 = a} M_y^d \, \hatPi^{b}_y \psi_y^0 \hatPi^{b}_y \, M_y^d \otimes \ketbrasame{y,d}
    \\
    &\simeq_{\negl} \sum_{\text{valid } y \mid \hatb(k_1,y_1)=a} \,  \sum_d M_y^d \, \hatPi^{\hatb(k,y)}_y \psi_y^0 \hatPi^{\hatb(k,y)}_y \, M_y^d \otimes \ketbrasame{y,d}
    \\
    &= \sum_{\text{valid } y \mid \hatb(k_1,y_1)=a} \,  \sum_d M_y^d \psi_y^0 M_y^d \otimes \ketbrasame{y,d} = \sum_{v \mid v_1=a} \sigma^{0,v},
\end{align}
where we used \cref{claim:nomatch_hatpi} in the second equality, and $\theta=0$ (so $k\in \calK_\calG^{2N}$) and the perfectness of $D$ in the third (approximate) equality.

Therefore, we have
\begin{equation}\label{eq:comp_indist_base0}
\begin{aligned}
    \sum_{a,d} \tr[
    Z^{(a)}_{1,d} M^d \hatPi^{1,a} \psi^0
    \hatPi^{1,a} M^d]
     &= \sum_{a} \tr[Z_1^{(a)}\sum_d M^d \hatPi^{1,a} \psi^0 \hatPi^{1,a}  M_y^d \otimes \ketbrasame{d}] 
    \\
    &\simeq_{\negl} \sum_v \tr[Z_1^{(v_1)}\sigma^{0,v}] \geq 1 - \gammaH,
\end{aligned}
\end{equation}
where the last inequality follows from \cref{lem:sandwich_argument}.

But given any state $\psi^{\theta}$ of $D$, $\sum_{a,d} \tr[Z^{(a)}_{1,d} M^d \hatPi^{1,a} \psi^\theta \hatPi^{1,a} M^d]$ can be estimated by the following efficient algorithm that is independent of $\theta$:
\begin{enumerate}
    \item Measure $\psi^{\theta}$ using the efficient measurement $\{\sum_{b_2 \ldots b_{2N}, x}\Pi^{b_1 b_2 \ldots b_{2N},x}\}_{b_1}$ to obtain outcome $b_1\in \{0,1\}$.
    \item Measure $M^d$ to obtain outcome $d\in \{0,1\}^{2Nw}$.
    \item Measure $Z_{1,d}$ to obtain bit $b_1'$.
    \item Output $1$ if $b_1 = b_1'$ and $0$ otherwise.
\end{enumerate}

Because $D$ is perfect, the state of the system after the first step is negligibly close to $\sum_a\hatPi^{1,a} \psi^\theta \hatPi^{1,a}$ in trace distance. After the second step, the state becomes $\sum_{a,d}M^d\hatPi^{1,a} \psi^\theta \hatPi^{1,a}M^d$. Therefore, the probability the algorithm outputs $1$ is negligibly close to $\sum_{a,d} \tr[Z^{(a)}_{1,d} M^d \hatPi^{1,a} \psi^\theta \hatPi^{1,a} M^d]$.

Therefore, \cref{eq:comp_indist_base0} and the computational indistinguishability of the $\psi^{\theta}$s (\cref{lem:indistinguishability}) gives 
\begin{equation}
    \sum_{a,d} \tr[Z^{(a)}_{1,d} M^d \hatPi^{1,a} \psi^\theta \hatPi^{1,a} M^d] \geq 1 - \gammaH + \negl.
\end{equation}

Overall, putting together the bounds we have derived for each of the three terms in \cref{eq:mpi_zm_threeterms}, we obtain
\begin{align}
    \sum_{a,y,d} \| M_{y}^{d}\hatPi_{i,y}^{a} - Z^{(a)}_{i,y,d} M_{y}^{d}\|^2_{\psi^{\theta}_{y}} \leq 2\gammaH + 
    \negl,
\end{align}
which completes the proof.
\end{proof}

\begin{lemma}\label{lem:anticomm_xzsigz_marginals}
Let $D$ be an efficient perfect device. For all $\theta\in [2N]$ and $c \in \{0,1\}$, we have
\begin{equation}
\sum_{a, v | v_\theta = c} \tr[X_\theta Z_\theta^{(a)}\sigma^{\theta,v}Z_{\theta}^{(a)}] \approx_{\sqrt{\gammaH} + \negl} 0.
\end{equation}
\end{lemma}
\begin{proof}
The proof of this lemma generalizes that of \cite[Lemma 4.18]{mv2021selftest}.

It suffices to prove the lemma for $\theta = 1$ as the argument is analogous for other $\theta$. Define
\begin{equation}
E_{c} \coloneqq \sum_{a,v|v_{1} = c} \tr[X_1 Z_1^{(a)}\sigma^{1,v}Z_{1}^{(a)}].
\end{equation}
By \cref{lem:anticommutation_xzsigmaz}, we have $E_0+E_1 \approx_{\sqrt{\gammaH}+\negl} 0$. Therefore, it suffices to show that $E_0\approx_{\sqrt{\gammaH} + \negl} E_1$. Using the definitions of $\sigma^{1,v}$ and $\Sigma(1,v)$ from \cref{def:sig_thetav}, and using \cref{def:observables}, we have
\begin{equation}
\begin{aligned}
E_c &= \sum_{a,v|v_{1} = c} \, \sum_{y,d\in \Sigma_{1,v}} \tr[X_1Z_1^{(a)}M_y^d \, \psi_y^{1} \, M_y^d\otimes \ketbra{y,d}{y,d} Z_1^{(a)}]\\
&= \sum_{a,v|v_{1} = c} \, \sum_{y,d\in \Sigmaset_{1,v}} \tr[X_{1,y,d}Z_{1,y,d}^{(a)}M_y^d \, \psi_y^{1} \, M_y^d Z_{1,y,d}^{(a)}]\\
&= \sum_{a,y}\, \sum_{d|\hath(k_1,y_1,d_1) = c} \tr[X_{1,y,d}Z_{1,y,d}^{(a)}M_y^d \, \psi_y^{1} \, M_y^d Z_{1,y,d}^{(a)}].
\end{aligned}
\end{equation}

Now, also define
\begin{equation}
F_c \coloneqq \sum_{a,y}\, \sum_{d|\hath(k_1,y_1,d_1) = c}  \tr[X_{1,y,d}M^d_y \hatPi^{1,a}_y\psi^1_y\hatPi^{1,a}_yM^d_y].
\end{equation}

By the technical \cref{claim:replacezm_crossterm_technical}, deferred to after this proof, we have $E_c\approx_{\sqrt{\gammaH}} F_c$. Therefore, to prove the lemma, it suffices to prove $F_0\simeq_{\negl} F_1$. To this end, we first substitute the definition $\hatPi^{1,a}_y  \coloneqq \sum_{b|b_1 = a} \hatPi^{b}_{y}$ into the expression for $F_c$:
\begin{equation}
    F_c = \sum_{b,y}\, \sum_{d|\hath(k_1,y_1,d_1) = c}  \tr[X_{1,y,d} M^d_y \hatPi^b_{y} \psi^1_y  \hatPi^b_{y} M^d_y] + \sum_{a,y} \, \sum_{d|\hath(k_1,y_1,d_1) = c}\, \sum_{\substack{b\neq b'\\b_1=b_1'=a}} \tr[X_{1,y,d} M_y^d \hatPi_y^b \psi^1_y \hatPi^{b'}_yM_y^d].
\end{equation}

But by \cref{claim:nomatch_hatpi}, the second term is $0$. Therefore,
\begin{equation}
    F_c = \sum_{b,y}\, \sum_{d|\hath(k_1,y_1,d_1) = c}  \tr[X_{1,y,d} M^d_y \hatPi^b_{y} \psi^1_y  \hatPi^b_{y} M^d_y].
\end{equation}

With the above expression in hand, we prove $F_0 \simeq_{\negl} F_1$ by contradiction. Assume that there exists a non-negligible function $\mu$ such that $F_0-F_1>\notnegl>0$ for some $\lambda$ (as in the proof of \cref{lem:uniformity_sigma_marginals}, the case $F_1-F_0>\notnegl>0$ can be handled analogously). We can then construct the following efficient algorithm $\calA$ which breaks the adaptive hardcore bit property of ENTCFs (\cref{property:adaptive_hardcore_bit}).

\begin{enumerate}
\item Input: a key $k_1$ sampled from $\Gen_{\calF}(1^{\lambda})_{\key}$.
\item Independently sample $2N-1$ keys $k_2,\ldots,k_{2N-1}$, each from $\Gen_{\calG}(1^{\lambda})_{\key}$.
\item Prepare the state $\psi^1$, then make measurement obtaining $y$, then make preimage measurement obtaining $(b,x)$, then make Hadamard measurement obtaining $d$, and then measure $X_1$ obtaining (single-bit) $h$.
\item Output: $(y_1,b_1,x_1,d_1,h)$ and print ``$f_{k_1,b_1}(x_1)=y_1 \text{ and } \hath(k_1,y_1,d_1) = h$''.
\end{enumerate}

Because $D$ is perfect, we have, with probability $\geq 1 - \negl$,
\begin{equation}\label{eq:anticomm_xzsigz_marginals_preimage}
f_{k_1,b_1}(x_1)=y_1.
\end{equation}

Write $\hath$ as shorthand for $\hath(k_1,y_1,d_1)$. Using $F_c= \pr_{\calA}(h=0, \hat{h}= c) - \pr_{\calA}(h=1, \hat{h}= c)$, we find
\begin{equation}\label{eq:anticomm_xzsigz_marginals_hardbit}
\begin{aligned}
\pr_{\calA}(\hath = h) &= \pr_{\calA}(\hath = h = 0) + \pr_{\calA}(\hath = h = 1)\\
&= \pr_{\calA}(\hath = 0, h = 1) + \pr_{\calA}(\hath = 1, h = 0)+ F_0 - F_1 \\
&> \pr_{\calA}(\hath \neq h) + \notnegl.
\end{aligned}
\end{equation}
\cref{eq:anticomm_xzsigz_marginals_preimage,eq:anticomm_xzsigz_marginals_hardbit} contradict the adaptive hardcore bit property of ENTCFs (\cref{property:adaptive_hardcore_bit}) as required.
\end{proof}

\begin{claim}\label{claim:replacezm_crossterm_technical}
Under the same setup as in the proof of \cref{lem:anticomm_xzsigz_marginals}, for any $c\in \{0,1\}$, we have
\begin{align}
E_c &\approx_{\sqrt{\gammaH}} F_c.
\end{align}
\end{claim}

\begin{proof}
The proof of this claim generalizes that of \cite[Lemma 4.22]{mv2021selftest}.

By simple algebra, we have
\begin{equation}
\begin{aligned}
E_c - F_c &= \sum_{a,y}\, \sum_{d|\hath(k_1,y_1,d_1) = c}\ipstate{X_{1,y,d}Z^{(a)}_{1,y,d}M_y^d}{Z^{(a)}_{1,y,d}M_y^d-M_y^d\hatPi^{1,a}_y}{\psi^{1}_y} \\
& \hspace{90pt}  + \ipstate{Z^{(a)}_{1,y,d}M_y^d-M_y^d\hatPi^{1,a}_y}{X_{1,y,d}M_y^d\hatPi^{1,a}_y}{\psi^{1}_y}.
\end{aligned}
\end{equation}
Therefore, by Cauchy-Schwarz, we have
\begin{equation}
\begin{aligned}
\abs{E_c - F_c} &\leq \sum_{a,y,d} \normstate{X_{1,y,d}Z^{(a)}_{1,y,d}M_y^d}{\psi^{1}_y}\cdot \normstate{Z^{(a)}_{1,y,d}M_y^d-M_y^d\hatPi^{1,a}_y}{\psi^{1}_y} \\
& \hspace{90pt}  + \normstate{Z^{(a)}_{1,y,d}M_y^d-M_y^d\hatPi^{1,a}_y}{\psi^{1}_y} \cdot \normstate{X_{1,y,d}M_y^d\hatPi^{1,a}_y}{{\psi^{1}_y}},
\end{aligned}
\end{equation}
where we dropped the restriction on $d$, which only makes the upper bound looser.

Again by Cauchy-Schwarz, we have
\begin{equation}
\begin{aligned}
\abs{E_c - F_c} &\leq \sqrt{\sum_{a,y,d} \normstate{X_{1,y,d}Z^{(a)}_{1,y,d}M_y^d}{\psi^{1}_y}^2} \cdot \sqrt{\sum_{a,y,d}\normstate{Z^{(a)}_{1,y,d}M_y^d-M_y^d\hatPi^{1,a}_y}{\psi^{1}_y}^2} \\
& \hspace{90pt}  + \sqrt{\sum_{a,y,d}\normstate{Z^{(a)}_{1,y,d}M_y^d-M_y^d\hatPi^{1,a}_y}{\psi^{1}_y}^2} \cdot \sqrt{\sum_{a,y,d}\normstate{X_{1,y,d}M_y^d\hatPi^{1,a}_y}{\psi^{1}_y}^2}.
\end{aligned}
\end{equation}

Now, we have
$\sum_{a,y,d} \normstate{X_{1,y,d}Z^{(a)}_{1,y,d}M_y^d}{\psi^{1}_y}^2 = 1$ and 
$\sum_{a,y,d}\normstate{X_{1,y,d}M_y^d\hatPi^{1,a}_y}{{\psi^{1}_y}} \simeq_{\negl} 1$, where the first equality directly follows from definitions and the second equality follows as $D$ is perfect. 

Hence,
\begin{equation}
|E_c - F_c| \leq 2\sqrt{\sum_{a,y,d}\normstate{Z^{(a)}_{1,y,d}M_y^d-M_y^d\hatPi^{1,a}_y}{\psi^{1}_y}} \leq 2\sqrt{\gammaH},
\end{equation}
where the last inequality follows from
\cref{lem:anticommutation_mpi_zm}.
\end{proof}

Using the last lemma, \cref{lem:anticomm_xzsigz_marginals}, we can now prove \cref{prop:anticommutation} as follows.

\begin{proof}[Proof of \cref{prop:anticommutation}] The proof of this proposition generalizes that of \cite[Proposition 4.17]{mv2021selftest}.

Because $D$ is efficient, $X_i$ and $Z_i$ are efficient binary observables. Therefore, it suffices to prove the proposition for $\theta = i$ by the computational indistinguishability of the $\sigma^\theta$s (\cref{lem:indistinguishability}) and part $4$ of the lifting lemma (\cref{lem:lifting}). Moreover, it suffices to only prove the proposition for $\theta = i = 1$ as the argument is analogous otherwise.

\begin{equation}\label{eq:anticommutation}
\begin{aligned}
\tr[\{Z_1,X_1\}^2\sigma^1] &= 4\sum_v \tr[(X_1 Z_1^{(0)}X_1Z_1^{(0)} + Z_1^{(1)}X_1Z_1^{(1)}X_1]\sigma^{1,v}) \\
&= 4 \sum_{v} (-1)^{v_1}\tr[(Z_1^{(0)}X_1Z_1^{(0)} + Z_1^{(1)}X_1Z_1^{(1)}]\sigma^{1,v})+O(\sqrt{\chi(1,v)}\cdot \sqrt{\tr[\sigma^{1,v}]} )\\
&\approx_{\sqrt{\gammaH}+\negl} 0 + O\Bigr(\sum_v \sqrt{\chi(1,v)}\cdot \sqrt{\tr[\sigma^{1,v}]} \Bigr),
\end{aligned}
\end{equation}
where the first equality is by simple algebra, the second equality is by \cref{def:zeta_chi_def} and the replacement lemma (note that $\|Z_{1}^{(a)}X_1Z_{1}^{(a)} \|_{\infty} \leq 1$ for $a\in\{0,1\}$), and the last approximation is by \cref{lem:anticomm_xzsigz_marginals} with $\theta = 1$ (which uses the perfectness of $D$). But, by Cauchy-Schwarz, we have
\begin{equation}
    \sum_v \sqrt{\chi(1,v)}\cdot \sqrt{\tr[\sigma^{1,v}]}
    \leq \sqrt{ \Bigl( \sum_v \chi(1,v) \Bigr)
    \Bigl(\sum_v \tr[\sigma^{1,v}]\Bigr)} 
    \leq \sqrt{\gammaH}.
\end{equation}
Therefore, \cref{eq:anticommutation} implies that
\begin{equation}
    \tr[\{Z_1,X_1\}^2\sigma^1] \approx_{\sqrt{\gammaH} + \negl} 0,
\end{equation}
which completes the proof.
\end{proof}

\subsection{Operator-state commutation}
\label{sec:conj_inv}

In this subsection, we prove \cref{prop:operator_state_commutation} that shows the $\sigma^{\theta}$ states approximately commute with $Z_i$ and $X_i$ for certain pairs of $(\theta,i)$.

\begin{proposition}[Operator-state commutation]\label{prop:operator_state_commutation}
Let $D$ be an efficient perfect device. For all $i,\theta \in [2N]$ with $i\neq \theta$, we have
\begin{equation}\label{eq:operator_state_commutation}
    Z_i \,  \sigma^\theta  \approx_{\gammaH+\negl}  \sigma^\theta \, Z_i \quad \text{and} \quad
    X_\theta \,  \sigma^\theta \approx_{\gammaH+\negl}  \sigma^\theta \, X_\theta,
\end{equation}
moreover, for all $q\in \{0,1\}^{2N}$, we have
\begin{equation}\label{eq:operator_state_commutation_q}
\begin{alignedat}{2}
    &Z_{q,i} \,  \sigma^\theta   &\,\approx_{\gammaH + \gammaHq + \negl}  \sigma^\theta \, Z_{q,i} &\quad \text{if $q_i=0$}, 
    \\
    &X_{q,\theta} \,  \sigma^\theta  &\,\approx_{\gammaH + \gammaHq + \negl}  \sigma^\theta \, X_{q,\theta} &\quad \text{if $q_\theta=1$}.
\end{alignedat}
\end{equation}
\end{proposition}

\begin{proof}
The equations in \cref{eq:operator_state_commutation_q} follow from those in \cref{eq:operator_state_commutation} by \cref{prop:tilde_equal_nontilde} and \cref{lem:operator_to_state}. The proofs of the two equations in  \cref{eq:operator_state_commutation} are  analogous. We prove the first in detail and comment on the minor changes required to prove the second.

To prove the first equation, it suffices to prove it for $i=1$ and $\theta=2$ as the proof for other $i\neq \theta$ is analogous. By the triangle inequality and \cref{lem:sigma_equal_sumv}, we have

\begin{equation}\label{eq:operator_state_commutation_l1_z}
    \| Z_1 \sigma^2 Z_1 - \sigma^2 \|_1 \leq \sum_v \| Z_1 \sigma^{2,v} Z_1 - \sigma^{2,v} \|_1 + \negl.
\end{equation}
We bound each term in the sum. Recall the definition $\zeta(2,1,v) \coloneqq \normstate{Z_1-(-1)^{v_1}\1}{\sigma^2,v}^2$ from \cref{def:zeta_chi_def}. Therefore, by \cref{lem:operator_to_state}, we have 
\begin{equation}
    \normone{ Z_1\sigma^{2,v}Z_1 - \sigma^{2,v} } = \normone{ Z_1\sigma^{2,v}Z_1 - (-1)^{v_1}\1 \sigma^{2,v}(-1)^{v_1}\1} \leq  2\sqrt{\zeta(2,1,v)} \cdot \sqrt{\tr[\sigma^{2,v}]}.
\end{equation}
Therefore, resuming  \cref{eq:operator_state_commutation_l1_z}:
\begin{equation}
    \| Z_1 \sigma^2 Z_1 - \sigma^2 \|_1 \leq 2\sum_v \sqrt{\zeta(2,1,v)} \cdot \sqrt{\tr[\sigma^{2,v}]} + \negl \leq 4 \sqrt{\gammaH} + \negl,
\end{equation}
where the last inequality follows from Cauchy-Schwarz and bounds on $\zeta$ (and $\chi$) in \cref{lem:zeta_chi_bounds_new}. Hence we have proved the first equation.

To prove the second equation in \cref{eq:operator_state_commutation}, it suffices to prove it for $\theta=1$. Then, we can exactly reuse the above proof after changing the symbols $\{Z_1, \sigma^2, \sigma^{2,v}, \zeta(2,1,v)\}$ to $\{X_1, \sigma^1, \sigma^{1,v}, \chi(1,v)\}$, respectively.
\end{proof}

Observe that \cref{prop:operator_state_commutation}  does \emph{not} say $Z_{q,i}$ and $X_{q,i}$ commute with $\sigma^\theta$ 
for all pairs $(i,\theta)$. To get around this problem, we use the computational indistinguishability of the $\sigma^\theta$s to argue that efficient observables must act similarly on different $\sigma^\theta$s. For example,
\begin{equation}\label{eq:example_operator_state_commutation_similar}
    Z_1 X_3 Z_2 \sigma^2 \approx X_3 Z_1 Z_2 \sigma^2,
\end{equation}
does not follow from \cref{prop:operator_state_commutation}, since $Z_2$ does not commute with $\sigma^2$. Nevertheless, by using the computational indistinguishability of $\sigma^2$ and $\sigma^3$, we can derive an ``operational version'' of \cref{eq:example_operator_state_commutation_similar}. The operational version allows us to interchange the left-hand and right-hand sides of \cref{eq:example_operator_state_commutation_similar} when they appear inside traces (i.e., $\tr$). We can only derive such an operational version because the computational indistinguishability of $\sigma^2$ and $\sigma^3$ only allows us to interchange $\sigma^2$ and $\sigma^3$ inside traces, see \cref{lem:lifting,lem:lifting_proj}. For a concrete example of how to use \cref{prop:operator_state_commutation}, see the long aligned equation in the proof of \cref{lem:swap_observables_tildeprod}.

\subsection{The swap isometry}
\label{sec:swap}

In this subsection, we define the swap isometry that we will show maps the observables and states of the device onto the ideal ones.

\begin{defn}\label{def:swap}
Let $D$ be a device and let $\calH\coloneqq \calH_D \otimes \calH_Y \otimes \calH_R$. The swap isometry is the map $\calV: \calH \rightarrow \mathbb{C}^{2^{2N}}\otimes \calH$ defined by 
\begin{equation}\label{eq:def_swap}
    \calV = \sum_{u\in \{0,1\}^{2N}}\ket{u} \otimes \prod_{i \in [2N]} X_{i}^{u_i} \prod_{j \in [2N]} Z_{j}^{(u_j)}.
\end{equation}
Note that by the commutation relations (\cref{prop:commutation}) the ordering within each $\prod_{i\in [2N]}$ does not matter.
\end{defn}

\begin{lemma}\label{lem:v_efficient}
    Let $D$ be an efficient device. Then, $\calV$ is efficient.
\end{lemma}
\begin{proof}
$\calV$ can be implemented efficiently by the following circuit: 
\begin{enumerate}
    \item Initialize $2N$ ancilla qubits to $\ket{0}$.
    \item Apply $H^{\otimes 2N}$ on the ancilla. 
    
    \item Apply $2N$ c-$Z_i$ gates where the control is by the $i$-th ancilla qubit for $i=1,\ldots,2N$.
    
    \item Apply $H^{\otimes 2N}$ on the ancilla.
    
    \item Apply $2N$ c-$X_i$ gates where the control is by the $i$-th ancilla qubit for $i=1,\ldots,2N$. 
\end{enumerate}
It can be verified by direct calculation that the above circuit indeed gives $\calV$.

We illustrate $\calV$ when $2N=4$ in \cref{fig:swap_isometry} below.
\begin{figure}[H]
\center
        \begin{tikzpicture}[thick]
        \tikzstyle{operator} = [draw,fill=white,minimum size=1.5em] 
        \tikzstyle{phase} = [fill,shape=circle,minimum size=5pt,inner sep=0pt]
        \tikzstyle{surround} = [fill=blue!10,thick,draw=black,rounded corners=2mm]

        \node at (0,0) (q1) {$\ket{0}$};
        \node at (0,-1) (q2) {$\ket{0}$};
        \node at (0,-2) (q3) {$\ket{0}$};
        \node at (0,-3) (q4) {$\ket{0}$};
        \node at (0,-4) (q5) {$\ket{\psi}$};
        
        \node[operator] (op11) at (1,0) {$H$} edge [-] (q1);
        \node[operator] (op12) at (1,-1) {$H$} edge [-] (q2);
        \node[operator] (op13) at (1,-2) {$H$} edge [-] (q3);
        \node[operator] (op14) at (1,-3) {$H$} edge [-] (q4);
        
        \node[phase] (phase11) at (2,0) {} edge [-] (op11);
	    \node[operator] (z1) at (2,-4) {$Z_1$} edge [-] (q5);
        \draw[-] (phase11) -- (z1);

        \node[phase] (phase12) at (3,-1) {} edge [-] (op12);
	    \node[operator] (z2) at (3,-4) {$Z_2$} edge [-] (z1);
        \draw[-] (phase12) -- (z2);

        \node[phase] (phase13) at (4,-2) {} edge [-] (op13);
	    \node[operator] (z3) at (4,-4) {$Z_3$} edge [-] (z2);
        \draw[-] (phase13) -- (z3);
        
        \node[phase] (phase14) at (5,-3) {} edge [-] (op14);
	    \node[operator] (z4) at (5,-4) {$Z_4$} edge [-] (z3);
        \draw[-] (phase14) -- (z4);
        
	    \node[operator] (op21) at (6,0) {$H$} edge [-] (phase11);
        \node[operator] (op22) at (6,-1) {$H$} edge [-] (phase12);
        \node[operator] (op23) at (6,-2) {$H$} edge [-] (phase13);
        \node[operator] (op24) at (6,-3) {$H$} edge [-] (phase14);
        
        \node[phase] (phase21) at (7,0) {} edge [-] (op21);
	    \node[operator] (x1) at (7,-4) {$X_1$} edge [-] (z4);
        \draw[-] (phase21) -- (x1);

        \node[phase] (phase22) at (8,-1) {} edge [-] (op22);
	    \node[operator] (x2) at (8,-4) {$X_2$} edge [-] (x1);
        \draw[-] (phase22) -- (x2);

        \node[phase] (phase23) at (9,-2) {} edge [-] (op23);
	    \node[operator] (x3) at (9,-4) {$X_3$} edge [-] (x2);
        \draw[-] (phase23) -- (x3);

        \node[phase] (phase24) at (10,-3) {} edge [-] (op24);
	    \node[operator] (x4) at (10,-4) {$X_4$} edge [-] (x3);
        \draw[-] (phase24) -- (x4);
        
        \node (end1) at (11,0) {} edge [-] (phase21);
        \node (end2) at (11,-1) {} edge [-] (phase22);
        \node (end3) at (11,-2) {} edge [-] (phase23);
        \node (end4) at (11,-3) {} edge [-] (phase24);
        \node (end5) at (11,-4) {} edge [-] (x4);
        
        \begin{pgfonlayer}{background} 
        \node[surround] (background) [fit = (q1) (q5) (end1)(end5)] {};
        \end{pgfonlayer}

        \end{tikzpicture}
	\caption{Illustration of the swap isometry $\calV$ when $2N=4$.}\label{fig:swap_isometry}
\end{figure}
\end{proof}

\subsection{Observables under the swap isometry}
\label{sec:swap_observables}

To prove the results in this and the following subsections, we make extensive use of \cref{prop:operator_state_commutation}.

\begin{lemma}\label{lem:swap_observables}
Let $D$ be an efficient perfect device. For all $k \in [2N]$ and $\theta\in \thetaset$, we have
\begin{equation}
 \begin{aligned}
    \calVdag (\sigmaz_k \otimes \1) \calV &= Z_k,
    \\
    \calVdag(\sigmax_k\otimes \1)\calV &\approx_{N\sqrt{\gammaH},\sigma^{\theta}}X_k,
\end{aligned}
\label{eq:swap_observables_x} 
\end{equation}
moreover, for all $q\in \{0,1\}^{2N}$, we have
\begin{alignat}{2}
    \calVdag (\sigmaz_k \otimes \1) \calV &\approx_{\gammaH+\gammaHq+\negl, \sigma^{\theta}} Z_{q,k} &\quad \text{if $q_k = 0$,}
    \\
    \calVdag(\sigmax_k\otimes \1)\calV &\approx_{N\sqrt{\gammaH} + \gammaHq,\sigma^{\theta}} X_{q,k}
    &\quad \text{if $q_k = 1$.}
\end{alignat}
\end{lemma}

\begin{proof}
The last two equations follow immediately from the first two and \cref{prop:tilde_equal_nontilde} (note $\gammaP=\negl$ as $D$ is perfect). The first equation follows by direct calculation using the fact that, for all $i\in [2N]$, we have $X_i^2=\1$ and $\sum_{a\in \{0,1\}}Z_i^{(a)}=\1$.

Consider the second equation. For notational convenience, we write for $i\in [2N]$, $k\in[2N]$, and $u = (u_i,\ldots,u_{2N})\in \{0,1\}^{2N-i+1}$,
\begin{enumerate}
    \item $\calZ_i(u)$ for the product of the $2N$-th (from the left) operator down to the $i$th operator in the $2N$-tuple $(Z_1^{(u_1)}, Z_2^{(u_2)},\ldots,Z_{k}^{(u_{k})},\ldots,Z_{2N}^{(u_{2N})})$. More precisely
    \begin{equation}
        \calZ_i(u) \coloneqq Z_{2N}^{(u_{2N})} Z_{2N-1}^{(u_{2N-1})} \cdots Z_{i}^{(u_{i})}.
    \end{equation}

    \item $\calZ_{k,i}(u)$ for the product of the $2N$-th operator down to the $i$th operator in the $2N$-tuple \newline
    $(Z_1^{(u_1)}, Z_2^{(u_2)},\ldots,Z_{k-1}^{(u_{k-1})}, Z_k^{(\overline{u_k})}, Z_{k+1}^{(u_{k+1})},\ldots,Z_{2N}^{(u_{2N})})$. More precisely
    \begin{equation}
        \calZ_{k,i}(u) \coloneqq 
        \begin{cases}
            Z_{2N}^{(u_{2N})} Z_{2N-1}^{(u_{2N-1})} \cdots Z_{i}^{(u_{i})} &\text{if $k<i$,}
            \\[10pt]
            Z_{2N}^{(u_{2N})} Z_{2N-1}^{(u_{2N-1})} \cdots Z_{k}^{(\overline{u_{k}})}\cdots Z_{i}^{(u_{i})} &\text{if $k\geq i$.}
        \end{cases}
    \end{equation}
    
    \item $\hatX_{k,i}=
        \sum_{u = (u_i,\ldots,u_{2N})\in \{0,1\}^{2N-i+1}} \calZ_{k,i}(u) X_k \calZ_i(u).$
\end{enumerate}
We also write $\calZ_{2N+1}(u) \coloneqq \1$, $\calZ_{k,2N+1}(u) = \1$, and $\hatX_{k,2N+1} = X_k$ for all $k \in [2N]$.

By direct calculation, we find that for $i\in [2N+1]$ and $k\in [2N]$,
\begin{equation}\label{eq:def_tildeX}
    \hatX_{k,i} = \calVdag_{k,i} (\sigmax_k \otimes \1) \calV_{k,i},
\end{equation}
where $\calV_{k,i}$ is an isometry defined in the next paragraph. \cref{eq:def_tildeX} allows us to deduce that
\begin{equation}\label{eq:tildeX_norm_leqone}
    \norminfty{\hatX_{k,i}}\leq 1,
\end{equation}
which will be crucial in the following main argument when we use the replacement lemma (\cref{lem:replace}).

For $i\in [2N]$ and $k\in [2N]$, we define the isometry $\calV_{k,i}$ according to whether $k<i$ or $k\geq i$. If $k\geq i$, we define $\calV_{k,i}$ by the following circuit:
\begin{enumerate}
    \item Initialize $2N$ ancilla qubits to $\ket{0}$.
    \item Apply $H^{\otimes (2N-i+1)}$ on the $i$-to-$2N$-th ancilla qubits.
    \item Apply $(2N-i+1)$ c-$Z_j$ gates where the control is by the $j$-th ancilla qubit and $j = i,\ldots,2N$.
    \item Apply $H^{\otimes (2N-i+1)}$ on the $i$-to-$2N$-th ancilla qubits.
    \item Apply a single c-$X_k$ gate where the control is by the $k$-th ancilla qubit.
\end{enumerate}
If $k<i$, we define $\calV_{k,i}$ by the same circuit except, in the second step, we additionally apply an $H$ gate on the $k$-th ancilla qubit. We also define $\calV_{k,2N+1} \coloneqq \1$ for all $k\in [2N]$. 

We illustrate $\calV_{k,i}$ when $2N = 4$, $k = 1$, and $i = 3$ below.

\begin{figure}[H]
\center
    \begin{tikzpicture}[thick]
    \tikzstyle{operator} = [draw,fill=white,minimum size=1.5em] 
    \tikzstyle{phase} = [fill,shape=circle,minimum size=5pt,inner sep=0pt]
    \tikzstyle{surround} = [fill=blue!10,thick,draw=black,rounded corners=2mm]

    \node at (0,0) (q1) {$\ket{0}$};
    \node at (0,-1) (q2) {$\ket{0}$};
    \node at (0,-2) (q3) {$\ket{0}$};
    \node at (0,-3) (q4) {$\ket{0}$};
    \node at (0,-4) (q5) {$\ket{\psi}$};
    
    \node[operator] (op11) at (1,0) {$H$} edge [-] (q1);
    \node[operator] (op13) at (1,-2) {$H$} edge [-] (q3);
    \node[operator] (op14) at (1,-3) {$H$} edge [-] (q4);
    
    \node[phase] (phase12) at (2,-2) {} edge [-] (op13);
    \node[operator] (z2) at (2,-4) {$Z_3$};
    \draw[-] (phase12) -- (z2);
    \draw[-] (q5) -- (z2);

    \node[phase] (phase13) at (3,-3) {} edge [-] (op14);
    \node[operator] (z3) at (3,-4) {$Z_4$} edge [-] (z2);
    \draw[-] (phase13) -- (z3);
    
    \node[operator] (op22) at (4,-2) {$H$} edge [-] (phase12);
    \node[operator] (op23) at (4,-3) {$H$} edge [-] (phase13);

    \node[phase] (phase21) at (5,0) {} edge [-] (op11);
    \node[operator] (x1) at (5,-4) {$X_1$} edge [-] (z3);
    \draw[-] (phase21) -- (x1);
    
    \node (end1) at (6,0) {} edge [-] (phase21);
    \node (end2) at (6,-1) {} edge [-] (q2);
    \node (end3) at (6,-2) {} edge [-] (op22);
    \node (end4) at (6,-3) {} edge [-] (op23);
    \node (end5) at (6,-4) {} edge [-] (x1);

    \begin{pgfonlayer}{background} 
    \node[surround] (background) [fit = (q1) (q5) (end1)(end5)] {};
    \end{pgfonlayer}
    
    \end{tikzpicture}
	\caption{Illustration of $\calV_{k,i}$ when $2N=4$, $k = 1$, and $i = 3$.}
\end{figure}

We proceed with our main argument. By direct calculation, we have
\begin{equation}\label{eq:swap_direct_x}
    \calVdag (\sigmax_k \otimes \1) \calV = \sum_{u\in \{0,1\}^{2N}} \calZ_{k,1}(u) X_k \calZ_1(u).
\end{equation}
Note that $\calVdag (\sigmax_k \otimes \1) \calV = \calVdag_{k,1} (\sigmax_k \otimes \1) \calV_{k,1} = \hatX_{k,1}$.

 By direct expansion, we have
 \begin{equation}
    \begin{aligned}
  \|\calVdag(\sigmax_k \otimes \1)\calV - X_k\|_{\sigma^{\theta}}^2&=     \tr[(\calVdag(\sigmax_k \otimes \1)\calV - X_k)^2 \sigma^{\theta}]  \\
    &= 1 + \tr[\calVdag(\sigmax_k \otimes \1)\calV\calVdag(\sigmax_k \otimes \1)\calV \sigma^{\theta}] - 2 \real \tr[\calVdag(\sigmax_k \otimes \1)\calV \, X_k \, \sigma^{\theta}] \\
    &= 1 + \tr[(\calV\calVdag)(\sigmax_k \otimes \1)\calV \sigma^{\theta}\calVdag(\sigmax_k \otimes \1)] - 2 \real \tr[\calVdag(\sigmax_k \otimes \1)\calV \, X_k \, \sigma^{\theta}]\\
    &\leq 2 - 2 \real \tr[\calVdag(\sigmax_k \otimes \1)\calV \, X_k \, \sigma^{\theta}],   
    \end{aligned} \label{eq:expand_observable_difference}
 \end{equation}
where, in the last inequality, we used the fact that $\calV\calVdag$ is a projector (Hermitian and idempotent) and $(\sigmaz_k \otimes \1)\calV \sigma^{\theta}\calVdag(\sigmaz_k \otimes \1)$ is a positive semi-definite operator.

We now show that the last term, $\real \tr[\calVdag(\sigmax_k \otimes \1)\calV \, X_k \, \sigma^{\theta}]$, is close to $1$. By the lifting lemma (\cref{lem:lifting}, part $5$) and the efficiency of $\calV$ (\cref{lem:v_efficient}), it suffices to consider the case when $\theta=k$. In the following, the terms \{commute, replace, o.s.-commute\} are shorthand for \{\cref{prop:commutation}, \cref{lem:replace}, \cref{prop:operator_state_commutation}\} respectively (``o.s.'' stands for ``operator-state''). When using the replacement lemma, we crucially use $\norminfty{\hatX_{k,i}}\leq 1$ (\cref{eq:tildeX_norm_leqone}).

\begin{align*}
    & \quad \tr[\calVdag(\sigmax_k \otimes \1)\calV \, X_k \, \sigma^{k}]
    \\
    &= \tr[\hatX_{k,1} \, X_k \, \sigma^{k}] 
    &&\text{($\hatX_{k,1} = \calVdag (\sigmax_k \otimes \1) \calV$)}
    \\
    &= \sum_{u} \tr[\calZ_{k,1}(u) X_k \calZ_1(u) \, X_k \sigma^{k}]
    &&\text{(\cref{eq:swap_direct_x})}
    \\
    &= \sum_{a} \tr\Bigl[Z_{1}^{(a)} \Bigl(\sum_{u = (u_2,\ldots,u_N)}\calZ_{k,2}(u) X_k \calZ_2(u) \Bigr) Z_{1}^{(a)} \, X_k \sigma^{k}\Bigr]
    &&\text{(definition of $\calZ_{k,2}(u)$, $\calZ_2(u)$)}
    \\
    &= \sum_{a} \tr[Z_{1}^{(a)}\hatX_{k,2} Z_{1}^{(a)} \, X_k \sigma^{k}]
    &&\text{(definition of $\hatX_{k,2}$)}
    \\
    &\approx_{\sqrt{\gammaH+\negl}} 
    \sum_{a} \tr[Z_{1}^{(a)} \hatX_{k,2} \, X_k Z_{1}^{(a)} \sigma^{k}]
    &&\text{(commute and replace)}
    \\
    &\approx_{\sqrt{\gammaH + \negl}} 
    \sum_{a} \tr[Z_{1}^{(a)} \hatX_{k,2} \, X_k  \sigma^{k} Z_{1}^{(a)}]
    &&\text{(o.s.-commute and replace)}
    \\
    &= \sum_{a} \tr[Z_{1}^{(a)}\hatX_{k,2} \, X_k  \sigma^{k}]
    &&\amatrix{\text{trace is cyclic and}}{\text{$(Z_{1}^{(a)})^2=Z_{1}^{(a)}$}}
    \\
    &= \tr[\hatX_{k,2} \, X_k  \sigma^{k}]
    &&\text{$\Bigl( \sum_a Z_1^{(a)} = \1 \Bigr)$}.
\end{align*}

Therefore, we have $\tr[\hatX_{k,1} \, X_k \, \sigma^{k}] \approx_{\sqrt{\gammaH}+\negl} \tr[\hatX_{k,2} \, X_k  \sigma^{k}]$.  Using the above reasoning another $k-2$ times and the triangle inequality gives 
$\tr[\hatX_{k,1} \, X_k \, \sigma^{k}]\approx_{k\sqrt{\gammaH} + \negl} 
   \tr[\hatX_{k,k} \, X_k  \sigma^{k}]$,
   where we use $k \cdot \negl \leq 2N \cdot \negl = \negl$. Then by the definitions of $\hatX_{k,k}$ and $\hatX_{k,k+1}$, and direct calculation, 
$ \tr[\hatX_{k,k} \, X_k  \sigma^{k}]= \sum_{a} \tr[Z_{k}^{(\abar)}\hatX_{k,k+1} Z_{k}^{(a)} \, X_k \sigma^{k}]$. Using $Z_k^{(a)} = (1+(-1)^aZ_k)/2$ to perform the sum over $a$ gives
\begin{equation}
\tr[\hatX_{k,k} \, X_k  \sigma^{k}]=\frac{1}{2} (\tr[\hatX_{k,k+1}  X_k \sigma^k] - \tr[Z_k\hatX_{k,k+1} Z_k X_k \sigma^k]).
\end{equation}

We approximate both terms using the reasoning above another $2N-k$ times and the triangle inequality (when handling the second term, we additionally use the commutativity of the $Z_i$s):
\begin{equation}
 \begin{aligned}
&\frac{1}{2} \tr[\hatX_{k,k+1}  X_k \sigma^k] - \frac{1}{2}\tr[Z_k\hatX_{k,k+1} Z_k X_k \sigma^k] \\
&\approx_{(2N-k)\sqrt{\gammaH} + \negl} \frac{1}{2} \tr[\hatX_{k,2N+1} X_k \sigma^k] - \frac{1}{2} \tr[Z_k \hatX_{k,2N+1}  Z_k X_k \sigma^k]
   \\
   &= \frac{1}{2} \tr[X_k X_k \sigma^k] - \frac{1}{2} \tr[Z_k X_k  Z_k X_k \sigma^k]
   \\
   &= \frac{1}{2} - \frac{1}{2} \tr[Z_k X_k Z_k X_k \sigma^k],
\end{aligned}   
\end{equation}
where the second equation is by $\hatX_{k,2N+1}=X_k$ and the last equation is by $X_k^2 = \1$.

Applying the triangle inequality to all of the above approximations and then taking the real part gives
\begin{equation}
    \real \tr[\calVdag(\sigmax_k \otimes \1)\calV \, X_k \, \sigma^{k}] \approx_{N\sqrt{\gammaH}+\negl} \frac{1}{2} - \frac{1}{4} (\tr[Z_k X_k Z_k X_k \sigma^k] + \tr[X_k Z_k X_k Z_k \sigma^k]).
\end{equation}
By directly unpacking the definition of the anti-commutation relation $\{X_k,Z_k\}\approx_{\sqrt{\gammaH}+\negl,\sigma^k} 0$ (\cref{prop:anticommutation}), we see that 
$\tr[Z_k X_k Z_k X_k \sigma^k] + \tr[X_k Z_k X_k Z_k \sigma^k] \approx_{\sqrt{\gammaH}+\negl} -2$ (no replacement lemma used). Hence
\begin{equation}
    \real \tr[\calVdag(\sigmax_k \otimes \1)\calV \, X_k \, \sigma^{k}] \approx_{N\sqrt{\gammaH} + \negl} 1.
\end{equation}
Lastly, recall from the comment below \cref{def:gamma} that we assume $N\gammaH$ is non-negligible in $\lambda$,
so $O( N \sqrt{\gammaH} + \negl) = O(N \sqrt{\gammaH})$. Hence the lemma.
\end{proof}

\begin{lemma}\label{lem:swap_observables_tildeprod}
Let $D$ be an efficient perfect device. For $k\in [N]$ and $\theta \in \thetaset$, we have
\begin{equation}
\begin{aligned}
    &\calVdag (\sigmax_k \otimes \sigmaz_{N+k}\otimes \1) \calV \approx_{N^{1/4} \gammaH^{1/8}, \sigma^{\theta}}  \tildeX_k \tildeZ_{N+k},
    \\
    &\calVdag (\sigmaz_k \otimes \sigmax_{N+k}\otimes \1) \calV \approx_{N^{1/4} \gammaH^{1/8}, \sigma^{\theta}}  \tildeZ_k \tildeX_{N+k}.
\end{aligned}
\end{equation}
\end{lemma}
\begin{proof}
By \cref{lem:mv_223} and \cref{lem:swap_observables}, for all $k\in [2N]$ and $\theta\in \thetaset$, we have
\begin{equation}\label{eq:swap_observables_flipV}
\begin{alignedat}{2}  
    &\calV \tildeZ_{k} \calVdag &&\approx_{\sqrt{\gammaH}+\negl,\calV\sigma^{\theta}\calVdag} \sigmaz_{k}\otimes \1,
    \\
    &\calV \tildeX_{k} \calVdag &&\approx_{\sqrt{N}\gammaH^{1/4},\calV\sigma^{\theta}\calVdag} \sigmaz_{k}\otimes \1.
\end{alignedat}
\end{equation}

It suffices to prove the first equation for $k=1$ as the proof for other $k$ is analogous. In the following, the terms \{indist., lift, o.s.-commute, replace\} are shorthand for \{\cref{lem:indistinguishability}, \cref{lem:lifting} (part $5$), \cref{prop:operator_state_commutation}, \cref{lem:replace}\} respectively.

\begin{align*}
    &\tr[(\calVdag (\sigmax_1\otimes \sigmaz_{N+1}\otimes \1) \calV - \tildeX_1 \tildeZ_{N+1})^{\dagger} (\calVdag (\sigmax_1\otimes \sigmaz_{N+1}\otimes \1) \calV - \tildeX_1 \tildeZ_{N+1}) \sigma^{\theta}]
    \\
    &= 2 - 2 \real \tr[\calVdag(\sigmax_1\otimes \sigmaz_{N+1}\otimes \1) \calV \tildeX_1 \tildeZ_{N+1} \sigma^{\theta}]
    &&\text{(rearrange)}
    \\
    &\simeq_{\negl} 2 - 2 \real \tr[\calVdag(\sigmax_1\otimes \sigmaz_{N+1}\otimes \1) \calV \tildeX_1 \tildeZ_{N+1} \sigma^{1}]   &&\text{(indist. and lift)}
    \\
    &\approx_{\sqrt{\gammaH}+\negl} 2 - 2 \real \tr[\calVdag(\sigmax_1\otimes \sigmaz_{N+1}\otimes \1) \calV \tildeX_1  \sigma^{1} \tildeZ_{N+1}]
    &&\text{(o.s.-commute and replace)}
    \\
    &=  2 - 2 \real \tr[(\calV \tildeZ_{N+1} \calVdag) (\sigmax_1\otimes \sigmaz_{N+1}\otimes \1) (\calV \tildeX_1  \calVdag) (\calV \sigma^{1} \calVdag)] &&\text{(trace is cyclic and $\calVdag\calV=\1$)}
    \\
    &\approx_{\gammaH^{1/4}+\negl}  2 - 2 \real \tr[(\sigmax_1\otimes \1_2\otimes \1) (\calV \tildeX_1  \calVdag) (\calV \sigma^{1} \calVdag)]
    &&\text{(\cref{eq:swap_observables_flipV} and replace)}
    \\
    &= 2 - 2 \real \tr[\calVdag(\sigmax_1\otimes \1_2\otimes \1) \calV \tildeX_1 \sigma^{1} ]
    &&\text{(trace is cyclic and $\calVdag\calV=\1$)}
    \\
    &\approx_{\sqrt{\gammaH}+\negl} 2 - 2 \, \real \tr[\calVdag(\sigmax_1\otimes \1_2\otimes \1) \calV  \sigma^{1} \tildeX_1]
    &&\text{(o.s.-commute and replace)}
    \\
    &= 2 - 2 \real \tr[(\calV \tildeX_1 \calVdag)(\sigmax_1\otimes \1_2\otimes \1) (\calV  \sigma^{1} \calVdag)]
    &&\text{(trace is cyclic and $\calVdag\calV=\1$)}
    \\
    &\approx_{N^{1/4} \gammaH^{1/8}} 2 - 2 \real \tr[(\1_2 \otimes \1_2\otimes \1) (\calV  \sigma^{1} \calVdag)]
    &&\text{(\cref{eq:swap_observables_flipV} and replace)}
    \\
    &=0.
\end{align*}
Therefore, by the triangle inequality, we have
\begin{equation}
   \calVdag (\sigmax_k \otimes \sigmaz_{N+k}\otimes \1) \calV  \approx_{N^{1/4} \gammaH^{1/8}, \sigma^{\theta}}  \tildeX_k \tildeZ_{N+k},
\end{equation}
and hence the first equation of the lemma.

The proof for the second equation is similar. We may again only consider the case $k=1$. Then, the second equation follows from the same steps as above except we lift to $\theta=N+1$ at the lifting step.
\end{proof}

\subsection{States under the swap isometry}
\label{sec:swap_states}

In this subsection, we show that the states $\sigma^{\theta,v}$ of a device that passes our protocol with high probability must be close to states that we call $\tau^{\theta, v}$ under the swap isometry. The states $\tau^{\theta, v}$ are defined as follows.

\begin{defn}[density operators $\tau^{\theta,v}$]\label{def:tau_thetav}
    Let $v\in \{0,1\}^{2N}$. For $\theta \in \thetaset$, we define the $2N$-qubit density operator $\tau^{\theta, v} \coloneqq \ketbrasame{\tau^{\theta, v}}$, according to the following three cases. 
\begin{equation}\label{eq:tau_def}
    \ket{\tau^{\theta, v}} \coloneqq
    \begin{cases}
        \ket{v_1} \otimes \cdots \otimes \ket{v_{\theta-1}} \otimes \ket{(-)^{v_\theta}} \otimes \ket{v_{\theta+1}} \otimes \cdots \otimes \ket{v_{2N}} &\text{if $\theta \in [2N]$,}
        \\
        \ket{v} \coloneqq \ket{v_1} \otimes \cdots \otimes \ket{v_{2N}} & \text{if $\theta = 0$,}
        \\
       \ket{\psi^{v}} &\text{if $\theta = \diamond$,}
    \end{cases}
\end{equation}
where $\ket{\psi^{v}}$ is as defined in \cref{eq:def_psi}.
\end{defn}
For convenience, we also define the following notation. For $a\in \{0,1\}$ and $k\in [2N]$, we define
\begin{equation}
\begin{aligned}
    &\ketbrasame{a}_k \coloneqq \1_2^{\otimes (k-1)}  \otimes \ketbrasame{a} \otimes \1_2^{\otimes (2N-k)},\\
    & \ketbrasame{(-)^a}_k \coloneqq \1_2^{\otimes (k-1)}  \otimes \ketbrasame{(-)^a} \otimes \1_2^{\otimes (2N-k)}.
\end{aligned}
\end{equation}
For $v \in \{0,1\}^{2N}$ and $j \in [2N]$, we define
\begin{equation}\label{def:comp_hadamard_projs}
\begin{aligned}
    &\ketbrasame{v_{j:2N}} \coloneqq 
    \1_{2}^{\otimes(j-1)} \otimes \ketbrasame{v_j} \otimes \ldots
    \otimes \ketbrasame{v_{2N}},\\
    &\ketbrasame{(-)^{v_{j:2N}}} \coloneqq
    \1_{2}^{\otimes(j-1)} \otimes \ketbrasame{(-)^{v_j}}\otimes \ldots
    \otimes \ketbrasame{(-)^{v_{2N}}}.
\end{aligned}
\end{equation}

For $i,\theta\in [2N]$ and $v\in \{0,1\}^{2N}$, we define 
\begin{equation}\label{eq:chizeta_prime_definition}
\begin{alignedat}{2}
    &\chi'(\theta,i,v) &&\coloneqq \norm{X_i - \calVdag ( \sigmax_i \otimes \1) \calV}^2_{\sigma^{\theta,v}},
    \\
    &\chidiamond'(i,v) &&\coloneqq \norm{\tildeX_i\tildeZ_{N+i} - \calVdag ( \sigmax_i \otimes \sigmaz_{N+i} \otimes \1) \calV}^2_{\sigma^{\diamond,v}},
    \\
    &\zetadiamond'(i,v) &&\coloneqq \norm{\tildeZ_i\tildeX_{N+i} - \calVdag ( \sigmaz_i \otimes \sigmax_{N+i} \otimes \1) \calV}^2_{\sigma^{\diamond,v}},
\end{alignedat}
\end{equation}

so that, by definition,
\begin{equation}
\begin{alignedat}{3}
    &X_i &&\approx_{\chi'(\theta,i,v), \sigma^{\theta,v}} && \calVdag (\sigmax_i \otimes \1) \calV,
    \\
    &\tildeX_i\tildeZ_{N+i} &&\approx_{\chidiamond'(i,v),\sigma^{\diamond,v}} && \calVdag ( \sigmax_i \otimes \sigmaz_{N+i} \otimes \1) \calV,
    \\
    &\tildeZ_i\tildeX_{N+i}  &&\approx_{\zetadiamond'(i,v), \sigma^{\diamond,v}} && \calVdag (\sigmaz_i \otimes \sigmax_{N+i} \otimes \1) \calV.
\end{alignedat}
\end{equation}
Similarly to our use of $\chi$ and $\zeta$ introduced in \cref{def:zeta_chi_def}, 
we will use $\chi'$, $\chidiamond'$, and $\zetadiamond'$ to bound distances between the output states of our swap isometry and the ideal states in \cref{lem:swap_states_nondiamond,lem:swap_states_diamond}.

By \cref{lem:swap_observables} and \cref{lem:swap_observables_tildeprod}, we see that, for all $i, \theta\in [2N]$,
\begin{equation}\label{eq:chiprime_bound}
    \sum_v \chi'(\theta,i,v) = O(N\sqrt{\gammaH}),
    \quad
    \sum_v \chidiamond'(i,v) = O(N^{1/4}\gammaH^{1/8}),
    \quad 
    \text{and}
    \quad
    \sum_v \zetadiamond'(i,v) = O(N^{1/4}\gammaH^{1/8}).
\end{equation}
With the above definitions in place, we can prove the following lemma.
\begin{lemma}\label{lem:swap_states_nondiamond}
    Let $D$ be an efficient perfect device. For all $\theta\in \thetasetint$ and $v\in \{0,1\}^{2N}$, there exists a positive semi-definite operator $\alpha^{\theta,v} \in \PosH$ such that
    \begin{equation}
        \sum_{v} \normone{\calV \sigma^{\theta, v} \calVdag - \tau^{\theta,v} \otimes \alpha^{\theta,v}} \leq O(N^{3/2} \gammaH^{1/4}).
    \end{equation}
\end{lemma}

\begin{proof}
    It suffices to prove the lemma for $\theta = 1$ as the proof for other $\theta$ is analogous (the proof for $\theta = 0$ is also simpler). First, we prove the following claim using \cref{lem:swap_observables}.
    \begin{claim}\label{claim:swap_states_claim}
    \begin{enumerate}
        \item There exists a positive semi-definite operator $\beta^v \in \PosH$, such that
        \begin{equation}
            \norm{\calV\sigma^{1,v}\calVdag - \ketbrasame{(-)^{v_1}} \otimes \beta^v}_1 \leq \epsilon(v),
        \end{equation}
        where $\epsilon(v) \coloneqq O\Bigl(\bigl(\chi(1,v)+\sqrt{\tr[\sigma^{1,v}]\cdot \chi'(1,1,v)}\bigr)^{1/2}\sqrt{\tr[\sigma^{1,v}]}\Bigr)$. 
        \item For all $k\in [2N]$ with $k\geq 2$, we have $\norm{\ketbra{v_k}{v_k}_k \otimes \1-\1}_{\calV\sigma^{1,v}\calVdag}^2=\zeta(1,k,v)$.
    \end{enumerate}
    \end{claim}

    \begin{proof}[Proof of \cref{claim:swap_states_claim}]
        To prove the first part, consider
        \begin{equation}
          \begin{aligned}
            \norm{\ketbrasame{(-)^{v_1}}_1 \otimes \1 - \1}_{\calV\sigma^{1,v}\calVdag}^2 
            &=\tr[\calV \sigma^{1,v} \calVdag  ]- \tr[(\ketbrasame{(-)^{v_1}}_1 \otimes \1) \calV \sigma^{1,v} \calVdag ]\\
            &= \tr[\sigma^{1,v}] - \tr[\calVdag (\ketbrasame{(-)^{v_1}}_1 \otimes \1) \calV \sigma^{1,v}] \\
            &\leq \tr[\sigma^{1,v}] - \tr[X_1^{(v_1)} \sigma^{1,v}] + O(\sqrt{\tr[\sigma^{1,v}]\cdot \chi'(1,1,v)})\\
            &= \chi(1,v)/4 + O(\sqrt{\tr[\sigma^{1,v}]\cdot \chi'(1,1,v)}),
        \end{aligned}  
        \end{equation}
        where the inequality uses the replacement lemma (\cref{lem:replace}) with the definition of $\chi'(1,1,v)$ in \cref{eq:chizeta_prime_definition} and \cref{lem:mv_224}, and the last equality is by \cref{eq:zeta_chi_as_norm}. 

        Using \cref{lem:operator_to_state} with the above equation, we obtain
        \begin{equation}
            \normone{\calV\sigma^{1,v}\calVdag - (\ketbrasame{(-)^{v_1}}_1 \otimes \1) \calV\sigma^{1,v}\calVdag
            (\ketbrasame{(-)^{v_1}}_1 \otimes \1)} \leq \epsilon(v),
        \end{equation}
        where
        \begin{equation}\label{eq:def_epsi_v}
            \epsilon(v) \coloneqq O\Bigl(\bigl(\chi(1,v)+\sqrt{\tr[\sigma^{1,v}]\cdot \chi'(1,1,v)}\bigr)^{1/2}\sqrt{\tr[\sigma^{1,v}]}\Bigr).
        \end{equation}
        But $(\ketbrasame{(-)^{v_1}}_1 \otimes \1)\calV\sigma^{1,v}\calVdag (\ketbrasame{(-)^{v_1}}_1 \otimes \1)$ is of the form $\ketbrasame{(-)^{v_1}} \otimes \beta^v$ for
        \begin{align}
            \beta^v \coloneqq (\bra{(-)^{v_1}}_1\otimes \1) \calV\sigma^{1,v}\calVdag 
        (\ket{(-)^{v_1}}_1 \otimes \1).
        \end{align}
        Hence the first part of the claim.

        The second part is simpler. Using the first equation of \cref{lem:swap_observables} and \cref{eq:zeta_chi_as_norm} ($k\geq 2$), we see
        \begin{equation}
            \norm{\ketbrasame{v_k}_k \otimes \1-\1}_{\calV\sigma^{1,v}\calVdag}^2 \\
            = \tr[\sigma^{1,v}] - \tr[\calVdag (\ketbrasame{v_k}_k \otimes \1) \calV\sigma^{1,v}] \\
            = \tr[\sigma^{1,v}] - \tr[Z_k^{(v_k)}\sigma^{1,v}] = \zeta(1,k,v),
        \end{equation}
        which completes the proof.
    \renewcommand{\qedsymbol}{$\square \, \text{Proof of \cref{claim:swap_states_claim}}.$}
    \end{proof}

    We proceed to use the claim to establish the lemma. For all $k\in [2N]$ with $k\geq 2$, consider
    \begin{align*}
        &\normone{\ketbrasame{(-)^{v_1}}_1 \otimes \beta^v - (\ketbrasame{v_k}_k \otimes \1) \left(\ketbrasame{(-)^{v_1}}_1 \otimes \beta^v \right) \, (\ketbrasame{v_k}_k \otimes \1)} 
        \\
        &\leq \normone{\ketbrasame{(-)^{v_1}}_1 \otimes \beta^v - \calV \sigma^{1,v} \calVdag} + \normone{\calV \sigma^{1,v} \calVdag- (\ketbrasame{v_k}_k \otimes \1) \, \calV \sigma^{1,v} \calVdag\, (\ketbrasame{v_k}_k \otimes \1)} \\
        &\quad\quad\quad\quad + \normone{(\ketbrasame{v_k}_k \otimes \1) \, \calV \sigma^{1,v} \calVdag\, (\ketbrasame{v_k}_k \otimes \1) - (\ketbrasame{v_k}_k \otimes \1) \, \ketbrasame{(-)^{v_1}}_1 \otimes \beta^v \, (\ketbrasame{v_k}_k \otimes \1)}
        \\
        &\leq 2\normone{\ketbrasame{(-)^{v_1}}_1 \otimes \beta^v - \calV \sigma^{1,v} \calVdag} + \normone{\calV \sigma^{1,v} \calVdag- (\ketbrasame{v_k}_k \otimes \1) \, \calV \sigma^{1,v} \calVdag\, (\ketbrasame{v_k}_k \otimes \1)}
        \\
        \shortintertext{using the first part of \cref{claim:swap_states_claim} to bound the first term and the second part of \cref{claim:swap_states_claim} with \cref{lem:operator_to_state} to bound the second term:
        }
        &\leq 2\epsilon(v) + 2\sqrt{\zeta(1,k,v)\cdot \tr[\calV \sigma^{1,v} \calVdag]}=2\epsilon(v) + 2\sqrt{\zeta(1,k,v) \cdot \tr[\sigma^{1,v}]}.
    \end{align*}
    
    That is, for all $k\in [2N]$ with $k\geq 2$, we have
    \begin{equation}\label{eq:swap_states_induct}
        \normone{\ketbrasame{(-)^{v_1}}_1 \otimes \beta^v - (\ketbrasame{v_k}_k \otimes \1) \, \ketbrasame{(-)^{v_1}}_1 \otimes \beta^v \, (\ketbrasame{v_k}_k \otimes \1)} \leq 2\epsilon(v) + 2\sqrt{\tr[\sigma^{1,v}] \cdot \zeta(1,k,v)}.
    \end{equation}
    We then have \begin{align*}\label{eq:sigmaonev_closeto}
       &\normone{\calV \sigma^{1,v} \calVdag - (\ketbrasame{v_{2:2N}}\otimes\1) \, (\ketbrasame{(-)^{v_1}}_1 \otimes \beta^v) \, (\ketbrasame{v_{2:2N}}\otimes\1)} 
       \shortintertext{using the triangle inequality:}
       &\leq \normone{\calV \sigma^{1,v} \calVdag -(\ketbrasame{(-)^{v_1}}_1 \otimes \beta^v)} 
        \\
       &\quad\quad+ \normone{(\ketbrasame{(-)^{v_1}}_1 \otimes \beta^v) - (\ketbrasame{v_{2:2N}}\otimes\1) \, (\ketbrasame{(-)^{v_1}}_1 \otimes \beta) \, (\ketbrasame{v_{2:2N}}\otimes\1)} 
       \shortintertext{using \cref{claim:swap_states_claim}:}
       &\leq \epsilon(v) + \normone{(\ketbrasame{(-)^{v_1}}_1 \otimes \beta^v)\\
       &\quad\quad- (\ketbrasame{v_{2:2N}}\otimes\1) \, (\ketbrasame{(-)^{v_1}}_1 \otimes \beta^v) \, (\ketbrasame{v_{2:2N}}\otimes\1)}
       \shortintertext{using \cref{lem:linear_error_accumulation} with \cref{eq:swap_states_induct}:}
       &\leq (2N-1)\, \epsilon(v) + 2\, \sum_{k=2}^N \sqrt{\tr[\sigma^{1,v}]\cdot \zeta(1,k,v)}.
    \end{align*}

    Finally, using the properties of $\{\chi',\chi,\zeta\}$ given in \cref{eq:chiprime_bound,lem:zeta_chi_bounds_new}, the definition of $\epsilon(v)$ in \cref{eq:def_epsi_v}, and the Cauchy-Schwarz inequality, we obtain
    \begin{equation*}
    \begin{aligned}
        &\sum_{v\in \{0,1\}^{2N}} \normone{\calV \sigma^{1,v} \calVdag -(\ketbrasame{v_{2:2N}}\otimes\1) \, (\ketbrasame{(-)^{v_1}}_1 \otimes \beta^v) \, 
        (\ketbrasame{v_{2:2N}}\otimes\1)}
        \\
        &\leq O\Bigl(N \sum_v\epsilon(v) + \sum_v \sum_{k=2}^{2N} \sqrt{\tr[\sigma^{1,v}]\cdot \zeta(1,k,v)} \Bigr)
        \\
        &\leq O\Bigl(N \sqrt{\gammaH} + N \, \sum_v \bigl(\chi(1,v)+\sqrt{\tr[\sigma^{1,v}]\cdot \chi'(1,1,v)}\bigr)^{1/2}\sqrt{\tr[\sigma^{1,v}]} \Bigr)
        \\
        &\leq O\Bigl( N \sqrt{\gammaH} + N \Bigl(\sum_v \chi(1,v)+\sqrt{\tr[\sigma^{1,v}]\cdot \chi'(1,1,v)}\Bigr)^{1/2}  \Bigr)
        \\
        &\leq O\Bigl(N (\gammaH + N\sqrt{\gammaH}+\negl)^{1/2} + N\sqrt{\gammaH}\Bigr) \leq O(N^{3/2} \gammaH^{1/4}).
    \end{aligned}
    \end{equation*}
    Noting that $(\ketbrasame{v_{2:2N}}\otimes\1) \, (\ketbrasame{(-)^{v_1}}_1 \otimes \beta^v) \, 
    (\ketbrasame{v_{2:2N}}\otimes\1)$ is of the form
    \begin{equation}
        \tau^{1,v} \otimes \alpha^{1,v} = \ketbrasame{(-)^{v_1}}\otimes \ketbrasame{v_2} \otimes \cdots \otimes \ketbrasame{v_{2N}} \otimes \alpha^{1,v},
    \end{equation}
    for some positive semi-definite operator $\alpha^{1,v}$, establishes the lemma.
\renewcommand{\qedsymbol}{$\square \, \text{Proof of \cref{lem:swap_states_nondiamond}}.$}
\end{proof}

\cref{lem:swap_states_nondiamond} characterizes the states $\sigma^{\theta,v}$ for $\theta\in \{0\} \cup [2N]$ . For $\theta=\diamond$, we have
\begin{lemma}\label{lem:swap_states_diamond}
    Let $D$ be an efficient perfect device.  For all $v\in \{0,1\}^{2N}$, there exists a positive semi-definite operator $\alpha^{\diamond, v} \in \PosH$ such that 
    \begin{equation}
        \sum_{v}\normone{\calV \sigma^{\diamond, v} \calVdag - \tau^{\diamond,v} \otimes \alpha^{\diamond, v}} \leq O(N^{17/16}\gammaH^{1/32}).
    \end{equation}

\end{lemma}
\begin{proof}
    The proof is similar to the proof of \cref{lem:swap_states_nondiamond}. We provide the details for completeness.

    We will use the fact that, for all $v\in \{0,1\}^{2N}$, we have
    \begin{equation}\label{eq:diamond_projectorform}
        \ketbrasame{\psi^{v}}\otimes \1 = \prod_{i=1}^N(\sigmaz_i\otimes \sigmax_{N+i}\otimes \1)^{(v_i)}(\sigmax_i \otimes \sigmaz_{N+i} \otimes \1)^{(v_{N+i})}.
    \end{equation}
    
    For all $i \in [N]$, we have
    \begin{equation}\label{eq:xz_approx_id_diamond}
    \begin{aligned}
        &\quad\norm{(\sigmax_i \otimes \sigmaz_{N+i} \otimes \1)^{(v_{N+i})}\otimes \1 - \1}^2_{\calV\sigma^{\diamond,v} \calVdag} 
        \\
        &= \tr[\calV \sigma^{\diamond,v} \calVdag] - \tr[(\sigmax_i \otimes \sigmaz_{N+i} \otimes \1)^{(v_{N+i})} \calV \sigma^{\diamond, v} \calVdag] 
        \\
        &= \tr[\sigma^{\diamond,v}] - \tr[\calVdag (\sigmax_i \otimes \sigmaz_{N+i} \otimes \1)^{(v_{N+i})} \calV \sigma^{\diamond, v}]
        \\
        &\leq  \tr[\sigma^{\diamond,v}] - \tr[(\tildeX_i \tildeZ_{N+i})^{(v_{N+i})} \sigma^{\diamond, v}] + O(\sqrt{\tr[\sigma^{\diamond,v}]\cdot \chidiamond'(i,v)})
        \\
        &= \chidiamond(i,v) + O(\sqrt{\tr[\sigma^{\diamond,v}]\cdot \chidiamond'(i,v)}),
    \end{aligned}
    \end{equation}
    where the inequality uses the replacement lemma (\cref{lem:replace}) with the definition of $\chidiamond'$ in \cref{eq:chizeta_prime_definition} and \cref{lem:mv_224}, and the last equality is by \cref{eq:zeta_chi_as_norm}.
    Similarly, we have
    \begin{equation}\label{eq:zx_approx_id_diamond}
        \norm{(\sigmaz_i \otimes \sigmax_{N+i} \otimes \1)^{(v_i)}\otimes \1 - \1}^2_{\calV\sigma^{\diamond,v} \calVdag}  \leq \zetadiamond(i,v) + O\Bigl(\sqrt{\tr[\sigma^{\diamond,v}]\cdot \zetadiamond'(i,v)}\Bigr).
    \end{equation}

    Using \cref{lem:operator_to_state} with \cref{eq:xz_approx_id_diamond,eq:zx_approx_id_diamond}, we obtain
    \begin{equation}\label{eq:diamond_state_induct}
    \begin{aligned}
        &\normone{\calV\sigma^{\diamond,v} \calVdag - (\sigmax_i \otimes \sigmaz_{N+i} \otimes \1)^{(v_{N+i})} \calV\sigma^{\diamond,v} \calVdag
        (\sigmax_i \otimes \sigmaz_{N+i} \otimes \1)^{(v_{N+i})}} \leq \epsilon_1(i,v),
        \\
        &\normone{\calV\sigma^{\diamond,v} \calVdag - (\sigmaz_i \otimes \sigmax_{N+i} \otimes \1)^{(v_i)} \calV\sigma^{\diamond,v} \calVdag
        (\sigmaz_i \otimes \sigmax_{N+i} \otimes \1)^{(v_i)}} \leq \epsilon_2(i,v),
    \end{aligned}
    \end{equation}
    where
    \begin{equation}\label{eq:diamond_stiate_epsilon_def}
    \begin{aligned}
        &\epsilon_1(i,v) \coloneqq O\Bigl( \bigl(\chidiamond(i,v) + \sqrt{\tr[\sigma^{\diamond,v}]\cdot \chidiamond'(i,v)} \bigr)^{1/2}\sqrt{\tr[\sigma^{\diamond,v}]}\Bigr),
        \\
        &\epsilon_2(i,v) \coloneqq O\Bigl( \bigl(\zetadiamond(i,v) + \sqrt{\tr[\sigma^{\diamond,v}]\cdot \zetadiamond'(i,v)} \bigr)^{1/2}\sqrt{\tr[\sigma^{\diamond,v}]}\Bigr).
    \end{aligned}
    \end{equation}

    Using \cref{eq:diamond_projectorform} and \cref{lem:linear_error_accumulation} with \cref{eq:diamond_state_induct} gives
    \begin{equation}
        \normone{\calV\sigma^{\diamond,v}\calVdag - \ketbrasame{\psi^{v}} \calV\sigma^{\diamond,v}\calVdag\ketbrasame{\psi^{v}}} \leq \sum_{i=1}^N \epsilon_1(i,v)+\epsilon_2(i,v).
    \end{equation}

    Therefore, summing over $v$ gives
    \begin{equation}
        \sum_v \normone{\calV\sigma^{\diamond,v}\calVdag - \ketbrasame{\psi^{v}} \calV\sigma^{\diamond,v}\calVdag\ketbrasame{\psi^{v}}}
        \leq
        \sum_{i=1}^N \sum_v \epsilon_1(i,v)+\epsilon_2(i,v) \leq O(N^{17/16}\gammaH^{1/32}).
    \end{equation}
    But $\ketbrasame{\psi^{v}} \calV\sigma^{\diamond,v}\calVdag\ketbrasame{\psi^{v}}$ is of the form $\tau^{\diamond,v}\otimes\alpha^{\diamond,v}$ for some positive semi-definite operator $\alpha^{\diamond,v}$. Hence the first part of the lemma.

\end{proof}
Lastly, we show there is not much information about $\theta$
and $v$ in $\alpha^{\theta,v}$. The proof is similar to that of \cite[Proposition 4.32 (arXiv version v2)]{gmp2022parallelrsp} but we include it for completeness\footnote{In v1 of this paper, we provided an independent proof of this lemma (see Lemmas 4.37 and 4.39 in v1) but it gave a worse bound. To be clear, we learned the main ideas of the current proof from \cite[Proposition 4.32 (arXiv version v2)]{gmp2022parallelrsp}.}.
\begin{lemma}\label{lem:ci_aux}
    Let $D$ be an efficient perfect device. There exists a state $\alpha \in \PosH$ such that the following holds. For all $\theta\in \thetaset$ and $v\in \{0,1\}^{2N}$, let $\alpha^{\theta,v}\in \PosH$ be as defined in \cref{lem:swap_states_nondiamond,lem:swap_states_diamond}. Then, there exist numbers $\{\delta(v) \geq 0 \mid v\in \{0,1\}^{2N}\}$ such that 
    $\alpha^{\theta,v} \capprox_{\delta(v)} \alpha/2^{2N}$ for all $v\in \{0,1\}^{2N}$ and
    \begin{align*}
        \sum_{v \in \{0,1\}^{2N}} \delta(v) \leq O(N^{3/2} \gammaH^{1/32}).
    \end{align*}
\end{lemma}
\begin{proof}
    In this proof, we write 
    \begin{equation}
        \epsilon \coloneqq \max_{\theta\in \thetaset} \sum_v \normone{\calV \sigma^{\theta,v}\calV^\dagger - \tau^{\theta,v}\otimes \alpha^{\theta,v}} \leq O(N^{3/2} \gammaH^{1/32}),
    \end{equation}
    where the inequality uses  \cref{lem:swap_states_nondiamond,lem:swap_states_diamond}.
    
    For all $\theta\in \thetaset$, we have
    \begin{equation}\label{eq:ci_epsilon_negl}
        \norm{\calV \sigma^\theta \calVdag - \sum_v \tau^{\theta,v} \otimes \alpha^{\theta,v}} \leq \sum_v 
        \norm{\calV \sigma^{\theta,v}\calVdag - \tau^{\theta,v} \otimes \alpha^{\theta,v} } + \negl \leq \epsilon + \negl,
    \end{equation}
    where the first inequality uses \cref{lem:sigma_equal_sumv} and the second inequality uses the definition of $\epsilon$.
    
    Let $\alpha^\diamond \coloneqq \sum_{v} \alpha^{\diamond,v}$.
    We first prove
    \begin{equation}
        \label{eq:ci_0_diamond}
        \sum_{v} \tau^{0,v} \otimes \alpha^{0,v} 
        \csimeq_{2\epsilon + \negl}
        \sum_{v} \tau^{0,v} \otimes  \alpha^\diamond/2^{2N}
    \end{equation}
    by contradiction. Suppose \cref{eq:ci_0_diamond} does not hold. Then there exists an efficient POVM $\{\Lambda, \1-\Lambda\}$ such that
    \begin{equation}\label{eq:ci_contradiction_hypothesis}
        \tr\Bigl[\Lambda \Bigl( \sum_{v} \tau^{0,v} \otimes \alpha^{0,v} - \sum_{v} \tau^{0,v} \otimes \alpha^\diamond/2^{2N} \Bigr) \Bigr] > 2\epsilon + \notnegl
    \end{equation}
    for some non-negligible function $\notnegl > 0$.
    Since the two states on either side of \cref{eq:ci_0_diamond} have a block-diagonal structure, we can without-loss-of-generality assume that $\Lambda$ has the same structure. That is,
    \begin{equation}
        \Lambda = \sum_v \tau^{0,v} \otimes \Lambda_v \quad \text{for some $\Lambda_v$}.
    \end{equation}
    Hence, our assumption in \cref{eq:ci_contradiction_hypothesis} is equivalent to 
    \begin{equation}\label{eq:ci_equiv_assumption}
        \sum_v \tr\Bigl[\Lambda_v(\alpha^{0,v} - \alpha^{\diamond}/2^{2N} )\Bigr] \geq 2\epsilon + \notnegl.
    \end{equation}
    Define a new POVM $\{\Gamma, \1- \Gamma\}$ by
    \begin{equation}
        \Gamma \coloneqq \calV^{\dagger} \left( \sum_v \tau^{0,v} \otimes \Lambda_v\right) \calV,
    \end{equation}
    which is efficient since $\calV$ is efficient.
    
    Then, we obtain a contradiction as follows:
    \begin{align*}
        \negl &\geq \tr[\Gamma(\sigma^0 - \sigma^\diamond)] 
        &&\text{($\sigma^0 \capprox \sigma^\diamond$)}
        \\
        &= \tr\left[ \left( \sum_v \tau^{0,v} \otimes \Lambda_v\right)\left(\calV \sigma^0 \calVdag -
        \calV \sigma^\diamond \calVdag\right)\right] 
        &&\text{(trace is cyclic)}
        \\
        &\geq \tr \left[ \left( \sum_v \tau^{0,v} \otimes \Lambda_v\right)\left(\sum_v \tau^{0,v} \otimes \alpha^{0,v} -
        \sum_u \tau^{\diamond,u} \otimes \alpha^{\diamond,u} \right)\right] 
        \\
        &\quad \quad - 2 \epsilon - 2\, \negl &&\text{(\cref{eq:ci_epsilon_negl})}
        \\
        &= \sum_{v} \tr[\Lambda_v \alpha^{0,v}]-
        \frac{1}{2^{2N}}\sum_v \tr\Bigl[ \Lambda_v \Bigl(\sum_u \alpha^{\diamond,u}\Bigr)\Bigr] - 2\epsilon - 2\,\negl
        &&\amatrix{\text{$\tr[ \tau^{0,v}\tau^{\diamond,u}] = 1/2^{2N}$}}{\text{for all $u,v \in \{0,1\}^{2N}$}}
        \\
        &= \sum_v \tr[ \Lambda_v(\alpha^{0,v} - \alpha^{\diamond}/2^{2N})] - 2\epsilon - 2\,\negl
        &&\text{(definition of $\alpha^\diamond$)}
        \\
        &\geq \notnegl - 2\,\negl &&\text{(\cref{eq:ci_equiv_assumption})},
    \end{align*}
    which implies $3 \, \negl \geq \notnegl$, a contradiction. Hence we have proved \cref{eq:ci_0_diamond}.
    
    Now, for any $\theta \in \thetaset$, we have
    \begin{equation}\label{eq:ci_chain}
        \sum_v \tau^{\theta,v} \otimes \alpha^{\theta,v}
        \capprox_{\epsilon + \negl}
        \calV \sigma^\theta \calVdag
        \capprox_\negl
        \calV \sigma^{0} \calVdag
        \capprox_{\epsilon+\negl}
        \sum_v \tau^{0,v} \otimes \alpha^\diamond / 2^{2N}
        = \sum_v \tau^{\theta,v} \otimes \alpha^\diamond / 2^{2N},
    \end{equation}
    where we used \cref{eq:ci_epsilon_negl} in the first approximation, the computational indistinguishability of the $\sigma^\theta$s (\cref{lem:indistinguishability}) in the second approximation, \cref{eq:ci_0_diamond} in the third approximation, and the fact that $\sum_v \tau^{\theta,v} = \1$ for all $\theta \in \thetaset$ in the last equality.
    
    Setting $\alpha \coloneqq \alpha^\diamond$,
    the lemma follows from \cref{eq:ci_chain}, the triangle inequality for computational indistinguishability (\cref{lem:ci_triangle}), and noting that
    the computational distinguishability between 
    $\sum_v \tau^{\theta,v} \otimes \alpha^{\theta,v}$
    and $\sum_v \tau^{\theta,v} \otimes \alpha^\diamond/2^{2N}$ can be expressed as the sum-over-$v$ of the computational distinguishabilities between $\alpha^{\theta,v}$ and $\alpha^\diamond/2^{2N}$, which we set as $\delta(v)$.
\end{proof}

\subsection{Soundness for states and measurements}
\label{sec:soundness_states_meas}

We now put everything together to give our main theorem, \cref{thm:soundness}.

First, we need to do some bookkeeping. In most of the preceding analysis, we have considered the device as an efficient \textit{perfect} device in the knowledge that the states $\psi^{\theta}$ of the actual device are equal to those of a perfect device up to error $\sqrt{\gammaP}$ in trace distance, see~\cref{prop:reduce_to_pd}. In fact, $\gammaP$ also controls how close the states $\sigma^{\theta,v}$ of the actual device are to those of a perfect device:
\begin{lemma}\label{lem:sigma_reduce_to_pd}
    Let $D$ be an efficient device and $\widetilde{D}$ the perfect device associated with $D$ according to \cref{prop:reduce_to_pd}. Respectively, let $\sigma^{\theta,v},\tildesigma^{\theta,v}\in \PosH$ be the states prepared by $D$ and $\widetilde{D}$ in the Hadamard round as defined in \cref{def:sig_thetav}. Then, for all $\theta\in \thetaset$, we have
    \begin{equation}
        \sum_{v\in \{0,1\}^{2N}} \normone{\sigma^{\theta,v}-\tildesigma^{\theta,v}} \leq \sqrt{\gammaP} + 2 \gammaP. 
    \end{equation}
\end{lemma}
\begin{proof}
    Since $\sigma^{\theta}$ is obtained from $\psi^{\theta}$ by a quantum channel, which cannot increase trace distance, \cref{prop:reduce_to_pd} implies $\normone{\sigma^\theta-\tildesigma^\theta}\leq \sqrt{\gammaP}$. As the operators in $\{\sigma^{\theta,v} - \tildesigma^{\theta,v}|v\in \{0,1\}^{2N}\}$ are Hermitian and  pairwise multiply to $0$ by \cref{def:sig_thetav}, we have 
    \begin{equation}
        \sum_{v\in \{0,1\}^{2N}} \normone{\sigma^{\theta,v}-\tildesigma^{\theta,v}} = \bignorm{\sum_{v\in \{0,1\}^{2N}} (\sigma^{\theta,v}-\tildesigma^{\theta,v})}_1 \leq  \normone{\sigma^{\theta} - \tildesigma^{\theta}} + 2\gammaP \leq \sqrt{\gammaP} + 2 \gammaP,
    \end{equation}
    where the first equality uses \cref{lem:normone_orthogonal} and the first inequality uses \cref{lem:sum_sigma_v}. Hence the lemma.
\end{proof}

To state our theorem, we recall the definition of $\tau^{\theta, v}$ from \cref{def:tau_thetav}: $\tau^{\theta, v} \coloneqq \ketbrasame{\tau^{\theta, v}}$, where
\begin{equation}\label{eq:tau_def_repeat}
    \ket{\tau^{\theta, v}} \coloneqq
    \begin{cases}
        \ket{v_1}\otimes \cdots  \otimes \ket{v_{\theta-1}} \otimes \ket{(-)^{v_\theta}} \otimes \ket{v_{\theta+1}} \otimes \cdots \otimes \ket{v_{2N}} &\text{if $\theta\in [2N]$},
        \\
        \ket{v} \coloneqq \ket{v_1} \otimes \cdots \otimes \ket{v_{2N}} &\text{if $\theta = 0$},
        \\
        \ket{\psi^{v}} &\text{if $\theta = \diamond$},
    \end{cases}
\end{equation}
where $\ket{\psi^{v}}$ is as defined in \cref{eq:def_psi}.

We also write, for all $u\in \{0,1\}^{2N}$, and $q\in \{0,1\}^{2N}$, $\ket{B_q^u}$ for the $2N$-qubit state 
\begin{equation}
\label{eq:bb84_product}
    \ket{B_q^u} \coloneqq \ket{B_{q_1}^{u_1}}\ket{B_{q_2}^{u_2}}\cdots\ket{B_{q_{2N}}^{u_{2N}}},
\end{equation}
where $\ket{B_{q_i}^{u_i}}$ denotes the BB84 states as in \cref{eq:bb84_states}.

Note that $\Pi_q^u = \ketbrasame{B_q^u}$ for all $q\in \{0,1\}^{2N}$ and $u\in \{0,1\}^{2N}$, where $\Pi_q^u$ is as defined in \cref{eq:measurements}. We finally stress that the theorem holds under the LWE hardness assumption made throughout this work.

\begin{thm}\label{thm:soundness}
    Let $D$ be an efficient device. Let $\calH \coloneqq \calH_{D}\otimes \calH_{Y} \otimes \calH_{R}$ be the Hilbert space of $D$.  Let $\calV:\calH \rightarrow \mathbb{C}^{2^{2N}}\otimes \calH$ be the swap isometry defined in \cref{def:swap}. For $\theta\in \thetaset$ and $v\in \{0,1\}^{2N}$, let $\sigma^{\theta,v} \in \Pos(\calH)$ be the states that $D$ prepares after returning the first answer in the Hadamard round, as defined in \cref{def:sig_thetav}. Let  $\{ \{P^u_q\}_{u\in\{0,1\}^{2N}} \, | \, q \in \{0,1\}^{2N} \}$ be the measurements $D$ performs to return the second answer in the Hadamard round when asked question $q$.
    
    Suppose that $D$ fails the protocol in \cref{fig:protocol} (with an input distribution $\mu$ on $\{0,1\}^{2N}$ and $N = \poly(\lambda)$) with probability at most $\epsilon$. Then, there exist states $\{\alpha^{\theta,v} \, | \, \theta\in \thetaset,v\in \{0,1\}^{2N}\}$, that are computationally indistinguishable from a single state $\alpha \in \Pos(\calH)$ in the way specified in \cref{lem:ci_aux}, such that
    \begin{align}
        &\sum_{v \in \{0,1\}^{2N}} \normone{\calV \sigma^{\theta,v} \calVdag - \tau^{\theta,v} \otimes \alpha^{\theta,v}} \leq O(N^{7/4}\epsilon^{1/32}),\label{eq:soundness_states}
        \\
        &\mathbb{E}_{q\leftarrow_\mu \{0,1\}^{2N}}\Biggl[\sum_{u,v\in \{0,1\}^{2N}} \normone{\calV P_q^u \sigma^{\theta,v} P_q^u \calV^\dagger - \bra{B_q^u} \tau^{\theta,v} \ket{B_q^u} \ketbrasame{B_q^u} \otimes \alpha^{\theta,v}}\Biggr] \leq O(N^2 \epsilon^{1/32}),
        \label{eq:soundness_measurements}
    \end{align}
    and, for all $q\in \qtestset$,
    \begin{equation}\label{eq:soundness_test_measurements}
        \sum_{u,v\in \{0,1\}^{2N}} \normone{\calV P_q^u \sigma^{\theta,v} P_q^u \calV^\dagger - \bra{B_q^u} \tau^{\theta,v} \ket{B_q^u} \ketbrasame{B_q^u} \otimes \alpha^{\theta,v}} \leq O(N^2 \epsilon^{1/32}).
    \end{equation}
\end{thm}

\begin{proof}
Let $\tildeD$ be the perfect device associated with $D$ according to \cref{prop:reduce_to_pd}. Let $\tildesigma^{\theta,v}$ be the states of $\tildeD$ corresponding to $\sigma^{\theta,v}$ of $D$. 

We first prove \cref{eq:soundness_states}, which characterizes the pre-measurement states $\sigma^{\theta,v}$. We have
\begin{equation}
\begin{aligned}
    \sum_v \normone{\calV \sigma^{\theta,v} \calVdag - \tau^{\theta,v} \otimes \alpha^{\theta,v}} &\leq \sum_v \normone{\calV \sigma^{\theta,v} \calVdag - \calV \tildesigma^{\theta,v} \calVdag}  + \sum_v \normone{\calV \tildesigma^{\theta,v} \calVdag - \tau^{\theta,v} \otimes \alpha^{\theta,v}} 
    \\
    &\leq \sqrt{\gammaP} + 2\gammaP + O(N^{17/16}\gammaH^{1/32} + N^{3/2}\gammaH^{1/4}),
\end{aligned}
\end{equation}
where, in the last inequality, we used \cref{lem:sigma_reduce_to_pd} to bound the first term and took the sum of the bounds in \cref{lem:swap_states_nondiamond,lem:swap_states_diamond} to bound the second term in a way that is independent of $\theta$. 

To get a bound in terms of $N$ and the failure probability, $\epsilon$, we can use \cref{prop:gamma_bound_by_fail}, which says $\gammaP,\gammaH= O(N\epsilon)$. Therefore, we obtain
\begin{equation}
    \sum_{v} \normone{\calV \sigma^{\theta,v} \calVdag - \tau^{\theta,v} \otimes \alpha^{\theta,v}} \leq O(N^{7/4}\epsilon^{1/32}).
\end{equation}
Hence the first equation of the theorem, \cref{eq:soundness_states}. 

We now proceed to prove the second and third equations of the theorem, \cref{eq:soundness_measurements,eq:soundness_test_measurements}, which characterize the post-measurement states $P^u_q\sigma^{\theta,v} P^u_q \calVdag$. We have
\begin{equation}\label{eq:post_meas_reduce_to_pd}
\begin{aligned}
    &\quad \sum_{u,v}\normone{\calV P^u_q \sigma^{\theta,v} P^{u}_q \calVdag - \bra{B_q^u} \tau^{\theta,v} \ket{B_q^u} \ketbrasame{B_q^u} \otimes \alpha^{\theta,v}}
    \\
    &\leq \sum_{u,v}\normone{\calV P^u_q \sigma^{\theta,v} P^{u}_q \calVdag - \calV P^u_q \tildesigma^{\theta,v} P^{u}_q \calVdag} + \normone{\calV P^u_q \tildesigma^{\theta,v} P^{u}_q \calVdag - \bra{B_q^u} \tau^{\theta,v} \ket{B_q^u} \ketbrasame{B_q^u} \otimes \alpha^{\theta,v}}
    \\
    &\leq \sum_{v}\normone{\sigma^{\theta,v} - \tildesigma^{\theta,v}} + \sum_{u,v}\normone{\calV P^u_q \tildesigma^{\theta,v} P^{u}_q \calVdag - \bra{B_q^u} \tau^{\theta,v} \ket{B_q^u} \ketbrasame{B_q^u} \otimes \alpha^{\theta,v}}
    \\
    &\leq \sqrt{\gammaP} + 2\gammaP + \sum_{u,v}\normone{\calV P^u_q \tildesigma^{\theta,v} P^{u}_q \calVdag - \bra{B_q^u} \tau^{\theta,v} \ket{B_q^u} \ketbrasame{B_q^u} \otimes \alpha^{\theta,v}},
\end{aligned}
\end{equation}
where the second inequality uses part $4$ of the replacement lemma (\cref{lem:replace}) and the third inequality uses \cref{lem:sigma_reduce_to_pd}.

\Cref{eq:post_meas_reduce_to_pd} allows us to assume that $D$ is already perfect and add on $\sqrt{\gammaP}+2\gammaP$ to the bound we derive on the first expression in \cref{eq:post_meas_reduce_to_pd} only at the end. Henceforth, we assume $D$ is perfect until indicated otherwise.

To prove the second part of the theorem, it primarily remains to prove
\begin{equation}\label{eq:q_measurement_bound}
    \sum_{u,v}\normone{\calV P_q^u \sigma^{\theta,v}P_q^u \calVdag - \bra{B_q^u} \tau^{\theta,v} \ket{B_q^u} \ketbrasame{B_q^u} \otimes \alpha^{\theta,v}} \leq O(N^{3/2} \gammaH^{1/32} + \sqrt{N} \gammaHq^{1/8}),
\end{equation}
for all $q\in \{0,1\}^{2N}$. In the following, we prove \cref{eq:q_measurement_bound} for $q = 1^{2N}$. After the proof, we explain how this proof can be slightly modified to prove \cref{eq:q_measurement_bound} for $q \neq 1^{2N}$.

Consider \cref{eq:q_measurement_bound} for $q = 1^{2N}$. For $u\in \{0,1\}^{2N}$, we have
\begin{equation}\label{eq:Pu_prod_Xu}
    P_{1^{2N}}^u = \prod_{i\in[2N]} X_i^{(u_i)}.
\end{equation}
Now, \cref{lem:swap_observables} and \cref{lem:mv_223} gives: for all $i\in [2N]$ and $\theta\in \thetaset$, $\calV X_i \calVdag \approx_{\delta,\calV\sigma^{\theta}\calVdag} \sigmax_i \otimes \1$, where
\begin{equation}
    \delta \coloneqq (N\sqrt{\gammaH})^{1/2}.
\end{equation}
Therefore, by \cref{lem:mv_224}, we have
\begin{equation}\label{eq:vxvproj_approx_sigmaxproj}
    \calV X_i^{(u_i)} \calVdag \approx_{\delta,\calV\sigma^{\theta}\calVdag} \ketbrasame{(-)^{u_i}}_i \otimes \1.
\end{equation}

We write
\begin{equation}
\begin{aligned} 
    &X_j[u] \coloneqq X_{2N}^{(u_{2N})} X_{2N-1}^{(u_{2N-1})} \cdots X_j^{(u_j)} = X_j^{(u_j)}X_{j+1}^{(u_{j+1})} \cdots X_{2N}^{(u_{2N})},
    \\
    &\ketbrasame{(-)^u} \coloneqq \ketbrasame{(-)^{u_1}}_1\otimes \cdots \otimes \ketbrasame{(-)^{u_{2N}}}_{2N},
\end{aligned}
\end{equation}
and recall $\ketbrasame{(-)^{u_{j:2N}}} \coloneqq \1_{2}^{\otimes(j-1)} \otimes \ketbrasame{(-)^{u_j}} \otimes \ldots \otimes \ketbrasame{(-)^{u_{2N}}}$ from \cref{def:comp_hadamard_projs}.

Then, by direct calculation, we have 
\begin{equation}\label{eq:proj_diff}
\begin{aligned}
   &\norm{\calV X_1[u] \calVdag - \ketbrasame{(-)^{u}}\otimes \1}^2_{\calV \sigma^{\theta}\calVdag} 
   \\
   = &\, \tr[X_1[u]\sigma^{\theta}] + \tr[(\ketbrasame{(-)^{u} }\otimes \1)\calV \sigma^{\theta}\calVdag]
   - 2 \real \tr[(\ketbrasame{(-)^{u} }\otimes \1)\calV X_1[u] \sigma^{\theta} \calVdag].
\end{aligned}
\end{equation}
We proceed to show that \cref{eq:proj_diff} summed over $u$ is close to zero. We do this by showing that the vector 
\begin{equation}
    \sum_u \real \tr[(\ketbrasame{(-)^{u} }\otimes \1)\calV X_1[u] \sigma^{\theta} \calVdag] \ket{u}
\end{equation}
is approximately equal to both 
\begin{equation}    
    \sum_u \tr[X_1[u]\sigma^{\theta}] \ket{u} \quad \text{and} \quad \sum_u \tr[(\ketbrasame{(-)^{u} }\otimes \1)\calV \sigma^{\theta}\calVdag] \ket{u}
\end{equation}
in $\ell_1$-norm distance.

In the following, we recall that, for vectors $x,y \in \mathbb{C}^n$, we defined $x \approx_{\epsilon} y$ to mean $\normone{x-y} \leq O(\epsilon)$, where $\normone{\cdot}$ denotes the $\ell_1$-norm. We give justifications for the steps below that do not involve an exact equality afterwards. The exact equalities follow from $\calV^\dagger\calV = \1$, $X_1[u] = X_2[u]X_1^{(u_1)}$, $\ketbrasame{(-)^{u_{1}}}_1 \cdot \ketbrasame{(-)^{u}} = \ketbrasame{(-)^{u_{1}}}_1 \cdot \ketbrasame{(-)^{u_{2:2N}}}$, and the cyclicity of the trace.

\begin{align}
    &\sum_u \real \tr[(\ketbrasame{(-)^{u} }\otimes \1)\calV X_1[u] \sigma^{\theta} \calVdag]  \ket{u} \label{eq:measurements_first_line}
    \\
    &= \sum_u \real \tr[ (\ketbrasame{(-)^{u}} \otimes \1) (\calV X_2[u] X_1^{(u_1)}\calVdag) (\calV  \sigma^{\theta}\calVdag)] \ket{u}
    \notag
    \\
    &\simeq_{\negl} \sum_u \real \tr[ (\ketbrasame{(-)^{u} }\otimes \1) (\calV X_2[u] X_1^{(u_1)}\calVdag) (\calV\sigma^{1}\calVdag)] \ket{u}
    \notag
    \\
    &\approx_{\sqrt{\gammaH}} \sum_u \real \tr[ (\ketbrasame{(-)^{u} }\otimes \1) (\calV X_2[u]  \sigma^{1} X_1^{(u_1)}\calVdag)] \ket{u}
    \notag
    \\
    &= \sum_u \real \tr[(\calV X_1^{(u_1)} \calVdag) (\ketbrasame{(-)^{u} }\otimes \1) (\calV X_2[u] \calVdag) (\calV \sigma^{1} \calVdag)] \ket{u}
    \notag
    \\
    &\approx_{\sqrt{\delta}} \sum_u \real \tr[(\ketbrasame{(-)^{u_1}}_1 \otimes \1) (\ketbrasame{(-)^{u} }\otimes \1) (\calV X_2[u] \calVdag) (\calV \sigma^{1} \calVdag)] \ket{u}
    \notag
    \\
    &= \sum_u \real \tr[ (\ketbrasame{(-)^{u_{1}}}_1 \otimes \1) (\ketbrasame{(-)^{u_{2:2N}} }\otimes \1)(\calV X_2[u] \calVdag) (\calV\sigma^{1}\calVdag)] \ket{u}
    \notag
    \\
    &\simeq_{\negl} \sum_u \real \tr[(\ketbrasame{(-)^{u_{1}}}_1\otimes \1) (\ketbrasame{(-)^{u_{2:2N}} }\otimes \1) (\calV X_2[u] \calVdag) (\calV\sigma^{\theta}\calVdag)] \ket{u}
    \notag
    \\
    &= \sum_u \real \tr[(\ketbrasame{(-)^{u} }\otimes \1)\calV X_2[u] \sigma^{\theta}\calVdag] \ket{u}.
    \notag
\end{align}

We justify the four steps above that do not involve an exact equality as follows:
\begin{enumerate}
    \item In the third line, we used the ``moreover'' part of the lifting-under-projections lemma (\cref{lem:lifting_proj}) with
    $Q^w = \1$, $\Pi^u = \ketbrasame{(-)^{u} }\otimes \1$,
    $P^u = \calV X_1[u] \calVdag$, $\psi = \calV\sigma^\theta\calVdag$, 
    and $\psi' = \calV \sigma^1 \calVdag$.
    
    \item In the fourth line, we used part $4$ of the replacement lemma (\cref{lem:replace}) with $P^u = \ketbrasame{(-)^{u} }\otimes \1$,
    \begin{align*}
        &\psi = \frac{1}{2} \calV (X_2[u]X_1^{(u_1)}\sigma^1 + 
        \sigma^1 X_1^{(u_1)}X_2[u]) \calVdag, \quad \text{and}\\
        &\psi' = \frac{1}{2} \calV (X_2[u] \sigma^1 X_1^{(u_1)}
        + X_1^{(u_1)} \sigma^1 X_2[u])\calVdag.
    \end{align*}
    Note that $\normone{\psi -\psi'} \leq O(\sqrt{\gammaH})$ by $\norminfty{\calV}=\norminfty{X_2[u]}=1$, and the operator-state commutation relation (\cref{prop:operator_state_commutation}).
    
    \item In the sixth line, we used part $3$ of the replacement lemma (\cref{lem:replace}) with
    $A = \calV X_1^{(u_1)} \calVdag$, $B = \ketbrasame{(-)^{u_1}}_1 \otimes \1$, $X_i^{(u_i)} =  \ketbrasame{(-)^{u_i}}_i \otimes \1$ for $i \in [2N]$,
    $Y_i = \calV X_i^{(u_i)} \calVdag$ for $i \in \{2,\ldots,2N\}$,
    and $\psi = \calV \sigma^1 \calVdag$. Importantly, we used $A \approx_{\delta, \psi} B$ from \cref{eq:vxvproj_approx_sigmaxproj}.
    
    \item In the second-to-last line, we again used the ``moreover'' part of the lifting-under-projections lemma (\cref{lem:lifting_proj}), this time with $Q^{u_1} =  \ketbrasame{(-)^{u_1}}_1 \otimes \1$, $\Pi^u = \ketbrasame{(-)^{u_{2:2N}} }\otimes \1$, $P^u = \calV X_2[u] \calVdag$, $\psi = \calV\sigma^1\calVdag$,
    and $\psi' = \calV \sigma^\theta \calVdag$.
\end{enumerate}

Comparing the expressions in the first and last lines (of the equations starting at \cref{eq:measurements_first_line}), we see that we can continue in the same way another $2N-1$ times, such that $\sigma^\theta$ is replaced by $\sigma^k$ by the lifting-under-projections lemma (\cref{lem:lifting_proj}) at time $k$, to obtain
\begin{equation}
    \sum_u \real \tr[(\ketbrasame{(-)^{u} }\otimes \1)\calV X_1[u] \sigma^{\theta} \calVdag]  \ket{u}\approx_{\delta'} \sum_u \real \tr[(\ketbrasame{(-)^{u} }\otimes \1)\calV 
    \sigma^{\theta}\calVdag] \ket{u},
\end{equation}
where
\begin{equation}\label{eq:thm_delta_prime}
    \delta' \coloneqq N(\sqrt{\delta}+\sqrt{\gammaH} + \negl) \leq O(N^{5/4}\gammaH^{1/8}).
\end{equation}

But $\real \tr[(\ketbrasame{(-)^{u} }\otimes \1)\calV 
\sigma^{\theta}\calVdag] = \tr[(\ketbrasame{(-)^{u} }\otimes \1)\calV 
\sigma^{\theta}\calVdag] (\geq 0)$.
Therefore, 
\begin{equation}\label{eq:proj_sumu_1}
    \sum_u \real \tr[(\ketbrasame{(-)^{u} }\otimes \1)\calV X_1[u] \sigma^{\theta} \calVdag] \ket{u} \approx_{\delta'} \sum_u \tr[(\ketbrasame{(-)^{u} }\otimes \1)\calV 
    \sigma^{\theta}\calVdag] \ket{u}.
\end{equation}

By an analogous argument, we can ``merge'' $\ketbrasame{(-)^u}$ into $\calV X_1[u] \calVdag$ (instead of merging $\calV X_1[u] \calVdag$ into $\ketbrasame{(-)^u}$ as done above) to obtain
\begin{equation}\label{eq:proj_sumu_2}
    \sum_u \real \tr[(\ketbrasame{(-)^{u} }\otimes \1)\calV X_1[u] \sigma^{\theta} \calVdag] \ket{u} \approx_{\delta'} \sum_u \tr[X_1[u] \sigma^\theta] \ket{u}.
\end{equation}

Substituting \cref{eq:proj_sumu_1,eq:proj_sumu_2} into \cref{eq:proj_diff} gives
\begin{equation}\label{eq:proj_diff_bound_step}
    \sum_u \norm{\calV X_1[u] \calVdag - \ketbrasame{(-)^{u} }\otimes \1}^2_{\calV \sigma^{\theta}\calVdag} \leq O(\delta').
\end{equation}
Then, applying \cref{lem:op_approx_compts} to \cref{eq:proj_diff_bound_step} gives
\begin{equation}\label{eq:proj_diff_bound}
    \sum_{u,v} \norm{\calV X_1[u] \calVdag - \ketbrasame{(-)^{u} }\otimes \1}^2_{\calV \sigma^{\theta,v}\calVdag} \leq  O(\delta').
\end{equation}

Therefore, we have
\begin{align*}
    & \quad \sum_{u,v} \calV P_{1^{2N}}^u \sigma^{\theta,v} P_{1^{2N}}^u\calVdag \otimes \ketbrasame{u,v} 
    \\
    &= \sum_{u,v} (\calV X_1[u] \calVdag) (\calV \sigma^{\theta,v} \calVdag) (\calV X_1[u] \calVdag) \otimes \ketbrasame{u,v}
    &&\text{(\cref{eq:Pu_prod_Xu} and $\calVdag\calV=\1$)}
    \\
    &\approx_{\delta'} \sum_{u,v} (\ketbrasame{(-)^u} \otimes \1)\calV \sigma^{\theta,v} \calVdag (\ketbrasame{(-)^u} \otimes \1) \otimes \ketbrasame{u,v} &&\amatrix{\text{Part $1$ of \cref{lem:post_meas_approx}}}{\text{with \cref{eq:proj_diff_bound}}}
    \\
    &\approx_{\delta''^2} \sum_{u,v} (\ketbrasame{(-)^u} \otimes \1) \tau^{\theta,v} \otimes \alpha^{\theta,v} (\ketbrasame{(-)^u} \otimes \1) \otimes \ketbrasame{u,v} &&\amatrix{\text{Part $2$ of \cref{lem:post_meas_approx}}}{\text{with \cref{lem:swap_states_nondiamond,lem:swap_states_diamond}}},
\end{align*}
where
\begin{equation}\label{eq:thm_delta_doubleprime}
    \delta''\coloneqq N^{3/2}\gammaH^{1/4}+N^{17/16}\gammaH^{1/32}.
\end{equation}
Note that the above approximations are for operators and we recall that for two operators $A$ and $B$, we defined $A\approx_\epsilon B$ to mean $\normone{A-B}^2\leq O(\epsilon)$, where we stress the square on the Schatten $1$-norm. Therefore, also recalling $\ket{B_{u,1^{2N}}}\coloneqq \ket{(-)^u}$, we can use the triangle inequality to obtain
\begin{equation}\label{eq:delta_primes_bound}
\begin{aligned}
    &\sum_{u,v}\normone{\calV P_{1^{2N}}^u \sigma^{\theta,v}P_{1^{2N}}^u \calVdag - \bra{B_{1^{2N}}^u} \tau^{\theta,v} \ket{B_{1^{2N}}^u} \ketbrasame{B_{1^{2N}}^u} \otimes \alpha^{\theta,v}} 
    \\
    &\leq O(\sqrt{\delta'}+\delta'') \leq O(N^{5/8}\gammaH^{1/16} + N^{3/2}\gammaH^{1/4}+N^{17/16}\gammaH^{1/32}) \leq O(N^{3/2} \gammaH^{1/32}),
\end{aligned}
\end{equation}
where we substituted the definitions of $\delta'$ and $\delta''$, given in \cref{eq:thm_delta_prime,eq:thm_delta_doubleprime} respectively, in the second inequality. This completes the proof of \cref{eq:q_measurement_bound} for $q = 1^{2N}$.

We now explain how the above proof of \cref{eq:q_measurement_bound} for $q = 1^{2N}$ can be slightly modified to prove \cref{eq:q_measurement_bound} for $q \neq 1^{2N}$. When $q\neq 1^{2N}$, the third and second-to-last lines of the equations starting at \cref{eq:measurements_first_line} require slight modification: to commute $Z_{q,i}^{(u_i)}$ past $\sigma^i$ (instead of moving $X_1^{(u_1)}$ past $\sigma^{\theta}$ as done above), we need to use the computational indistinguishability of $\sigma^i$ with $\sigma^j$ for some $j\neq i$. This is due to the conditions in the operator-state commutation lemma, \cref{prop:operator_state_commutation}. In addition, we need to redefine the terms $\delta$ and $\delta'$ appearing in the proof of \cref{eq:q_measurement_bound} for $q = 1^{2N}$ by
\begin{equation}
    \delta \coloneqq (N\sqrt{\gammaH}+\gammaHq + \negl)^{1/2} \quad \text{and} \quad \delta' \coloneqq N (\sqrt{\delta} + \sqrt{\gammaH+\gammaHq} + \negl),
\end{equation}
while $\delta''$ can be left unchanged. These re-definitions are simply due to \cref{lem:swap_observables} and \cref{prop:operator_state_commutation} giving different results for different $Z_{q,i}$ and $X_{q,i}$. With these modifications in place, we can use the first inequality of \cref{eq:delta_primes_bound} to obtain \cref{eq:q_measurement_bound} for $q\neq 1^{2N}$.

At this point, we drop the assumption that $D$ is perfect. In this case, as explained at the start of proof of \cref{eq:soundness_measurements}, we obtain the bound
\begin{equation}
\begin{aligned}\label{eq:pre_expectation_bound}
    &\sum_{u,v}\normone{\calV P_q^u \sigma^{\theta,v}P_q^u \calVdag - \bra{B_q^u} \tau^{\theta,v} \ket{B_q^u} \ketbrasame{B_q^u} \otimes \alpha^{\theta,v}}
    \\
    \leq& \sqrt{\gammaP} + 2\gammaP +  O(N^2\epsilon^{1/32} + N \epsilon_{H,q}^{1/8}) \leq O(\sqrt{N\epsilon} + N^{49/32}\epsilon^{1/32} + N^{5/8} \epsilon_{H,q}^{1/8}) \leq O(N^2\epsilon^{1/32} + N \epsilon_{H,q}^{1/8}),
\end{aligned}
\end{equation}
where we also used \cref{eq:q_measurement_bound} for the first inequality and \cref{prop:gamma_bound_by_fail} for the second inequality.

Therefore, we have
\begin{align*}
    &\quad \mathbb{E}_{q\leftarrow_\mu \{0,1\}^{2N}}\Biggl[\sum_{u,v\in \{0,1\}^{2N}} \normone{\calV P_q^u \sigma^{\theta,v} P_q^u - \bra{B_q^u} \tau^{\theta,v} \ket{B_q^u} \ketbrasame{B_q^u} \otimes \alpha^{\theta,v}}\Biggr]
    \\
    &\leq O\Bigl(\sum_{q\in \{0,1\}^{2N}} \mu(q) \bigl(N^2 \epsilon^{1/32} + N\epsilon_{H,q}^{1/8}\bigr)\Bigr) &&\text{(\cref{eq:pre_expectation_bound})}
    \\
    &\leq O\Bigl(N^2\epsilon^{1/32} + N\Bigl(\sum_{q\in \{0,1\}^{2N}} \mu(q) \epsilon_{H,q} \Bigr)^{1/8}\Bigr)
    &&\text{(Jensen's inequality)}
    \\
    &\leq O(N^2 \epsilon^{1/32}) &&\text{(\cref{def:failure_prob})}.
\end{align*}
Therefore, we have proved the second equation of the theorem, \cref{eq:soundness_measurements}. The last equation of the theorem, \cref{eq:soundness_test_measurements}, follows by applying \cref{def:failure_prob} to \cref{eq:pre_expectation_bound}.
\end{proof}

 \section{Applications}\label{sec:applications}

\subsection{Device-independent quantum key distribution}\label{sec:diqkd}

In this section, we describe how to adapt the protocol for DIQKD under computational assumptions in \cite{metger2021diqkd} to use our self-testing protocol as its main component. The resulting DIQKD protocol operates under the same setting and assumptions as in \cite{metger2021diqkd} except we remove the IID assumption. In particular, we highlight the fact that we retain the advantage of the generated key being information-theoretically secure. For details about the setting and assumptions, we refer the reader to \cite{metger2021diqkd}, in particular, its Figure 1. Our protocol will also make use of a ``cut-and-choose'' technique in \cite[Theorem 4.33 (arXiv version v2)]{gmp2022parallelrsp}). Henceforth, DIQKD without qualification refers to DIQKD under computational assumptions.

Recall that in our self-testing protocol there is a single verifier interacting with a single device. On the other hand, in DIQKD, there are two verifiers, Alice and Bob, that each interact with their own (untrusted) device. In DIQKD \emph{under computational assumptions}, the two devices are not assumed to be non-communicating and are modeled as a single device with two \emph{components}, one on Alice's side, and one on Bob's. At a high level, to resolve the difference in the number of verifiers, we will let Alice play the role of the single verifier in our self-testing protocol while Bob will play a relaying role. 

In \cref{fig:protocol_diqkd_test}, we describe a single test round of our DIQKD protocol. In \cref{fig:protocol_diqkd_gen}, we describe how to modify the test round to give a single generation round of our DIQKD protocol. Note that all communication between Alice and Bob are via their public authenticated channel. Eve, who is computationally unbounded, may compute any secret information hidden in this communication but we make the standard assumption that she cannot send the secret information to the devices (e.g., see \cite[Figure 1]{metger2021diqkd}).

\begin{figure}
    \renewcommand*{\arraystretch}{1.2}
    \vspace{0pt}
    \centering
    \scalebox{0.83}{
    \begin{tabular}{p{17cm}}
    \hline
    \smallskip
\begin{enumerate}
    \item Alice samples $\theta \leftarrow_U \thetaset$ uniformly at randomly, generates $2N$ key-trapdoor pairs $(k_1,t_1)$, $\dots$, $(k_{2N},t_{2N})$ according to $\theta$, and sends $k_{N+1}, \dots, k_{2N}$ to Bob. Note that Alice has all the trapdoors $\{t_i\}_{i=1}^{2N}$. Then Alice sends $k_1,\dots,k_N$ to her component. Bob sends $k_{N+1},\dots, k_{2N}$ to his component.
    \item Alice receives back $(y_1,\dots, y_N) \in \calY^N$ and Bob receives back images $(y_{N+1},\dots, y_{2N})\in \calY^N$.
    \item  Alice samples $c\leftarrow_U\{\text{preimage}, \text{Hadamard}\}$ uniformly at random, sends it to Bob, and they both send $c$ to their components.
    \item[] Case $c=\text{preimage}$. Alice receives  $(b_1,\dots b_N, x_1,\dots, x_N)\in \{0,1\}^{N+Nw}$ from her component and Bob receives $(b_{N+1},\dots, b_{2N}, x_{N+1},\dots, x_{2N})\in \{0,1\}^{N+Nw}$ from his component and sends it to Alice. Alice verifies $(b_1,\dots,b_{2N},x_1,\dots,x_{2N})$ according to our self-testing protocol.
    
    \item[] Case $c=\text{Hadamard}$. 
    \begin{enumerate}
        \item Alice receives $(d_1,\dots, d_N)\in \{0,1\}^{Nw}$ from her component and Bob receives $(d_{N+1},\dots,d_{2N})\in \{0,1\}^{Nw}$ from his component.
        \item Alice samples $a\leftarrow_U\{0,1\}$ uniformly at random. 
        \begin{itemize}[--]
            \item If $a=0$, Alice samples $q\leftarrow_U\qtestset$ uniformly at random. 
            \item If $a=1$, Alice sets $q= 1^N0^N$.
        \end{itemize}
        Note that the resulting distribution on $(q_1,\dots,q_{2N})\in \{0,1\}^{2N}$ is the same as in Step 4 of our self-testing protocol (\cref{fig:protocol}) with $\mu$ chosen as the distribution that always outputs $1^N0^N$.
        
        Alice sends $q_{N+1}$ to Bob. Alice sends $q_1,\dots, q_N$ to her component. Bob sends $q_{N+1},\dots, q_{N+1}$ ($=q_{N+1},\dots, q_{2N}$) to his component. 
        
        \item Alice receives $(u_1, \dots, u_N)\in \{0,1\}^N$ from her component and Bob receives $(u_{N+1}, \dots, u_{2N})\in \{0,1\}^N$ from his component.  Alice sends ``Test'' to Bob. Bob sends $\{(y_i,d_i,u_i)\}_{i=N+1}^{2N}$ to Alice. Alice verifies $\{(y_i,d_i,u_i)\}_{i=1}^{2N}$ according to our self-testing protocol using the trapdoors that she holds, $(t_1,\dots, t_{2N})$.
    \end{enumerate}
    \end{enumerate}
    \\
    \hline
    \end{tabular}
    }
    \caption{Test round for device-independent quantum key distribution (DIQKD) protocol.}
    \label{fig:protocol_diqkd_test}
\end{figure}

\begin{figure}
    \renewcommand*{\arraystretch}{1.2}
    \vspace{0pt}
    \centering
    \scalebox{0.83}{
    \begin{tabular}{p{17cm}}
    \hline
    \smallskip
    Same as the test round (see \cref{fig:protocol_diqkd_test}) except with the following modifications.
    \begin{itemize}[--]
        \item At Step 1, Alice chooses $\theta = \diamond$.
        \item At the start of Step 3, Alice chooses $c=\text{Hadamard}$.
        \item At the start of Step 3(b), instead of sampling $q$, Alice sets $q=1^N0^N$.
        \item Replace Step 3 (c) by the following.  Alice receives $(u_1, \dots, u_N)\in \{0,1\}^N$ from her component and Bob receives $(u_{N+1}, \dots, u_{2N})\in \{0,1\}^N$ from his component. Alice sends ``Generation'' to Bob.
    \end{itemize} 
    \\
    \hline
    \end{tabular}
    }
     \caption{Generation round for device-independent quantum key distribution (DIQKD) protocol.}
     \label{fig:protocol_diqkd_gen}
\end{figure}

We construct our overall DIQKD protocol by using multiple test rounds followed by a single generation round. More specifically, at the start of the protocol, Alice selects $m = \poly_1(N) \in \mathbb{N}$ and $0< \delta = 1/\poly_2(N)$, where the $\poly_i(N)$'s are some polynomials in $N$ that are sufficiently large relative to the robustness bounds in \cref{thm:soundness} (see \cite[Proof of Corollary 4.35 (arXiv version v2)]{gmp2022parallelrsp}).  Then, Alice samples $s\leftarrow_U \{0,\dots, m-1\}$ and performs $sm$ test rounds. Let $B_j\coloneqq\{(j-1)m+1,\dots, jm\}$ for $j=1,\dots, s$. If there exists some $j \in [s]$ such that the fraction of test rounds in $B_j$ that fail is $>\delta$, then Alice aborts the protocol. Otherwise, Alice samples $r\leftarrow [m]$ and performs a further $r-1$ test rounds followed by a generation round. These steps are due to our use of the aforementioned ``cut-and-choose'' technique in \cite{gmp2022parallelrsp}.

After the generation round, Alice and Bob proceed to key extraction, which is essentially the same as that in \cite[Protocol 3]{metger2021diqkd}. More specifically, Alice calculates $\hath_1,\dots, \hath_{N}$ as defined above \cref{eq:hardcore_bitstring}. Then
Alice sets $\tilde{u}_i = u_i \oplus \hath_i$ and Bob sets $\tilde{u}_{N+i} = u_{N+i}$.
Then $(\tilde{u}_i)_{i\in [N]}$ and $(\tilde{u}_{N+i})_{i\in [N]}$ are the raw shared secret keys of Alice and Bob, respectively. Finally, Alice and Bob apply one-way error correction and privacy amplification to these raw keys to obtain their final shared secret keys. This completes our description of the protocol. 

The completeness and soundness of this DIQKD protocol essentially follow from the completeness and soundness of our self-testing protocol combined with the proofs of \cite[Theorem 1]{metger2021diqkd} and \cite[Theorem 4.33 (arXiv version v2)]{gmp2022parallelrsp}. We sketch the salient aspects below.

\paragraph{Completeness (sketch).} Consider the case when the device is honest. By construction, an honest device can pass the test rounds with $\geq 1-\negl$ probability due to the completeness of our self-testing protocol (\cref{thm:completeness}.) It is also clear from \cref{eq:honest_final_state} (in the case $\theta = \diamond$) that $\tilde{u}_i = \tilde{u}_{N+i}$ for all $i\in I$.  The key rate\footnote{In defence of our key rate, we note here that the first protocol for DIQKD in the nonlocal setting without IID assumptions \cite{ruv2013command} also had a key rate that tends to zero as $N\rightarrow \infty$.} of our protocol will be $1/\poly(N)$ because Alice can implement up to $m^2 = \poly_1(N)^2$ test rounds before a generation round and $\poly_1(N)$ is a high-degree polynomial in $N$. 

We mention that our DIQKD protocol as described requires an honest device to use nonlocal controlled-$Z$ operations that act on both of its components. As discussed in \cite{metger2021diqkd}, it is preferable for the honest device to be able to only use only local operations, i.e., operations that act on only one of its components. However, as pointed out in \cite{metger2021diqkd}, it is straightforward to remove the need for the nonlocal controlled-$Z$ operation via gate teleportation. Doing this would require some small changes to the protocol and we refer the reader to \cite[Appendix A]{metger2021diqkd} for details.

\paragraph{Soundness (sketch).}
First, to prove the soundness of this protocol, we need to assume that LWE is hard even when a device has access to quantum advice, see \cite[Remark 4.3 (arXiv version v2)]{gmp2022parallelrsp}. As pointed out in \cite[Remark 4.3 (arXiv version v2)]{gmp2022parallelrsp}, this assumption is used in \cite{gheorghiuvidick2019rsp,brakerski2021cryptographic,gmp2022parallelrsp}. With this assumption, the same argument in \cite[Proof of Theorem 4.33 (arXiv version v2)]{gmp2022parallelrsp} implies that, conditioned on the protocol not aborting, the failure probability of the device in the generation round is of order $\sqrt{\delta}$. 
The physical state after the generation round is 
\begin{equation}
    \sum_{u,v} (P_{1^N0^N}^u)_{A'B'}\sigma_{A'B'E}^{\diamond, v}(P_{1^N0^N}^u)_{A'B'} \otimes \ketbrasame{u}_U,
\end{equation}
where $A'$ is the (quantum) register of Alice's component of the device,  $B'$ is the (quantum) register of Bob's component of the device, $U$ is the register containing Alice and Bob's measurement outcomes, and $E$ contains Eve's quantum side information.
By \cref{thm:soundness}, after tracing out 
systems $A'$ and $B'$, the physical state is
$\eta$-close to the ideal state
\begin{equation}
    \rho = \sum_{v,u\in \{0,1\}^{2N}} \bra{B_{1^N0^N}^u} \tau^{\diamond,v} \ket{B_{1^N0^N}^u} \ketbrasame{u}_{U} \otimes \rho_E^{v},
\end{equation}
where $\eta \coloneqq \sqrt{\delta} + (1 -\sqrt{\delta})N^2 \delta^{1/64} = 1/ \poly(N)$, and
we recall that $\tau^{\diamond,v}$ and $\ket{B_q^u}$ are defined in \cref{eq:tau_def_repeat,eq:bb84_product}. After Alice and Bob changes $u_i$ to $\tilde{u}_i$ for $i\in [N]$ and $i\in \{N+1,\dots, 2N\}$, respectively, the ideal state becomes
\begin{align*}
    \tilde{\rho} = \Bigl(\frac{1}{2^N}\sum_{a\in \{0,1\}^{N}}  \ketbrasame{a,a}_{U} \Bigr) \otimes \Bigl( \sum_{v \in \{0,1\}^{2N}} \rho_E^{v} \Bigr).
\end{align*}
That is, the 
register $U$ becomes independent of Eve's quantum side information, and
\begin{align*}
    H(U\mid E)_{\tilde{\rho}} = N.  
\end{align*}
By the continuity bound of conditional entropy \cite[Lemma 2]{winter2016entropy}, the entropy in the register $U$ conditioned on Eve's side information on the physical state is at least 
\begin{equation}
    N - \eta N - O(\eta \log(1/\eta)) \geq N - 1/\poly(N).
\end{equation}
Hence, our DIQKD protocol is capable of generating $\Omega(N)$ bits of shared key.

\subsection{Dimension test}\label{sec:dimension}
In this section, we present a protocol in \cref{fig:protocol_dimension} for testing the quantum dimension of an efficient device. This protocol is a simplified version of the self-testing protocol in \cref{fig:protocol}.
\begin{figure}[ht]
    \renewcommand*{\arraystretch}{1.2}
    \vspace{0pt}
    \centering
    \scalebox{0.83}{
    \begin{tabular}{p{17cm}}
    \hline
    \smallskip
    1. Set $N = \lambda$. Sample $\theta \leftarrow_U \{0,1,\ldots,N\}$ uniformly at random. Sample $N$ key-trapdoor pairs $(k_1,t_{k_1}),\ldots,(k_{N},t_{k_{N}})$ from an ENTCF according to $\theta$ as follows:
    \begin{enumerate}[leftmargin=50pt]
        \item[$\theta = 0$:] all pairs are sampled from $\Gen_{\calG}(1^{\lambda})$
        \item[$\theta > 0$:] the $\theta$-th key-trapdoor pair is sampled from $\Gen_{\calF}(1^{\lambda})$ and the remaining $N-1$ pairs are all sampled from $\Gen_{\calG}(1^{\lambda})$. 
    \end{enumerate}
    Send the keys $k = (k_1,\ldots,k_N)$ to the device.
    
    \\
    
    2. Receive $y = (y_1,\ldots,y_N)\in \calY^{N}$ from the device.
        
    \\
    
    3. Sample round type ``preimage" or ``Hadamard" uniformly at random and send to the device.
      
    \\
    
       \underline{Case ``preimage''}: receive 
       \begin{equation*}
       (b,x) = (b_1,\ldots,b_{N},x_1,\ldots,x_{N})
       \end{equation*}
       from the device, where $b \in \{0,1\}^{N}$ and $x\in \{0,1\}^{Nw}$.
    
      If $\CHK(k_i,y_i,b_i,x_i)=0$ for all $i\in [N]$, \textbf{accept}, else \textbf{reject}. 
      
    \\
    
        \underline{Case ``Hadamard''}:
        receive 
        \begin{equation*}
            d=(d_1,\ldots,d_{N}) \in \{0,1\}^{Nw}
        \end{equation*}
        from the device.
        
    \\
    
        \hspace{3.15mm} Sample $q\leftarrow_{U} \{0^N,1^N\}$ and send $q$ to the device.
        
    \\
    
        \hspace{3.15mm} Receive $u\in \{0,1\}^{N}$ from the device.
        \begin{enumerate}[\text{Case} A.,leftmargin=55pt]
            \item $\theta = 0$ and
            \begin{enumerate}[leftmargin=5pt]
                \item[$q =0^N$: ] if $\hatb(k_i,y_i)\neq u_i$ for some $i \in [N]$, \textbf{reject}, else \textbf{accept}.
                \item[$q =1^N$: ] \textbf{accept}.
            \end{enumerate}
            \item $\theta \in [N]$ and 
            \begin{enumerate}[leftmargin=5pt]
                \item[$q=0^N$:] if $\hatb(k_i,y_i)\neq u_i$ for some $i \neq \theta$, \textbf{reject}, else \textbf{accept}.
                \item[$q=1^N$:] if $\hath(k_\theta,y_\theta,d_\theta) \neq u_\theta$, \textbf{reject}, else \textbf{accept}.
            \end{enumerate}
        \end{enumerate}
    \\
        \hline
    \end{tabular}
    }
    \caption{A protocol that tests the quantum dimension of a computationally efficient device.}
    \label{fig:protocol_dimension}
\end{figure}

There exists an efficient honest quantum device that is accepted by the dimension test with probability $\geq 1-\negl$. The strategy of this device follows that described in the proof of \cref{thm:completeness}, except it is simpler because there is no need to apply controlled-$\sigmaz$ gates. In particular, when $q = 0^N$, the honest device measures $N$ qubits in the computational basis, and when $q = 1^N$, it measures $N$ qubits in the Hadamard basis. Moreover, there exists an efficient classical verifier, which again follows from the efficient function generation and efficient decoding properties of ENTCFs (\cref{property:efficient_generation,property:efficient_decoding}). We omit a formal completeness proof.

We first show in \cref{thm:approx_qubits} that if the device can pass this protocol with high probability, it must have $N$-approximate qubits following the definitions given in \cite{Vidick_Fsmp_2021}. We reproduce these definitions below for completeness.
\begin{defn}
    The Pauli group on $n$ qubits, denoted by $\calP_n$, is the multiplicative group of order $2\cdot 4^n$ defined by 
    \begin{equation}
        \calP_n \coloneqq \{(-1)^c \sigmax(a) \sigmaz(b) \mid c\in \{0,1\}, a\in \{0,1\}^n, b \in \{0,1\}^n\},
    \end{equation}
    where $\sigmax(a) \coloneqq (\sigmax_1)^{a_1}\cdots (\sigmax_n)^{a_n}$ and $\sigmaz(b) \coloneqq (\sigmaz_1)^{b_1}\cdots (\sigmaz_n)^{b_n}$. 
\end{defn}

\begin{defn}
    \label{def:approx_qubits}
    We say a Hilbert space $\calH$ has $n$ $\delta$-approximate qubits if there are reflections (observables that square to the identity)
    $X_1, \ldots, X_n$ and $Z_1, \ldots, Z_n$ on $\calH$ and a state $\rho \in \calD(\calH)$ such that
    for all $a, b, a', b' \in \{0,1\}^n$,
    \begin{equation}
        \norm{ X(a) Z(b) X(a') Z(b') -
        (-1)^{b \cdot a'} X(a \oplus a') Z(b \oplus b') }_{\rho}^2 \leq \delta,
    \end{equation}
    where $X(a) \coloneqq X_1^{a_1}\cdots X_n^{a_n}$, $Z(b) \coloneqq Z_1^{b_1} \cdots Z_n^{b_n}$, and $\oplus$ denotes bit-wise XOR.
\end{defn}

The motivation behind     \cref{def:approx_qubits} is the next theorem which is a version of \cite[Corollary 10.9]{Vidick_Fsmp_2021} that follows from \cite[Theorem 10.6]{Vidick_Fsmp_2021} (attributed to Gowers and Hatami~\cite{gowershatami2017rep}).
\begin{thm}\label{thm:gowers}
    Suppose a Hilbert space $\calH \cong \C^d$ has $n$ $\delta$-approximate qubits with respect to reflections 
    $X_1, \ldots, X_n$ and $Z_1, \ldots, Z_n$ on $\calH$ and a state $\rho \in \calD(\calH)$.
    Then there exists a $d' \geq d$, an isometry $V: \C^d \to \C^{d'}$, and a representation $g \colon \calP_n \to \mathcal{U}(\C^{d'})$ such that
    \begin{align*}
        \E_{a,b} \norm{ X(a)Z(b) - V^\dagger g( \sigmax(a) \sigmaz(b)) V}_{\rho}^2 \leq O(\delta),
    \end{align*}
    where the expectation is over uniformly random $a \leftarrow\{0,1\}^n$ and $b\leftarrow \{0,1\}^n$.
\end{thm}
Intuitively, the theorem says that the action of any product of the $X_i$s and $Z_i$s on $\rho$ is approximately the same as the action of the corresponding product of the $\sigmax_i$s and $\sigmaz_i$s embedded in $\calH$ via $g$ and $V$.
Since the $\sigmax_i$s and $\sigmaz_i$s specify $n$ qubits in $(\mathbb{C}^2)^{\otimes n}$, the $X_i$s and $Z_i$s specify $n$ $\delta$-approximate qubits in $\calH$.

\begin{thm}\label{thm:approx_qubits}
    Let $D$ be an efficient device with Hilbert space $\calH = \calH_D \otimes \calH_Y \otimes \calH_R$.
    If $D$ passes the dimension test of \cref{fig:protocol_dimension} with probability $\geq 1 - \epsilon$, then
    $\calH$ has $N$ $\delta$-approximate qubits for
    $\delta \coloneqq O(N^2\epsilon^{1/32}).$
\end{thm}
\begin{proof}[Proof sketch]
    It suffices to show that the $X_i$ and $Z_i$
    defined in \cref{def:observables} satisfy the condition of
    \cref{def:approx_qubits} with $\rho$ set to $\sigma^0$. This condition can be proved using techniques in the proof of \cref{thm:soundness}, in particular, those used in \cref{eq:measurements_first_line} to handle products of projectors (which can also handle products of observables).
\end{proof}

While \cref{thm:gowers} suggests that a Hilbert space  $\calH$ having $n$ $\delta$-approximate qubits should have a large quantum dimension, it is not apriori obvious how this can be proved. In the rest of this section, we prove that the soundness guarantee of our self-test (\cref{thm:soundness}) implies that any device passing the dimension test in \cref{fig:protocol_dimension} must have a large quantum dimension. 

The intuition is that when the dimension test is passed with high probability, \cref{thm:soundness} guarantees the existence of a quantum state $\rho^{\star}$
on the quantum part of the device's memory that is close to the maximally mixed state up to some isometry. 
More specifically, $\rho^{\star}$ comes from using \cref{thm:soundness} to force the device to perform a Hadamard basis measurement on $N$ qubits that are in  the computational basis and discarding the measurement results.
Then, the main proposition of this section, \cref{prop:rank}, shows that the guarantee on $\rho^{\star}$ is strong enough for us to lower bound the rank of $\rho^{\star}$, which is also a lower bound on the quantum dimension of the device's memory. We proceed to give a formal proof of soundness and stress that we are making the LWE hardness assumption throughout the rest of this section.

To prove \cref{prop:rank}, we use the vector-operator correspondence mapping, $\vc$, as defined in \cite[Chapter 1.1.2]{watrous2018}. Let $A\in \LH$. Informally, $\vc(A) \in \calH \otimes \calH$ is the column vector formed by concatenating the rows of $A$ vertically. We will use the following properties of $\vc$ without further comment:
\begin{enumerate}
    \item Suppose $A$ is of the form $\sum_i \lambda_i\ketbrasame{v_i}$ with $\lambda_i \geq 0$. Then, $\vc(\sqrt{A}) = \sum_i \sqrt{\lambda_i} \ket{v_i}\overline{\ket{v_i}}$, where the overline denotes element-wise complex conjugation.
    
    \item The $\ell_2$-norm of $\vc(A)$ equals the Frobenius norm of $A$. In symbols, $\norm{\vc(A)} = \norm{A}_F$.
    
    \item If $B,C\in \LH$, then $(B\otimes C) \vc(A) = \vc(B A C^\intercal)$, where $^\intercal$ denotes the transpose. In particular, if $U\in \LH$ is unitary, then $(U\otimes \overline{U}) \vc(A) = \vc(U A U^\dagger)$.
\end{enumerate}
We also need the following technical lemma.
\begin{lemma}[{\cite[Lemma 3.34]{watrous2018}}]
    \label{lem:f_to_1}
    Let $P_1, P_2 \in \PosH$ be positive semi-definite operators. Then,
    \begin{equation}
        \norm{\sqrt{P_1} - \sqrt{P_2}}_F
        \leq \sqrt{\norm{P_1 - P_2}_1}.
    \end{equation}
\end{lemma}

\begin{proposition}
    \label{prop:rank}
     Let $\rho, \alpha \in D(\calH)$ be density operators. If there exists a unitary $U \in \calL(\mathbb{C}^{2^n}\otimes \calH)$ such that
    \begin{equation}\label{eq:prop_rank}
        \normone{U (\ketbrasame{0}^{\otimes n} \otimes \rho) U^\dagger - 2^{-n}\1 \otimes \alpha}\leq \epsilon,
    \end{equation}
    then $\Rank(\rho) \geq (1 - \epsilon)2^n$.
\end{proposition}
\begin{proof}
    Let us write the eigen-decompositions of $\rho$ and $\alpha$ as
    $\rho = \sum_{i=1}^N
    \lambda_i \ketbrasame{v_i}$ and 
    $\alpha = \sum_{k = 1}^{N'} \mu_k \ketbrasame{u_k}$, where $\lambda_i,\mu_k > 0$. Note that $\Rank(\rho) = N$; $\sum_i \lambda_i = 1$ because $\rho$ is normalized; $\mu_k\leq 1$ for all $k\in [N']$ because $\alpha$ is normalized. We have
    \begin{equation}
    \begin{aligned}
        &\vc(\sqrt{\ketbrasame{0}^{\otimes n} \otimes \rho}) = \sum_{i=1}^N \sqrt{\lambda_i}
        \ket{0}^{\otimes n}\ket{v_i} \otimes \ket{0}^{\otimes n} \overline{\ket{v_i}} ,
        \\
        &\vc(\sqrt{2^{-n}\1 \otimes \alpha})
        = \sum_{k=1}^{N'}\sum_{x\in \{0,1\}^n} \sqrt{\mu_k/2^n}\ket{x}\ket{u_k} \otimes \ket{x} \overline{\ket{u_k}}.
    \end{aligned}
    \end{equation}
    Therefore, the condition of the proposition, \cref{eq:prop_rank}, implies that
    \begin{equation}
    \begin{aligned}
        &\quad \bignorm{ U \otimes \overline{U} \Bigl(\sum_{i=1}^N \sqrt{\lambda_i}
        \ket{0}^{\otimes n}\ket{v_i} \otimes \ket{0}^{\otimes n} \overline{\ket{v_i}}\Bigr) - 
        \sum_{k=1}^{N'}\sum_{x\in \{0,1\}^n} \sqrt{\mu_k/2^n}\ket{x}\ket{u_k} \otimes \ket{x} \overline{\ket{u_k}} } 
        \\
        & = \norm{ U [ \ketbrasame{0}^{\otimes n} \otimes \sqrt{\rho} \, ] U^\dagger - 
        2^{-n/2} \1 \otimes \sqrt{\alpha}}_F  \leq  \norm{ U [\ketbrasame{0}^{\otimes n} \otimes \rho \, ] U^\dagger - 2^{-n}\1 \otimes \alpha}_1^{1/2} \leq \sqrt{\epsilon}.
    \end{aligned}
    \end{equation}
    
    The key observation of this proof is the following. Let $\ket{\alpha}$, $\ket{\beta}$ be bipartite states on $\calH_A' \otimes \calH_B'$ such that $\ket{\alpha}$ is normalized and has Schmidt rank $R$, and the Schmidt coefficients of $\ket{\beta}$ are each at most $b$. Then,
    \begin{equation}
        \abs{\BraKet{\alpha}{\beta}}^2\leq Rb^2.
    \end{equation}
    To see this, write $\ket{\alpha} = \sum_{i=1}^R a_i \ket{e_i} \otimes \ket{e_i'}$ and $\ket{\beta} = \sum_{j=1}^d b_j \ket{f_j} \otimes \ket{f_j'}$ in terms of their Schmidt decompositions, where $d \coloneqq \min\{\dim \calH,\dim \calH'\}$. Then, by using the Cauchy-Schwarz inequality twice,
    \begin{equation}
    \begin{aligned}
        \abs{\BraKet{\alpha}{\beta}}  &\leq b \sum_{i=1}^R a_i \sum_{j=1}^d \abs{\BraKet{e_i}{f_j}}\abs{\BraKet{e_i'}{f_j'}} 
        \\
        &\leq b \Bigl(\sum_{i=1}^R a_i^2\Bigr)^{1/2} \Bigl(\sum_{i=1}^R \sum_{j=1}^d \abs{\BraKet{e_i}{f_j}}\abs{\BraKet{e_i'}{f_j'}}\Bigr)^{1/2}
        \\
        &\leq b \Bigl(\sum_{i=1}^R \Bigl(\sum_{j=1}^d \abs{\BraKet{e_i}{f_j}}^2\Bigr)^{1/2}\Bigl(\sum_{k=1}^d \abs{\BraKet{e_i'}{f_k'}}^2\Bigr)^{1/2}\Bigr)^{1/2} = b \sqrt{R}.
    \end{aligned}
    \end{equation}
    
    To conclude the proof, we use the above observation as follows. Set 
    \begin{equation}
    \ket{\alpha} \coloneqq U \otimes \overline{U} \Bigl(\sum_{i=1}^N \sqrt{\lambda_i}
    \ket{0}^{\otimes n}\ket{v_i}  \otimes \ket{0}^{\otimes n}\overline{\ket{v_i}}\Bigr)
    \quad \text{and} \quad
    \ket{\beta} \coloneqq \sum_{k=1}^{N'}\sum_{x\in \{0,1\}^n} \sqrt{\mu_k/2^n}\ket{x}\ket{u_k} \otimes \ket{x} \overline{\ket{u_k}},
    \end{equation}
    then,
    \begin{equation}
        \abs{\BraKet{\alpha}{\beta}} \geq \real \, \BraKet{\alpha}{\beta} =
          \frac{1}{2} (\norm{\ket{\alpha}}^2 + \norm{\ket{\beta}}^2 - \norm{\ket{\alpha} - \ket{\beta}}^2)
        \geq  1-\epsilon/2.
    \end{equation}
    Now, $\ket{\alpha}$ is normalized (because $\sum_i \lambda_i = 1$) and has Schmidt rank $N$, and the Schmidt coefficients of $\ket{\beta}$ are each at most $\sqrt{\mu_k/2^n} \leq 2^{-n/2}$. Therefore, by the above observation,
    \begin{equation}
        N \geq (1-\epsilon/2)^2/(2^{-n/2})^2 \geq (1-\epsilon) \, 2^n,
    \end{equation}
    which completes the proof.
\end{proof}

\begin{remark}
Some ideas behind our proof of \cref{prop:rank}
come from \cite[Proof of Theorem 8.3]{ji2020mipre}.
\end{remark}

We now use \cref{prop:rank} to prove the main theorem of this section. Much of the proof is devoted to bookkeeping to ensure that the (normalized) density operator condition in \cref{prop:rank} is satisfied and that we are bounding the \emph{quantum} dimension. 
\begin{thm}
\label{cor:dimension}
Let $D$ be an efficient device with Hilbert space $\calH = \calH_D \otimes \calH_Y \otimes \calH_R$. Let the classical-quantum decomposition of $\calH$ be $\calH_C \otimes \calH_Q$, so that all states and observables of $D$ on $\calH$ are classical on $\calH_C$, i.e., block-diagonal in a fixed basis $\{\ket{c}\mid c \in [\dim(\calH_C)]\}$ of $\calH_C$. If $D$ can pass the dimension test protocol of \cref{fig:protocol_dimension} with probability $\geq 1 - \epsilon$, then the quantum dimension of $D$, $\dim(H_Q)$, is at least $(1 - O(N^2\epsilon^{1/32}))2^N$.
\end{thm}

\begin{proof}
It suffices to assume $\epsilon = O(1/N^{64})$, else the bound on $\dim(H_Q)$ holds vacuously as $N\rightarrow \infty$. In this proof, we use $\poly(N,\epsilon)$ to mean a polynomial of order $O(N^2\epsilon^{1/32})$ for convenience.

We model $D$ in a way that is analogous to \cref{sec:devices}, which defines states $\sigma^\theta\in \calH_D \otimes \calH_Y \otimes \calH_R$ and projective measurements $P_{0^N}$ and $P_{1^N}$ corresponding to the device's second answer when asked questions $q=0^N,1^N$ respectively. By following the arguments in \cref{sec:soundness}, we can prove the following analogue of \cref{eq:soundness_test_measurements} in \cref{thm:soundness}:
\begin{equation}\label{eq:dim_vu}
    \sum_{v \in \{0,1\}^N} \sum_{u \in \{0,1\}^N} \normone{\calV P_{1^N}^u \sigma^{0,v} P_{1^N}^u\calVdag - 2^{-N} \ketbrasame{-^u} \otimes \alpha^{0,v}} \leq \poly(N, \epsilon),
\end{equation}
where $\ket{-^u}\coloneqq H^{\otimes N} \ketbrasame{u} H^{\otimes N}$. 

In the following, for $v\in \{0,1\}^{N}$, we write $\sigma^v \coloneqq \sigma^{0,v}$ and $\alpha^v \coloneqq \alpha^{0,v}$ for convenience. 
Let $S\coloneqq \{v \in \{0,1\}^N \mid \tr[\sigma^v] > 0\}$. For $v \in S$, we write $\hat{\sigma}^{v} \coloneqq \sigma^{v}/\tr[\sigma^{v}]$ and $\tilde{\alpha}^v \coloneqq \alpha^{v}/\tr[\sigma^{v}]$. Then,
\begin{equation}
    \sum_{v \in S} \tr[\sigma^v] \bignorm{\calV \Bigl(\sum_{u \in \{0,1\}^N} P_{1^N}^u \hat{\sigma}^v P_{1^N}^u\Bigr)\calVdag
    - 2^{-N}\sum_{u \in \{0,1\}^N} \ketbrasame{-^u} \otimes \tilde{\alpha}^v}_1 \leq \poly(N, \epsilon).
\end{equation}
Note that $\sum_{u \in \{0,1\}^N} \ketbrasame{-^u} = \1$. 

Let $v_{\min}$ be the $v\in S$ that minimizes $\normone{\calV (\sum_{u \in \{0,1\}^N} P_{1^N}^u \hat{\sigma}^v P_{1^N}^u)\calVdag
- 2^{-N}\1 \otimes \tilde{\alpha}^v}$. Then, \cref{lem:sum_sigma_v} gives
\begin{equation}
    \bignorm{\calV \Bigl(\sum_{u \in \{0,1\}^N} P_{1^N}^u \hat{\sigma}^{v_{\min}} P_{1^N}^u\Bigr)\calVdag
    - 2^{-N} \1 \otimes \tilde{\alpha}^{v_{\min}}}_1 \leq \poly(N,\epsilon)/(1-\gammaP) \leq \poly(N,\epsilon),
\end{equation}
where the last inequality holds by \cref{prop:gamma_bound_by_fail} and $\epsilon = O(1/N)$.

Note that $\tilde{\alpha}^{v_{\min}}$ is not necessarily normalized, but we can show $\tr[\tilde{\alpha}^{v_{\min}}]$ is close to $1$
as follows.
To simplify notation, let $\sigma \coloneqq \hat{\sigma}^{v_{\min}}$ and $\tilde{\alpha} \coloneqq \tilde{\alpha}^{v_{\min}}$. Then,
\begin{align*}
    \abs{1- \tr[\tilde{\alpha}]} = &\Bigl| \bignorm{\calV \Bigl(\sum_{u \in \{0,1\}^N} P_{1^N}^u \sigma P_{1^N}^u\Bigr)\calVdag}_1 - \normone{2^{-N}\1\otimes \tildealpha} \Bigr| \\
    \leq &\bignorm{\calV \Bigl(\sum_{u \in \{0,1\}^N} P_{1^N}^u \sigma P_{1^N}^u \Bigr)\calVdag - 2^{-N} \1 \otimes \tildealpha}_1 \leq \poly(N,\epsilon).
\end{align*}
Now, let $\alpha \coloneqq \tilde{\alpha} / \tr[\tilde{\alpha}]$. Note that $\tr[\alpha] = \tr[\sigma] = 1$.
Then,
\begin{equation}
\begin{aligned}
    &\quad \bignorm{\calV \Bigl(\sum_{u \in \{0,1\}^N} P_{1^N}^u \sigma P_{1^N}^u\Bigr)\calVdag - 2^{-N}\1 \otimes \alpha}_1
    \\
    &\leq  \bignorm{\calV \Bigl( \sum_{u \in \{0,1\}^N} P_{1^N}^u \sigma P_{1^N}^u \Bigr)\calVdag - 2^{-N} \1 \otimes \tr[\tilde{\alpha}]\alpha}_1 + \normone{2^{-N}\1 \otimes \tr[\tilde{\alpha}]\alpha - 2^{-N}  \1 \otimes \alpha} \leq \poly(N,\epsilon).
\end{aligned}
\end{equation}

Let $\rho \coloneqq \sum_{u \in \{0,1\}^N} P_{1^N}^u \sigma P_{1^N}^u$. From the definition of $\calV$ in \cref{fig:swap_isometry}, we see
that there exists a unitary $U\in \calL(\mathbb{C}^{2^n} \otimes \calH)$ of the form 
\begin{equation}
    U = \sum_{c = 1}^{\dim \calH_C} \ketbrasame{c} \otimes U_c,  
\end{equation}
where each $U_c \in \calL(\mathbb{C}^{2^n} \otimes \calH_Q)$ is unitary, such that
\begin{align}
    \label{eq:rho_iden}
    \normone{U (\ketbrasame{0}^{\otimes N} \otimes \rho) U^\dagger
    - 2^{-N}\1 \otimes \alpha} \leq \poly(N,\epsilon).
\end{align}
Here, we used the fact that the controlled-$Z_i$ and controlled-$X_i$ gates appearing in $\calV$ are both block-diagonal in the $\{\ketbrasame{c}\}_c$ basis because $Z_i$ and $X_i$ are observables of $D$ on $\calH$. 

We can also write $\rho$ and $\alpha$ as
\begin{equation}\label{eq:cq_decomposition}
    \rho = \sum_{c = 1}^{\dim \calH_C} \ketbrasame{c} \otimes \rho_c
    \quad \text{and} \quad
    \alpha = \sum_{c = 1}^{\dim \calH_C} \ketbrasame{c} \otimes \alpha_c,
\end{equation}
where $\rho_c,\alpha_c \in \Pos(\calH_Q)$ are such that $\sum_c\tr[\rho_c] = \sum_c\tr[\alpha_c] = 1$.

Then, \cref{eq:rho_iden} implies
\begin{equation}\label{eq:dim_c}
    \sum_{c=1}^{\dim \calH_C} 
    \normone{U_c (\ketbrasame{0}^{\otimes N} \otimes \rho_c) U_c^\dagger
    - 2^{-N}\1 \otimes \alpha_c} \leq \poly(N,\epsilon).
\end{equation}
Analogously to how we handled the sum over $v \in \{0,1\}^N$ in \cref{eq:dim_vu}, we can use \cref{eq:dim_c} to show that
there exists some $c^\star \in [\dim \calH_C]$ and (normalized) density operators $\rho^\star$ and $\alpha^\star$ on $\calH_Q$ such that 
\begin{equation}
    \normone{U_{c^\star} (\ketbrasame{0}^{\otimes N} \otimes \rho^\star) U_{c^*}^\dagger
    - 2^{-N}\1 \otimes \alpha^\star} \leq \poly(N,\epsilon).
\end{equation}

Finally, we apply \cref{prop:rank} to obtain
\begin{equation}
    \dim \calH_Q \geq \Rank(\rho^\star) \geq (1 - \poly(N,\epsilon))2^N,
\end{equation}
which completes the proof.
\end{proof}

\begin{remark}
    We allow the prover to prepare $\rho$ by a procedure that includes measurements and resetting qubits. These may introduce extra classical dimensions to the system but not extra quantum dimensions. Since in our proof, we factor out all the possible classical dimensions (see~\cref{eq:cq_decomposition}) the prover's measurements and resetting of qubits do not affect our lower bound on the \emph{quantum} dimension.
\end{remark}

\bibliographystyle{alphaurl}
\bibliography{references_testing}
\end{document}